%% file: main.tex
\documentclass[11pt]{article}

\usepackage[usenames,dvipsnames]{xcolor}
\definecolor{Gred}{RGB}{219, 50, 54}
\definecolor{Ggreen}{RGB}{60, 186, 84}
\definecolor{Gblue}{RGB}{72, 133, 237}
\definecolor{Gyellow}{RGB}{247, 178, 16}
\definecolor{ToCgreen}{RGB}{0, 128, 0}
\definecolor{myGold}{RGB}{231,141,20}
\definecolor{myBlue}{rgb}{0.19,0.41,.65}
\definecolor{myPurple}{RGB}{175,0,124}

\usepackage{hyperref}
\hypersetup{
			colorlinks=true,
	citecolor=ToCgreen,
	linkcolor=Sepia,
	filecolor=Gred,
	urlcolor=Gred
	}

\usepackage{pset}
\usepackage{subcaption}
\usepackage{changes}
\usepackage{wrapfig}
\usepackage{thm-restate}
\usepackage{cite}
\definechangesauthor[name={Sitan}, color=red]{sitan}
\definechangesauthor[name={Ankur}, color=blue]{ankur}
\colorlet{Changes@Color}{purple}
\usepackage[font=footnotesize]{caption}
\usepackage{changepage}
\usepackage[font=itshape]{quoting}

\let\C\undefined
\usepackage{constants}
\usepackage{mathtools}
\mathtoolsset{showonlyrefs}

\DeclarePairedDelimiter{\abs}{\lvert}{\rvert}
\DeclarePairedDelimiter{\iprod}{\langle}{\rangle}
\DeclarePairedDelimiter{\brk}{[}{]}
\DeclarePairedDelimiter{\brc}{\{}{\}}

\makeatletter
\def\Var{\@ifnextchar[{\@withc}{\@withoutc}}
\def\@withc[#1]{\mathop{\mathbb{V}}_{#1}\brk}
\def\@withoutc{\mathop{\mathbb{V}}\brk}
\makeatother

\makeatletter
\def\bone{\@ifnextchar[{\@withd}{\@withoutd}}
\def\@withd[#1]{\mathop{\mathds{1}}_{#1}\brk}
\def\@withoutd{\mathop{\mathds{1}}\brk}
\makeatother

\makeatletter
\def\psE{\@ifnextchar[{\@withe}{\@withoute}}
\def\@withe[#1]{\td{\mathop{\mathbb{E}}}_{#1}\brk}
\def\@withoute{\td{\mathop{\mathbb{E}}}\brk}
\makeatother

\DeclarePairedDelimiterX{\expectarg}[1]{[}{]}{%
  \ifnum\currentgrouptype=16 \else\begingroup\fi
  \activatebar#1
  \ifnum\currentgrouptype=16 \else\endgroup\fi
}

\newcommand{\innermid}{\nonscript\;\delimsize\vert\nonscript\;}
\newcommand{\activatebar}{%
  \begingroup\lccode`\~=`\|
  \lowercase{\endgroup\let~}\innermid 
  \mathcode`|=\string"8000
}

\newconstantfamily{c}{symbol=c}
\newconstantfamily{N}{symbol=N}
\newcommand{\diag}{\text{diag}}
\renewcommand{\Re}{\text{Re}}
\renewcommand{\co}{\mathbb{C}}
\newcommand{\vx}{\vec{x}}

\newcommand{\vmu}{\vec{\mu}}

\newcommand{\NA}{\text{NA}}

\newcommand{\calD}{\mathcal{D}}
\newcommand{\calJ}{\mathcal{J}}
\renewcommand{\d}{{\mathrm{d}}}

\newcommand{\Poi}{\text{Poi}}

\renewcommand{\E}{\mathbb{E}}

\newcommand{\radius}{\mathcal{R}}

\renewcommand{\O}{\mathcal{O}}
\newcommand{\lbound}{\underline{\gamma}}
\newcommand{\factor}{\overline{\gamma}}
\newcommand{\Id}{\text{Id}}

\newcommand{\sgn}{\mathop{\text{sgn}}}

\title{Algorithmic Foundations for the Diffraction Limit}
\author{Sitan Chen\thanks{This work was supported in part by a Paul and Daisy Soros Fellowship, NSF CAREER Award CCF-1453261, and NSF Large CCF-1565235 and was done in part while S.C. was an intern at Microsoft Research AI.} \\MIT  \and Ankur Moitra\thanks{This work was supported in part by NSF CAREER Award CCF-1453261, NSF Large CCF-1565235, a David and Lucile Packard Fellowship, and an Alfred P. Sloan Fellowship.}\\ MIT}

\newtheorem{Alg}{Algorithm}

\algdef{SE}[SUBALG]{Indent}{EndIndent}{}{\algorithmicend\ }%
\algtext*{Indent}
\algtext*{EndIndent}

\usepackage{graphicx}

\begin{document}

\maketitle

\begin{abstract}
For more than a century and a half it has been widely-believed (but was never rigorously shown) that the physics of diffraction imposes certain fundamental limits on the resolution of an optical system. However our understanding of what exactly can and cannot be resolved has never risen above heuristic arguments which, even worse, appear contradictory. In this work we remedy this gap by studying the diffraction limit as a statistical inverse problem and, based on connections to provable algorithms for learning mixture models, we rigorously prove upper and lower bounds on the statistical and algorithmic complexity needed to resolve closely spaced point sources. In particular we show that there is a phase transition where the sample complexity goes from polynomial to exponential. Surprisingly, we show that this does {\em not} occur at the Abbe limit, which has long been presumed to be the true diffraction limit.



\end{abstract}


\input{intro_arxiv}

\input{lowerboundpreview}

\input{prelims}

\input{algo}

\input{stoc_lowerbound}

\section{Conclusion and Open Problems}
\label{sec:conclusion}

We hope that our work will be a stepping-stone towards developing a rigorous theory of resolution limits in more sophisticated optical systems. The setting that we study, namely diffraction through a perfectly circular aperture under incoherent illumination, is arguably the most basic model one can study in Fourier optics. As a natural next step, one can ask whether the techniques developed in this work can be pushed to answer questions about the following more challenging setting:

\paragraph{Coherent illumination}

As described in Appendix~\ref{sec:physical}, in the presence of light emanating from a single point source, the (complex-valud) amplitude of the electric field at a point $P$ on the observation plane is proportional to $e^{i\omega}\cdot J_1(z/\sigma)/(z/\sigma)$, where $e^{i\omega}$ is some phase factor, $z$ is the angular displacement of the point $P$ from the optical axis, and $\sigma$ is the spread parameter which depends on the wavelength of the light and the radius of the aperture. This means that the actual probability distribution over where on the observation plane a photon gets detected is proportional to the squared modulus of this, i.e. $J_1(z/\sigma)^2/(z/\sigma)^2$. Throughout this work, we worked under the assumption that in the presence of many point sources, the light emanating from the various point sources is \emph{incoherent}. In other words, there is no interference introduced by the extra phase factors, and mathematically this translates to a probability distribution given by a \emph{nonnegative} linear combination of the probability densities coming from the individual sources of light, and this is what gives rise to the mixture model we studied.

The \emph{coherent} setting is quite different. Suppose that for point source $j$, the extra phase factor in the electric field at any given point is $e^{i\omega_j}$ for some complex number $\omega_j$. Under coherent illumination from multiple point sources, it is the \emph{electric field} which is a linear combination, namely of the electric fields associated to each indivdiual point source. This gives rise to the following natural probabilistic model:

\begin{defn}[Coherent superpositions of Airy disks]
	A \emph{coherent superposition of $k$ Airy disks $\rho$} is a distribution over $\R^2$ specified by phases $\omega_1,\ldots,\omega_k\in\C$, relative intensities $\lambda_1,\ldots,\lambda_k\ge 0$ summing to 1, centers $\vmu_1,\ldots,\vmu_k\in\R^2$, and an \emph{a priori} known ``spread parameter'' $\sigma > 0$. Its density at $x$ is proportional to \begin{equation}
		\rho(\vx) \propto \left|\sum^k_{i=1}\lambda_i\cdot  e^{i\omega_j}\cdot \frac{J_1(\norm{\vx - \vmu_i}/\sigma)}{\norm{\vx - \vmu_i}/\sigma}\right|^2.
	\end{equation}
\end{defn}

\noindent One can ask analogues of all of the questions considered in the present work for this probabilistic model. This seems to be both mathematically natural and a physically well-motivated setting which now departs from the mixture model setup usually studied within theoretical computer science. 

\paragraph{Acknowledgments}

We would like to thank Elchanan Mossel and Tim Roughgarden for helpful feedback on earlier versions of this work.

\bibliographystyle{alpha}
\bibliography{biblio}

\appendix

\input{science}

\input{model}

\input{quotes}

\input{perturbation}

\input{fig_gen}




\end{document}

%% file: intro_arxiv.tex
\section{Introduction}

For more than a century and a half it has been widely believed (but was never rigorously shown) that the physics of diffraction imposes certain fundamental limits on the resolution of an optical system. In the standard physical setup, we observe incoherent illumination from far-away point sources through a perfectly circular aperture (see Figure~\ref{fig:fraunhofer}). Each point source produces a two-dimensional image, computed explicitly by Sir George Biddell Airy in 1835\cite{airy1835diffraction} and now called an {\em Airy disk}. For a point source of light whose angular displacement from the optical axis is $\vmu\in\R^2$, the normalized intensity at a point $\vec{x}$ on the observation plane is given by \begin{equation}I(\vec{x}) = \frac{1}{\pi \sigma^2} \left(\frac{2J_1(\norm{\vec{x} - \vmu}_2/\sigma)}{\norm{\vec{x} - \vmu}_2/\sigma}\right)^2\end{equation} where $J_1$ is a Bessel function of the first kind. Under Feynman's path integral formalism, $I(x)$ is precisely the pdf of the distribution over where the photon is detected (see Appendix~\ref{sec:science}). The physical properties of the optical system, namely its numerical aperture and the wavelength of light being observed, determine $\sigma$ which governs the amount by which each point source gets blurred.

\begin{figure}
	\centering
	\includegraphics[width=0.55\textwidth]{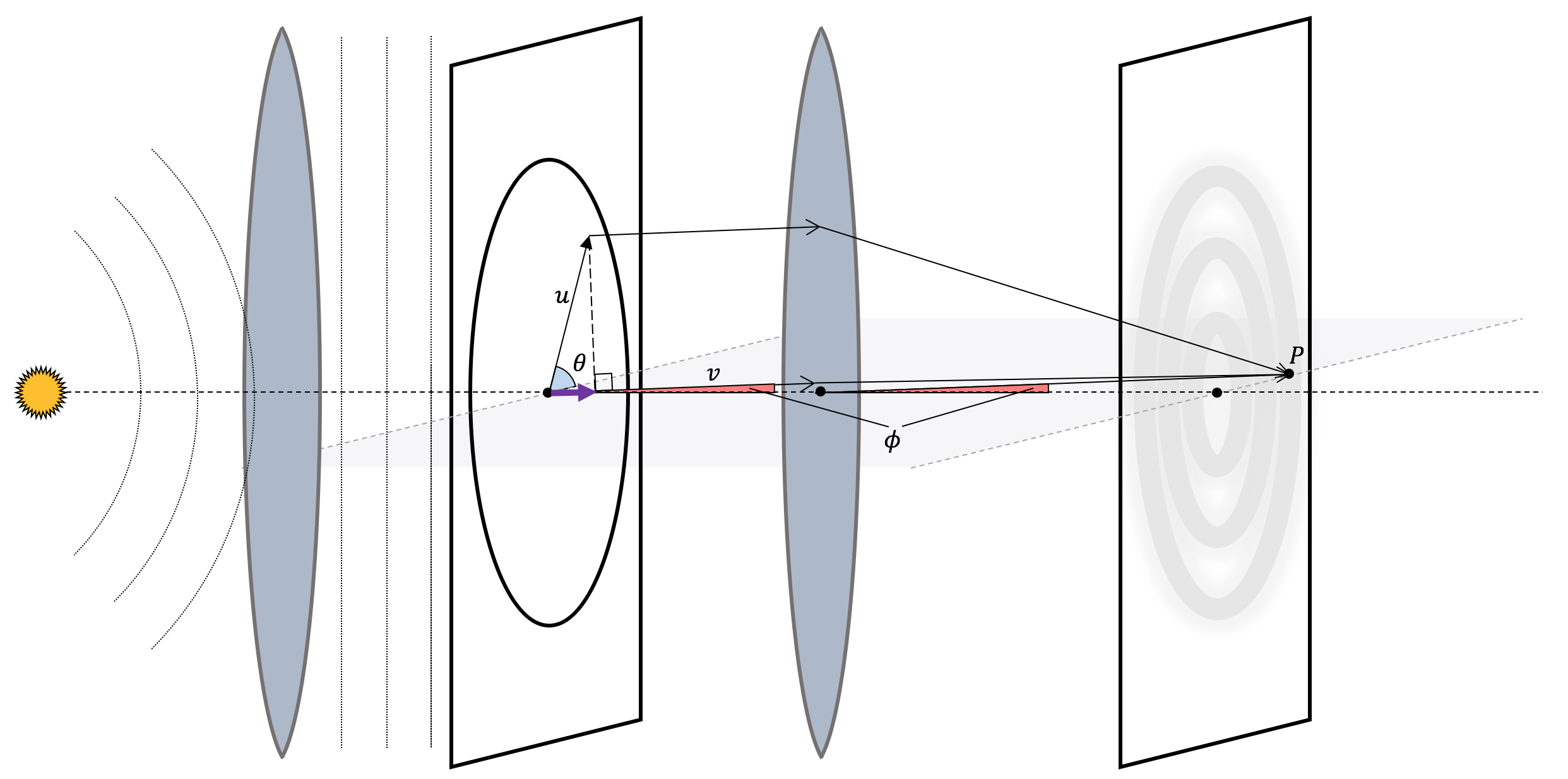}
	\caption{Fraunhofer diffraction of incoherent illumination from point source through aperture onto observation plane}
	\label{fig:fraunhofer}
\end{figure}

Intuitively, when point sources are closer together it seems harder to resolve them. However, despite considerable interest over the years \cite{abbe1873beitrage,rayleigh1879xxxi,schuster1904introduction,sparrow1916spectroscopic,houston1927compound,buxton1937xli}, our understanding of what exactly can and cannot be resolved has never risen above heuristic arguments. In 1879 Lord Rayleigh \cite{rayleigh1879xxxi} proposed a criterion for assessing the resolving power of an optical system, which is still widely-used today, of which he wrote:

\begin{quoting}
	\noindent ``This rule is convenient on account of its simplicity and it is sufficiently accurate in view of the necessary uncertainty as to what exactly is meant by resolution.''
\end{quoting}

\noindent Over the years, many researchers have proposed alternative criteria and offered arguments about why some are more appropriate than others. For example, in 1916 Carroll Sparrow proposed a new criterion \cite{sparrow1916spectroscopic} that bears his name, which he justified as follows:

\begin{quoting}
	\noindent ``It is obvious that the undulation condition should set an upper limit to the resolving power $\dots$ The effect is observable both in positives and in negatives, as well as by direct vision $\dots$ My own observations on this point have been checked by a number of my friends and colleagues.''
\end{quoting}

\noindent Even more resolution criteria were proposed, both before and after, by Ernst Abbe \cite{abbe1873beitrage}, Sir Arthur Schuster \cite{schuster1904introduction}, William Houston \cite{houston1927compound}, etc. Their popularity varies depending on the application area and research community. Many researchers have also pushed back on the idea that there is a fundamental diffraction limit at all. In his 1964 Lectures on Physics \cite[Section 30-4]{feynman2011feynman}, Richard Feynman writes:

\begin{quoting}
	\noindent ``$\dots$ it seems a little pedantic to put such precision into the resolving power formula. This is because Rayleigh's criterion is a rough idea in the first place. It tells you where it begins to get very hard to tell whether the image was made by one or by two stars. Actually, if sufficiently careful measurements of the exact intensity distribution over the diffracted image spot can be made, the fact that two sources make the spot can be proved even if $\theta$ is less than $\lambda/L$.''
\end{quoting}

\noindent Or as Toraldo di Francia \cite{di1955resolving} puts it:

\begin{quoting}
\noindent	``Mathematics cannot set any lower limit for the distance of two resolvable points.''
\end{quoting}

Our goal in this work is to remedy this gap in the literature and place the notion of the diffraction limit on rigorous statistical foundations by drawing new connections to recent work in theoretical computer science on provably learning mixture models, as we will describe next. First we remark that the way the diffraction limit is traditionally studied is in fact a mixture model. In particular we assume that, experimentally, we can measure photons that are sampled from the true diffracted image. However we only observe a finite number of them because our experiment has finite exposure time, and indeed as we will see in some settings the number of samples needed to resolve closely-spaced objects can explode and be essentially impossible just from statistical considerations. Moreover we may only be able to record the location of observed photons up to some finite accuracy, which can also be thought of as being related to sampling error. The main question we will be interested in is:

\begin{quote}
\emph{How many samples (i.e. photons) are needed to accurately estimate the centers and relative intensities of a mixture (i.e. superposition) of two or more Airy disks, as a function of their separation and the parameters of the optical system?}
\end{quote}


This is a central question in optics. Fortunately, there are many parallels between this question and the problem of provably learning mixture models that surprisingly seem to have gone undiscovered. In particular, let us revisit Sparrow's argument that resolution is impossible when the density function becomes unimodal. In fact there are already counter-examples to this claim, albeit not for mixtures of Airy disks. It is known that there are algorithms for learning mixtures of two Gaussians that take a polynomial number of samples and run in polynomial time. These algorithms work even when the density function is unimodal, and require just that the overlap between the components can be bounded away from one. Moreover when there are $k$ components it is known that there is a critical separation above which it is possible to learn the parameters accurately with a polynomial number of samples, and below which accurate learning requires a superpolynomial number of samples information-theoretically \cite{regev2017learning}. Thus a natural way to formulate what the diffraction limit is, so that it can be studied rigorously, is to ask:

\begin{quote}
\emph{At what critical separation does the sample complexity of learning mixtures of $k$ Airy disks go from polynomial to exponential?}
\end{quote}

In this work we will give algorithms whose running time and sample complexity are polynomial in $k$ above some critical separation, and prove that below some other critical separation the sample complexity is necessarily exponential in $k$. These bounds will be within a universal constant, and thus we approximately locate the true diffraction limit. There will also be some surprises along the way, such as the fact that the {\em Abbe limit}, which has long been postulated to be the true diffraction limit, is not actually the correct answer! 

Before we proceed, we also want to emphasize that there is an important conceptual message in our work. First, for mixtures of Gaussians the model was only ever supposed to be an {\em approximation} to the true data generating process. For example, Karl Pearson introduced mixtures of Gaussians in order to model various physical measurements of the Naples crabs. However mixtures of Gaussians always have some chance of producing samples with negative values, but Naples crabs certainly do not have negative forehead lengths! In contrast, for mixtures of Airy disks the model is an extremely accurate approximation to the observations in many experimental setups {\em because it comes from first principles.} It is particularly accurate in astronomy where for all intents and purposes the lens is spherical and the star is so far away that it is a point source, and the question itself is highly relevant because it arises when we want to locate double-stars \cite{falconi1967limits} . 

Furthermore we believe that there ought to be many more examples of inverse problems in science and engineering where tools and ideas from the literature on provably learning mixture models ought to be useful. Indeed both mixtures of Gaussians and mixtures of Airy disks can be thought of as inverse problems with respect to simple differential equations, for the heat equation and a modified Bessel equation respectively. While this is a well-studied topic in applied mathematics, usually one makes some sort of smoothness assumption on the initial data. What is crucial to both the literature on learning mixtures of Gaussians and our work is that we have a parametric assumption that there are few components. Thus we ask: Are there provable algorithms for other inverse problems, coming from differential equations, under parametric assumptions? Even better: Could techniques inspired by the method of moments play a key role in such a development?

\subsection{Overview of Results}
\label{subsec:overview}

It is often the case that heuristic arguments, despite being quite far from a rigorous proof, predict the correct thresholds for a wide range of statistical problems. However here there will be a surprise. In a seminal work in 1873, Ernest Abbe formulated what is now called the {\em Abble limit}. Since then it has been widely accepted in the optics literature as the critical distance below which diffraction makes resolution impossible for classical optical systems. In the mixture model formalism outlined above, it corresponds to a separation of $\pi\sigma$ between any pair of Airy disk centers $\vmu_i,\vmu_j$. This distance arises naturally because it corresponds to the radius of the support of the Fourier transform of the Airy disk kernel $A_{\sigma}: \vec{x}\mapsto \frac{1}{\pi \sigma^2}\left(\frac{J_1(\vec{x}/\sigma)}{\vec{x}/\sigma}\right)^2$ (see Appendix~\ref{subsec:menagerie} for further discussion). 


One of the main results of this work is to show that resolution is statistically hard even above the Abbe limit! Specifically, we show that even for mixtures of Airy disks whose centers have a pairwise separation that is a constant factor larger than the Abbe limit, the problem of recovering their locations can require $\exp(\Omega(\sqrt{k}))$ samples. The main challenge is that no configuration where the Airy disk centers are all on the same line can beat the Abbe limit. Instead we construct a new, natural lower bound instance. 

\begin{thm}[Informal, see Theorem~\ref{thm:mainlowerbound}]
	Let $\lbound \triangleq \sqrt{4/3}\approx 1.155$. For any $0<\epsilon<1$, there exist two superpositions of $k$ Airy disks $\rho,\rho'$ which are both $\lbound\cdot (1-\epsilon)\cdot \pi\sigma$-separated and such that 1) $\rho$ and $\rho'$ have noticably different sets of centers, and yet 2) it would take at least $\exp(\Omega(\epsilon \sqrt{k}))$ samples to distinguish whether the samples came from $\rho$ or from $\rho'$.\label{thm:lowerbound_informal}
\end{thm}

On the other hand, we also show that when the Airy disks have separation that is a small constant factor larger than this critical distance, there is an algorithm for recovering the centers that takes a polynomial number of samples and runs in polynomial time. 

\begin{thm}[Informal, see Theorem~\ref{thm:above}]
Define the absolute constant $\factor = \frac{2j_{0,1}}{\pi} = 1.530\ldots$, where $j_{0,1}$ is the first positive zero of the Bessel function $J_0$. Let $\rho$ be a $\factor\cdot \pi\sigma$-separated superposition of $k$ Airy disks where every disk has relative intensity at least $\lambda$. Then for any target error $\epsilon > 0$, there is an algorithm with time and sample complexity $N = \poly\left(k,1/\Delta,1/\lambda,1/\epsilon\right)$ which outputs an estimate for the centers and relative intensities of $\rho$ which incurs error $\epsilon$ with probability at least $9/10$. Furthermore, this holds even when there is granularity in the photon detector, as long as it is at most some inverse polynomial in all parameters.
\label{thm:above_informal}
\end{thm}

The main open question of our work is to prove matching upper and lower bounds that pin down the true diffraction limit. However, as we will discuss, this is a challenging problem in harmonic analysis, despite being connected to areas where there has been considerable recent progress. Moreover this phase transition for resolution is actually more dramatic than what happens for mixtures of Gaussians \cite{regev2017learning}. Even ignoring the issue of computational complexity, for spherical Gaussian mixtures it is known that at separation $o(\sqrt{\log k})$, super-polynomially many samples are needed, while at separation $\Omega(\sqrt{\log k})$, polynomially many suffice. 

We now say a word about the techniques that go into proving Theorem~\ref{thm:lowerbound_informal} and Theorem~\ref{thm:above_informal}. It turns out that both are closely related to proving a modified version of an \emph{Ingham-type estimate} \cite{komornik2005fourier}:

\begin{question}\label{question:harmonic}
	What is the smallest $\Delta$ for which the quantity 
	\begin{equation}
		\int_{B}\left|\sum^k_{j=1}\lambda_j e^{-2\pi i \langle \vmu_j, \omega\rangle} \right|^2\, d\omega \ge \frac{1}{\poly(k)}\norm{\lambda}^2_2\label{eq:estimate}
	\end{equation}
	for all vectors $\lambda\in\R^k$ and all sets of centers $\brc{\vmu_j}$ for which $\norm{\vmu_i - \vmu_j}_2 > \Delta$ for all $j\neq j'$, where the integration is over the origin-centered unit ball $B\subset \R^2$?
\end{question}

\noindent In particular, the main technical step for showing Theorem~\ref{thm:above_informal} is to show that the critical $\Delta$ in Question~\ref{question:harmonic} is at most $2j_{0,1}/\pi$. This can be obtained via the following extremal function. A \emph{ball minorant}, is a function $F$ satisfying the properties that 
\begin{enumerate}
\item[(1)] $F(x)\le \bone{x\in B}$ and 
\item[(2)] $\widehat{F}$ is supported on the ball of radius $\Delta$
\end{enumerate}
\noindent In \cite{holt1996beurling,carneiro2017hilbert,gonccalves2018note} it was shown that such a ball minorant exists for $\Delta = 2j_{0,1}/\pi$ (interestingly, this paved the way to some recent progress on Montgomery's famous pair correlation conjecture for the Riemann zeta function \cite{carneiro2017hilbert}). One can use property $(1)$ to pass from integrating against the function $\bone{x\in B}$ to integrating against $F$. And because by property $(2)$ $F$ is localized in the frequency domain, the latter integral is large.
In fact the one-dimensional analogue of Question~\ref{question:harmonic} was resolved in \cite{moitra2015super} using the univariate analogue of $F$, namely the \emph{Beurling-Selberg minorant}. However the algorithmic approach only made sense in one-dimension. In our case, we employ the tensor generalization of the matrix pencil method, originally introduced in \cite{huang2015super}.
We defer the details of this to Section~\ref{subsec:above}.

For the lower bound in Theorem~\ref{thm:lowerbound_informal}, we need to answer a variant of Question~\ref{question:harmonic}.

\begin{question}\label{question:harmonic2}
	What is the smallest $\Delta$ for which 
	\begin{equation}
		\int_{B}\left|\sum^k_{j=1}\lambda_j e^{-2\pi i \langle \vmu_j, \omega\rangle} - \sum^k_{j=1}\lambda'_j e^{-2\pi i\langle\vmu'_j, \omega\rangle} \right|^2\, d\omega \ge \frac{1}{\poly(k)}\label{eq:estimate2}
	\end{equation}
	for all \emph{nonnegative} $\lambda,\lambda'\in\R^k$ whose entries sum to one and $\brc{\vmu_j}, \brc{\vmu'_j}\in\R^2$ for which $\norm{\vmu_j - \vmu_{j'}}_2 > \Delta$ and $\norm{\vmu'_j - \vmu'_{j'}}_2 > \Delta$ for $j\neq j'$, where integration is over the origin-centered unit ball $B$?
\end{question}


The connection to Theorem~\ref{thm:lowerbound_informal} is straightforward: By Plancherel's and smoothness properties of $A_{\sigma}$, one can upper bound the $L_1$ distance between the mixture of Airy disks given by parameters $\brc{\lambda_j},\brc{\vmu_j}$ and the mixture given by $\brc{\lambda'_j}, \brc{\vmu'_j}$ in terms of the left-hand side of \eqref{eq:estimate2}. So if one can construct a set of $\Delta$-separated centers $\brc{\vmu_j}, \brc{\vmu'_j}$ for which \eqref{eq:estimate2} fails to hold but for which the collection $\brc{\vmu_j}$ is separated from $\brc{\vmu'_j}$ but the resulting mixtures of Airy disks are $o(1/\poly(k))$-close in total variation distance it implies that resolution is statistically impossible with a polynomial number of samples. This is the recipe used in known lower bounds \cite{moitra2010settling,hardt2015tight,regev2017learning} for learning mixtures of Gaussians.



For our purposes, it turns out that ``tensoring" one-dimensional lower bounds does not work because it would not beat the Abbe limit \cite{moitra2015super}. Morally, this is because tensoring the unit interval with itself would give us the unit square, which corresponds to separation in the $L_{\infty}$ distance rather than the $L_2$ distance, and the $L_2$ distance is the right distance in optics because it is rotationally invariant. The main technical contribution in our lower bound is to give a more sophisticated construction given by interleaving two triangular lattices and placing the centers at points on these lattices (see Figure~\ref{fig:lattice}). 
The analysis is rather delicate, and we defer the details to Section~\ref{sec:lowerboundprev} and Section~\ref{sec:lowerbound}.


To complete the picture, we show that there is no diffraction limit when the number of Airy disks is a constant. In particular we show that for any constant number of Airy disks there is an algorithm that takes a polynomial number of samples and runs in polynomial time that learns the parameters to any desired accuracy {\em regardless of the separation}. 

\begin{thm}[Informal, see Theorem~\ref{thm:main}]
	Let $\rho$ be a $\Delta$-separated superposition of $k$ Airy disks where every disk has relative intensity at least $\lambda$. Then for any target error $\epsilon > 0$ and failure probability $\delta > 0$, there is an algorithm which draws $N = \poly\left((k\sigma/\Delta)^{k^2},1/\lambda,1/\epsilon,\log(1/\delta)\right)$ samples from $\rho$, runs in time $O(N)$, and outputs an estimate for the centers and relative intensities of $\rho$ which incurs error $\epsilon$ with probability at least $1 - \delta$. Furthermore, this holds even when there is granularity in the photon detector, as long as it is at most some inverse polynomial in $N$.
	\label{thm:main_informal}
\end{thm}

This result turns out to be simple in retrospect, and comes from assembling a few standard tools from the literature on provably learning mixture models. Nevertheless it underscores an important point that existing tools can already have important implications for inverse problems the sciences. Our approach is to first estimate the Fourier transform $\widehat{\rho}$ from samples and then pointwise divide by $\widehat{A}_{\sigma}$. In this way we can simulate noisy access to the Fourier transform of the mixture of delta functions at $\vmu_1,\ldots,\vmu_k$. However $\widehat{A}_{\sigma}$ has compact support, so we can only access frequencies with bounded $L_2$ norm. Now we can reduce to the one-dimensional case  \cite{moitra2015super} by projecting $\rho$ along two nearby directions, solving each resulting univariate problem, and then solving an appropriate linear system to recover the centers and relative intensities. This method is reminiscent of \cite{kalai2010efficiently,moitra2010settling}, which gives algorithms for learning high-dimensional mixtures of Gaussians based on reducing to a series of one-dimensional problems and stitching together these estimates carefully. 


\subsection{Related Work}
\label{sec:related}

We have already mentioned that our work is closely related to the vast literature on learning mixture models and, in particular, on learning mixtures of Gaussians~\cite{dasgupta1999learning,DasguptaSchulman:00,AroraKannan:01,VempalaWang:02,AchlioptasMcSherry:05,brubaker2008isotropic,kalai2010efficiently,moitra2010settling,belkin2015polynomial,hardt2015tight,hsu2013learning,ge2015learning,regev2017learning,hopkins2018mixture,kothari2017outlier,diakonikolas2018list}. Here we mention some other connections to work on recovering spike trains from noisy, band-limited Fourier measurements.

\paragraph{Superresolution}

The seminal work of \cite{donoho1989uncertainty,donoho1992superresolution} was one of the first to put this question on rigorous footing. Donoho studied the modulus of continuity for this problem on a grid as the grid width goes to zero. Later Candes and Fernandez-Granda \cite{candes2014towards} gave a practical algorithm based on $L_1$ minimization over a continuous domain. There has been a long line of work on this problem which it would also be impossible to survey fully, so we refer the reader to \cite{candes2013super,tang2013compressed,fernandez2013support,liao2015music,moitra2015super,fernandez2016super,kunis2016multivariate,morgenshtern2016super} and references therein. 

We remark that essentially all works on super-resolution in high dimensions focus on the case where measurements are $L_{\infty}$-band-limited rather than $L_2$-band-limited. Given the prevalence of Airy disks and circular apertures in statistical optics, one upshot of our work is that, technical issues related to the so-called box (aka $L_{\infty}$ ball) minorant problem notwithstanding, the $L_2$ setting may be the more practically relevant one to consider anyways.

\paragraph{Sparse Fourier Transform} There are also connections to the extensive literature on the sparse Fourier transform, which can be interpreted in some sense as the ``agnostic'' version of the super-resolution problem where the goal is to compete with the error of the best $k$-sparse approximation to the discrete Fourier transform, even in the presence of noise, using few measurements \cite{gilbert2002near,gilbert2005improved,hassanieh2012nearly,gilbert2014recent,indyk2014nearly,kapralov2016sparse}. When the $k$ spikes need not be at discrete locations and the low-frequency measurements are randomly chosen, this is the problem of \emph{compressed sensing off the grid} introduced by \cite{tang2013compressed}, for which recovery is possible with far fewer measurements. This can be thought of as the one-dimensional case of the setting of \cite{huang2015super}. To our knowledge, the only work that addresses the continuous, high-dimensional version of the sparse Fourier transform is the very recent work of \cite{jin2020robust}. The emphasis in this literature is primarily on obtaining sample complexity near-linear in $k$, whereas our guarantees are only polynomial in $k$. Consequently the results in the sparse Fourier transform literature lose log factors in the level of separation they require, whereas in our setting the emphasis is primarily on the level of separation needed to get polynomial-time and -sample algorithms.

\subsection{Visualizing the Diffraction Limit}

In this short section we provide some figures to help conceptualize our results. Figure~\ref{fig:histogram} illustrates the basic notion that separation is information-theoretically unnecessary for parameter learning of superpositions of Airy disks. We compare the discretized empirical distribution of samples from two diffraction patterns whose components have separation well below the diffraction limit and thus well below what conventional wisdom in optics suggests is resolvable. While the differences in the diffraction patterns are minute, they do indeed become statistically significant with enough samples. Eventually it becomes possible to conclude that the gray diffraction pattern is generated by one point source and the red diffraction pattern is generated by two.

\begin{figure}
\centering
\includegraphics[width=\textwidth]{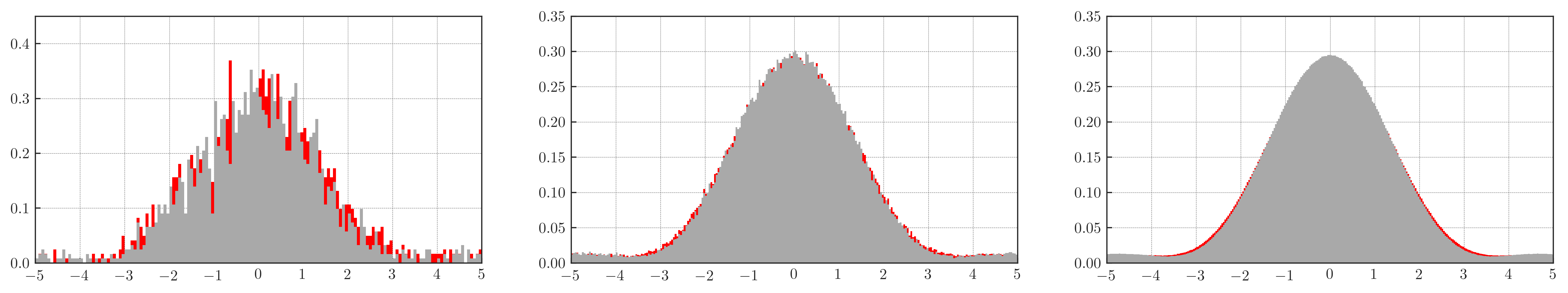}
\caption{\textbf{With enough samples, one can distinguish which of two superpositions the data comes from, even below the diffraction limit.} In each plot, a histogram of $x$-axis positions of photons sampled from a superposition of two equal-intensity Airy disks (red) centered on the $x$-axis with separation a tenth of the Abbe limit is overlaid with a histogram of $x$-axis positions of photons sampled from a single Airy disk at the origin (gray). As number of samples increases (left to right), minute differences between the two intensity profiles become clear.}
\label{fig:histogram}
\end{figure}

Next, we present a striking visual representation of the statistical barrier imposed by the diffraction limit when the number of components is large. Recall that the upshot of Theorems~\ref{thm:lowerbound_informal} and \ref{thm:above_informal} is that $k$ plays a leading role in determining when resolution is and is not feasible: slightly above the Abbe limit, the sample (and computational) complexity is polynomial in $k$, and anywhere beneath the Abbe limit, the sample complexity becomes exponential in $k$. This helps clarify why in some domains like astronomy, where there are only ever a few tightly spaced point sources, there is evidently no diffraction limit. Yet in other domains like microscopy where there are a large number of tightly spaced objects, the diffraction limit is indeed a fundamental barrier, at least in the classical physical setup. This helps explain why different communities have settled on different beliefs about whether there is or is not a diffraction limit.
 
In Figure~\ref{fig:tv} we experimentally investigate this phenomenon and illustrate how the total variation distance scales as we vary the number of disks and the separation in our earlier constructions. It is evident from these plots that for any superposition of a few Airy disks, there is no sharp dividing line between what is and is not possible to resolve. But when the number of Airy disks becomes large, with any reasonable number of samples, it is feasible to resolve the superposition if and only if their separation is at least as large as the Abbe limit.

\begin{figure}
	\centering
		\includegraphics[width=0.48\textwidth]{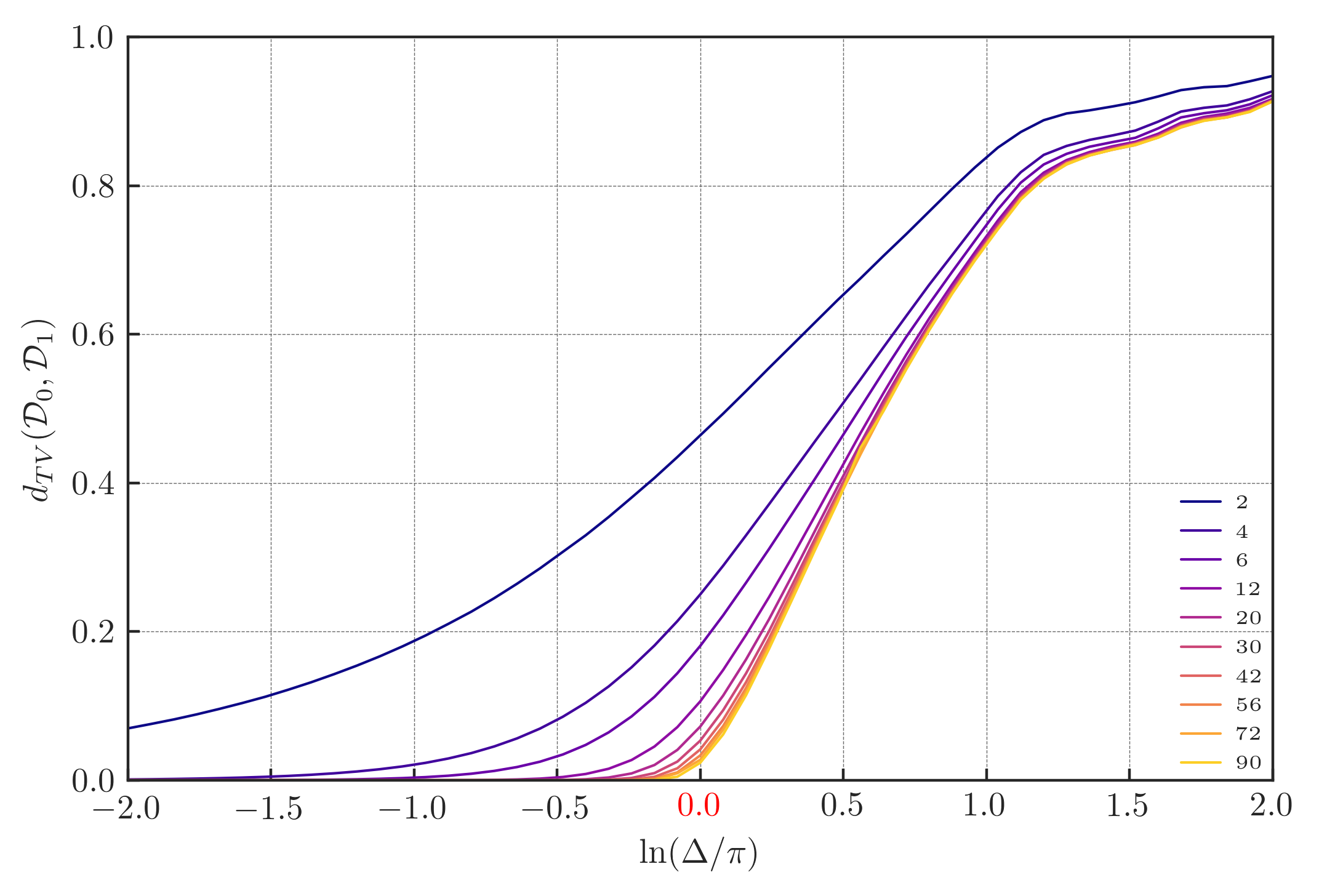}
		\includegraphics[width=0.48\linewidth]{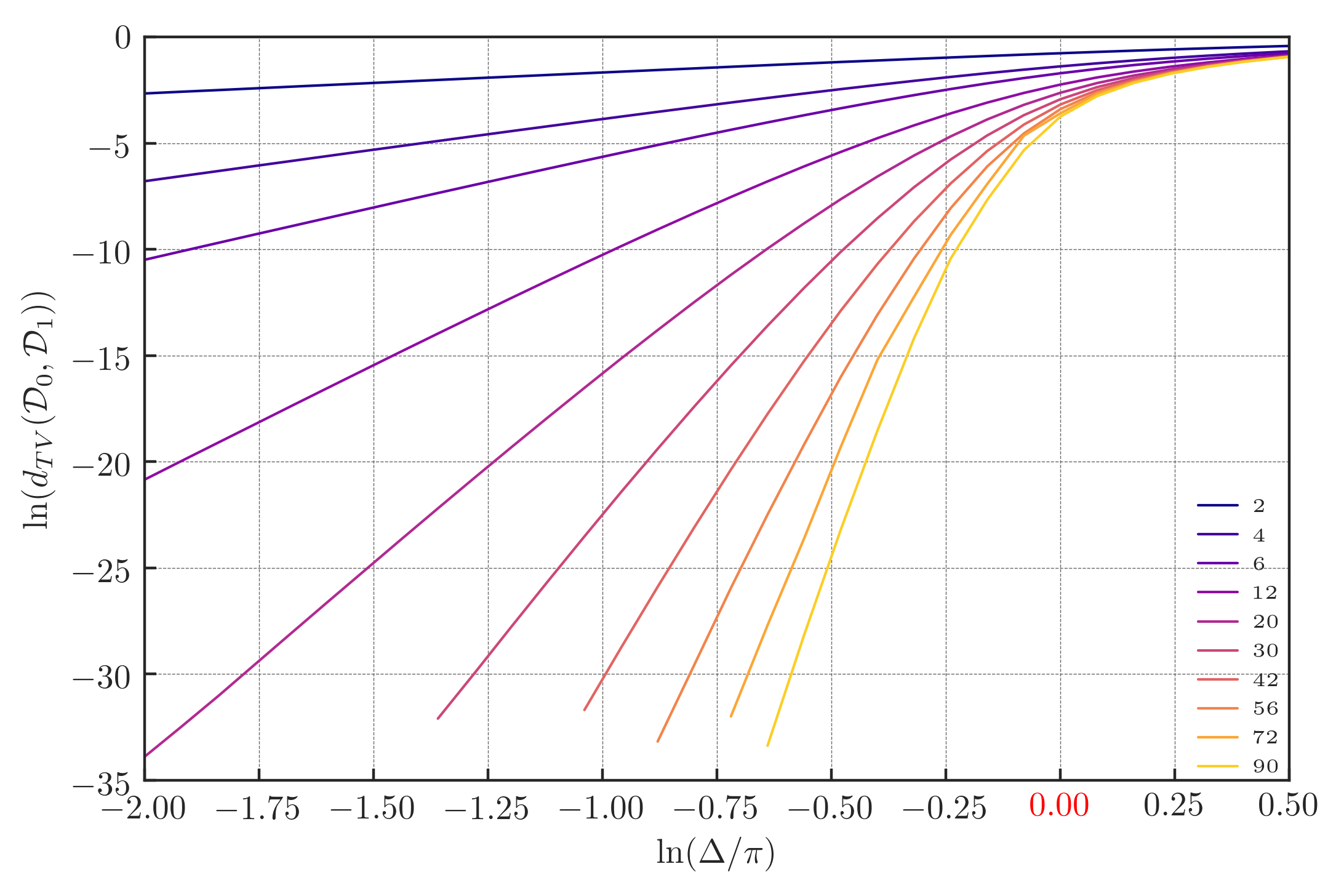}
	\caption{\textbf{The Abbe limit as a statistical phase transition.} For any level of separation $\Delta$ and number of disks $k$, we carefully construct a pair of hypotheses $\calD_0(\Delta,k), \calD_1(\Delta,k)$ which are each superpositions of $k/2$ Airy disks where the separation among its components is at least $\Delta$. The left figure plots total variation distance $d_{\text{TV}}(\calD_0(\Delta,k),\calD_1(\Delta,k))$ between the two distributions as a function of $\Delta$, for various choices of $k$, with the Abbe limit highlighted in red. The right figure plots total variation distance on a log-scale.}
	\label{fig:tv}
\end{figure}

We emphasize that in the instance constructed for Figure~\ref{fig:tv} (as well as the instance we construct and analyze for Theorem~\ref{thm:lowerbound_informal}), the centers are plotted on a line. For such instances, by projecting in the direction of the line and using our deconvolution techniques, one can actually reduce to the problem of one-dimensional super-resolution, for which polynomial-time algorithms exist for \emph{any} separation strictly greater than the diffraction limit \cite{moitra2015super}, and by adapting the lower bound in \cite{moitra2015super} to this specific instance, one can see that this is tight. In contrast, if the centers can be placed anywhere in $\R^2$, there is a constant factor gap ($\sqrt{4/3}$ versus $\frac{2j_{0,1}}{\pi}$) between the lower bound in Theorem~\ref{thm:lowerbound_informal} and the upper bound in Theorem~\ref{thm:above_informal}.

\subsection{Roadmap}
In Section~\ref{sec:lowerboundprev} we give a preview of our lower bound proof by providing a self-contained answer to Question~\ref{question:harmonic2}. In Section~\ref{sec:prelims}, we give an overview of our probabilistic model, some notation, and other mathematical preliminaries. In Section~\ref{sec:algo}, we prove the algorithmic results in Theorems~\ref{thm:main_informal} and \ref{thm:above_informal}. In Section~\ref{sec:lowerbound} we complete the proof of our lower bound from Theorem~\ref{thm:lowerbound_informal}.
In Section~\ref{sec:conclusion} we conclude with some directions for future work.
In Appendix~\ref{sec:science}, we overview previous attempts in the optics literature to put the diffraction limit on rigorous footing. In Appendix~\ref{sec:physical}, we describe and motivate our model and also define the various resolution criteria which have appeared in the literature. In Appendix~\ref{app:quotes}, we catalogue quotations from the literature that are representative of the points of view addressed in the introduction. In Appendix~\ref{app:jennrich}, we complete some deferred proofs. Lastly, in Appendix~\ref{sec:fig_details}, we give details on how Figure~\ref{fig:tv} was generated.

%% file: lowerboundpreview.tex
\section{Lower Bound Preview}
\label{sec:lowerboundprev}

In this section we give a self-contained proof of one of the main technical ingredients in the proof of our main result, Theorem~\ref{thm:lowerbound_informal}. Before proceeding, it will be convenient to introduce a bit of notation; any outstanding notation we will present Section~\ref{sec:prelims}, e.g. our convention for the Fourier transform. Recalling that $\lbound\triangleq \sqrt{4/3}$, define \begin{equation}m \triangleq \frac{2}{(1 - \epsilon)\lbound\pi\sigma}\label{eq:mdef}\end{equation} for any small constant $\epsilon > 0$ so that the critical level of separation for which Theorem~\ref{thm:lowerbound_informal} applies is $\Delta \triangleq 2/m = \lbound\cdot (1 - \epsilon)\cdot \pi\sigma$.\footnote{This $\pi\sigma$ scaling is not important to the result in this section but is the natural choice of scaling for Airy disks, so it will be convenient to work with this when we apply the results of this section to prove Theorem~\ref{thm:lowerbound_informal}.} Additionally, let $k$ be an odd square and define \begin{equation}\nu_{j_1,j_2} = \frac{\Delta}{2}\cdot(j_1,\sqrt{3}\cdot j_2), \; \; \; j_1,j_2 \in \calJ\triangleq\brk*{-\frac{\sqrt{k} - 1}{2},\ldots,\frac{\sqrt{k} - 1}{2}}.\label{eq:centersdef}\end{equation} This construction is illustrated in Figure~\ref{fig:lattice}: there, similarly colored points correspond to centers in the same mixture, and our choice of $\brc{\nu_{j_1,j_2}}$ ensures that the level of separation between any two points in a particular mixture is $\Delta$, which is slightly less than $\sqrt{4/3}$ times the Abbe limit of $\pi\sigma$. As such, the following tells us that the answer to Question~\ref{question:harmonic2} is surprisingly at least $\sqrt{4/3}$, rather than 1 as the Abbe limit would suggest:

\begin{lem}
	There exists a vector $u\triangleq (u_{j_1,j_2})_{j_1,j_2\in\calJ}\in\R^k$ for which \begin{equation}
		\abs*{\sum_{j_1,j_2\in\calJ} u_{j_1,j_2} e^{-2\pi i \iprod{\nu_{j_1,j_2},\vec{x}}}}^2 \le \exp\left({-\Omega(\epsilon\sqrt{k})}\right)
	\end{equation} for all $\norm{x}\le 1/\pi\sigma$. Furthermore, $\sgn(u_{j_1,j_2}) = (-1)^{j_1+j_2}$, and \begin{equation}\sum_{j_1+j_2 \ \text{even}}|u_{j_1,j_2}| = \sum_{j_1+j_2 \ \text{odd}}|u_{j_1,j_2}| = 1.\label{eq:alternating}\end{equation}\label{lem:l2bound}
\end{lem}

We need the following ingredient from the proof of the one-dimensional lower bound in \cite{moitra2015super}.

\begin{defn}
	The Fejer kernel is given by \begin{equation}
		K_{\ell}(x) = \frac{1}{\ell^2}\sum^{\ell}_{j = -\ell}(\ell - |j|)e(jx) = \frac{1}{\ell^2}\left(\frac{\sin \ell\pi x}{\sin \pi x}\right)^2. \label{eq:fejer}
	\end{equation} We will denote the $r$-th power of $K_{\ell}(\cdot)$ by $K^r_{\ell}(\cdot)$.
\end{defn}

\begin{fact}
	$K_{\ell}$ is even and periodic with period 1. For $x\in[-1/2,1/2]$, $K_{\ell}(x)\le \frac{1}{4\ell^2 x^2}$.\label{fact:kell}
\end{fact}

\begin{proof}
	That $K_{\ell}$ is even and periodic follow from the second definition of $K_{\ell}$ in \eqref{eq:fejer}. For the bound on $K_{\ell}$, we can use the elementary bounds $\sin\pi x\ge 2x$ for $x\in[0,1/2]$ and $(\sin \ell\pi x)^2\le 1$.
\end{proof}

\begin{figure}
	\centering
	\includegraphics[width=0.4\textwidth]{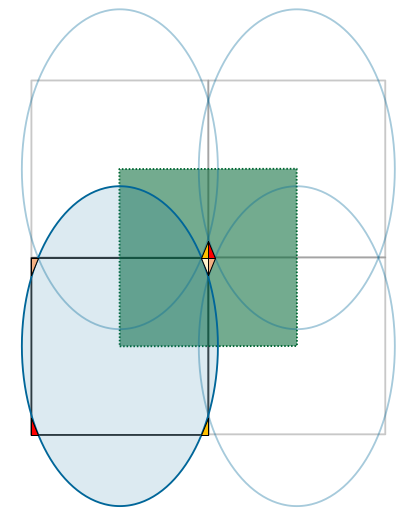}
	\caption{The squares correspond to periods of $K^r_{\ell}$, while the ellipses have major and minor axes of length $\lbound(1 - \epsilon)$ and $2(1 - \epsilon)$. The figure is centered around the origin, and the bottom-left ellipse $K$ is the set of points $\left(\frac{x_1}{m} - \frac{1}{2}, \frac{x_2\sqrt{3}}{m} - \frac{1}{2}\right)$ as $(x_1,x_2)$ ranges over the origin-centered $L_2$ ball of norm $1/\pi\sigma$. By appropriately translating the four quadrants of this ellipse by distances in $\Z^2$, we obtain overlapping regions whose union is given by $R\backslash S$, where $R = [-1/2,1/2]\times[-1/2,1/2]$ is given by the central square (green) and $S$ is the multi-colored set in the middle given by tranlates of the four connected components of $([-1,0]\times[-1,0])\backslash K$.}
	\label{fig:lowerbound}
\end{figure}

\begin{proof}[Proof of Lemma~\ref{lem:l2bound}]
	Let $\ell = 4/\epsilon$ and $r = (\sqrt{k} - 1)/2\ell = \Theta(\epsilon\sqrt{k})$, and assume without loss of generality that $\ell$ is even. Consider the function $H:\R^2\to\R$ given by \begin{equation}H(x_1,x_2) = K^r_{\ell}\left(\frac{x_1}{m} - \frac{1}{2}\right)\cdot K^r_{\ell}\left(\frac{x_2\sqrt{3}}{m} - \frac{1}{2}\right).\label{eq:H}\end{equation} We know that $\widehat{K}_{\ell}[t] = \frac{1}{\ell^2}\sum^{\ell}_{j = -\ell}(\ell - |j|)\delta(t - j)$, so $\widehat{K^r_{\ell}}[t] = \sum_{j\in\calJ}\alpha_j \delta(t - j)$ for nonnegative $\alpha_j$ which sum to 1. Assuming without loss of generality suppose that $m$ defined by \eqref{eq:mdef} is an odd integer, we conclude that for $\vec{t} = (t_1,t_2)\in\mathbb{C}^2$,
	\begin{align}
		\widehat{H}[\vec{t}] &= \sum_{j_1,j_2\in\calJ}h_{j_1,j_2} e^{-\pi i m(t_1+ t_2/\sqrt{3})} \delta(mt_1 - j_1)\cdot \delta(mt_2/\sqrt{3} - j_2) \\
		&= \sum_{j_1,j_2\in\calJ}h_{j_1,j_2} (-1)^{j_1+j_2}\delta(\vec{t} - \nu_{j_1,j_2}),
	\end{align} 
	where $h_{j_1,j_2} = \alpha'_{j_1}\alpha'_{j_2} \ge 0$, where $\alpha'_j\triangleq \alpha_j\cdot\bone{j = 0} + m\alpha_j\cdot \bone{j\neq 0}$. We will take \begin{equation}
		u_{j_1,j_2} \triangleq h_{j_1,j_2} (-1)^{j_1+j_2} \; \; \; \forall j_1,j_2\in\calJ.
	\end{equation} Observe that $\sgn(u_{j_1,j_2}) = (-1)^{j_1+j_2}$ as desired.

	By taking the inverse Fourier transform of $\hat{H}$, we get that \begin{equation}H(x_1,x_2) = \sum_{j_1,j_2}u_{j_1,j_2} e^{2\pi i \iprod{\nu_{j_1,j_2},\vec{x}}}.\label{eq:Huseful}\end{equation} To complete our proof, it therefore suffices to show that $H(\vec{x})\le \exp(-\Omega(\epsilon\sqrt{k}))$ for all $\norm{x} \le 1/\pi\sigma$.

	Let $R\subseteq \R^2$ denote the region $[-1/2,1/2]\times[-1/2,1/2]$. But as $x$ ranges over the $L_2$ ball of norm $1/\pi\sigma$, $\left(\frac{x_1}{m} - \frac{1}{2}, \frac{x_2\sqrt{3}}{m} - \frac{1}{2}\right)$ ranges over the interior of the ellipse centered at $(-1/2,1/2)$ with axes of length $\lbound(1-\epsilon)$ and $2(1 - \epsilon)$. For the subsequent discussion in this paragraph, we refer the reader to Figure~\ref{fig:lowerbound}. By periodicity of $K^r_{\ell}$, the image of this ellipse under $K^r_{\ell}$ is identical to the region $T\triangleq R\backslash S$, where $S$ is defined as follows. Denote the interior of the ellipse by $B_1$, and denote its translates along the vectors $(0,1)$, $(1,0)$, and $(1,1)$ by $B_2,B_3,B_4$. Define $S$ to be the set of points in $R$ that belong to none of $B_1,B_2,B_3,B_4$. 

	We claim that $S$ contains the origin-centered $L_{\infty}$ ball of radius $\epsilon/2\sqrt{2}$. Note that $S$ is given by translating the four connected components of $([-1,0]\times[-1,0])\backslash B_1$, which is nonempty because $B_1$ consists of points $(x_1,x_2)$ satisfying \begin{equation}
		\frac{4}{\lbound^2(1 - \epsilon)^2}(x_1 - 1/2)^2 + \frac{1}{(1 - \epsilon)^2}(x_2 - 1/2)^2 \le 1.\label{eq:ellipse}
	\end{equation} In particular, for $x_1,x_2\in[-1,0]$ satisfying $|x_1 - 1/2|, |x_2 - 1/2| > (1 - \epsilon)/2$, observe that the left-hand quantity in \eqref{eq:ellipse} satisfies \begin{equation}
		\frac{4}{\lbound^2(1 - \epsilon)^2}(x_1 - 1/2)^2 + \frac{1}{(1 - \epsilon)^2}(x_2 - 1/2)^2 > \left(\frac{4}{\lbound^2(1 - \epsilon)^2} + \frac{1}{(1-\epsilon)^2}\right)\cdot \frac{(1 - \epsilon)^2}{4} = 1,
	\end{equation}
	where the last step follows by our choice of $\lbound = \sqrt{4/3}$. We conclude that $S$ contains the origin-centered $L_{\infty}$ ball of radius $\epsilon/2$ as claimed.

	Now by Fact~\ref{fact:kell}, for any $(x_1,x_2)\in R$ we have that $K^r_{\ell}(x_1),K^r_{\ell}(x_2) \le \frac{1}{4^{2r}\ell^{4r}x^{4r}}$. So because $S$ contains the origin-centered $L_{\infty}$ ball of radius $\epsilon/2$, for $(x_1,x_2)\in T$ we conclude that $K^r_{\ell}(x_1), K^r_{\ell}(x_2) \le 1/4^{4r}$. We conclude that for $\norm{\vec{x}} \le 1/\pi\sigma$, $H(\vec{x}) \le \exp(-\Omega(r)) = \exp(-\Omega(\epsilon\sqrt{k}))$.


	The last step is just to scale $u$ so that \eqref{eq:alternating} holds. First note that by substituting $x = 0$ into \eqref{eq:Huseful}, we have that \begin{equation*}\sum_{j_1,j_2\in\calJ} u_{j_1,j_2} = H(0,0) = K^{2r}_{\ell}(-1/2) = \frac{1}{\ell^{4r}}\sin^{4r}(\ell\pi/2).\end{equation*} In particular, because we assumed at the outset that $\ell$ is even, $H(0,0) = 0$. Together with the fact that $\sgn(u_{j_1,j_2}) = (-1)^{j_1+j_2}$, we get the first equality in \eqref{eq:alternating}. Finally, note that $\sum |u_{j_1,j_2}| > 1$ because $\sum\alpha_j = 1$ and  $h_{j_1,j_2} \ge \alpha_{j_1}\alpha_{j_2}$ for all $j_1,j_2$. Thus, by multiplying $u$ by a factor of at most 2, we get the second equality in \eqref{eq:alternating}.
\end{proof}

\begin{figure}
	\centering
	\includegraphics[width=0.3\textwidth]{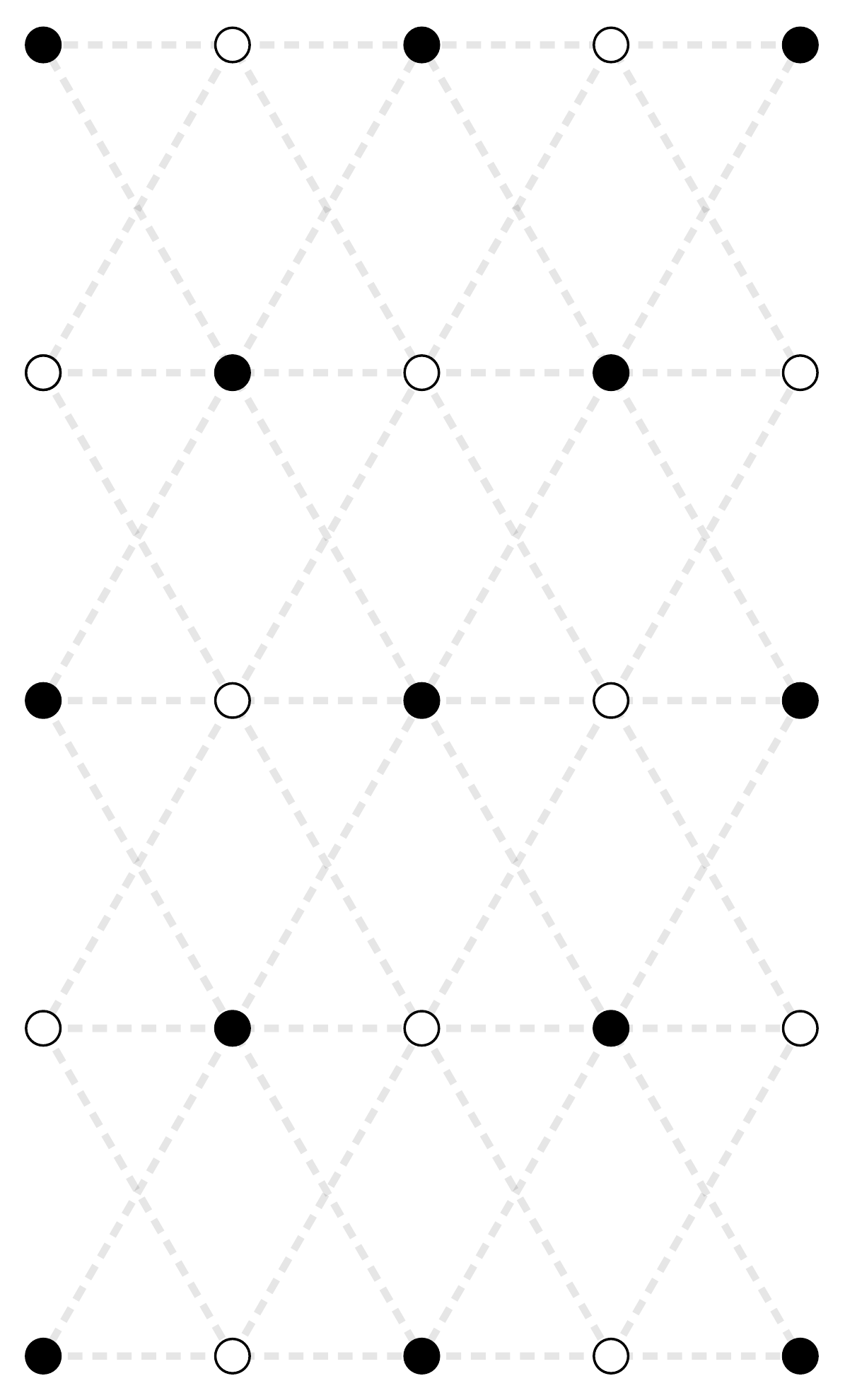}
	\caption{Locations of centers of Airy disks for the two mixtures in the lower bound instance of Theorem~\ref{thm:mainlowerbound} when $k = 25$. Black (resp. white) points correspond to centers for $\rho$ (resp. $\rho'$). The separation between any adjacent pair of identically colored points is $2/m = \Delta$, and the points of any particular color form a triangular lattice.}
	\label{fig:lattice}
\end{figure}

%% file: prelims.tex

\section{Preliminaries}
\label{sec:prelims}

In this section we explain the terminology and notation that we will adopt in this work and also provide some technical preliminaries that will be useful later.

\paragraph{Generative Model}

We first formally define the family of distributions we study in this work.

\begin{restatable}{defn}{airy}[Superpositions of Airy Disks]
	A \emph{superposition of $k$ Airy disks} $\rho$ is a distribution over $\R^2$ specified by relative intensities $\lambda_1,...,\lambda_k\ge 0$ summing to 1, centers $\vmu_1,...,\vmu_k\in\R^2$, and an \emph{a priori} known ``spread parameter'' $\sigma> 0$. Its density is given by \begin{equation}\rho(\vec{x}) = \sum^k_{i=1}\lambda_i\cdot A_{\sigma}\left(\vec{x} - \vmu_i\right) \ \ \ \text{for} \ \ \ A_{\sigma}(\vec{z}) = \frac{1}{\pi \sigma^2}\left(\frac{J_1(\norm{\vec{z}}_2/\sigma)}{\norm{\vec{z}}_2/\sigma}\right)^2.\label{eq:modeldef}\end{equation} Note that the factor of $\frac{1}{\pi\sigma^2}$ in the definition of $A_{\sigma}$ is to ensure that $A_{\sigma}(\cdot)$ is a probability density.

	Also define \begin{equation}
		\Delta \triangleq \min_{i\neq j}\norm{\vmu_i - \vmu_j}_2 \ \ \ \text{and} \ \ \ \radius\triangleq \max_{i\in[k]}\norm{\vmu_i}_2.
	\end{equation}\label{defn:model}
\end{restatable}

It will be straightforward to extend the above model to take into account error stemming from the fact that the photon detector itself only has finite precision.

\begin{defn}[Discretization Error]
	Given discretization parameter $\varsigma>0$, we say $\vec{x}$ is a \emph{$\varsigma$-granular sample} from $\rho$ if it is produced via the following generative process: 1) a point $\vec{x}'$ is sampled from $\rho$, 2) $\vec{x}$ is obtained by moving $\vec{x}'$ an arbitrary distance of at most $\varsigma$.
\end{defn}

\paragraph{Fourier Transform} We will use the following convention in defining the Fourier transform. Given $f\in L_2(\R^d)$, \begin{equation}\widehat{f}(\omega) \triangleq \int_{\R^d} f(x) \cdot e^{-2\pi i\langle \omega,x\rangle} \, dx.\label{eq:fourier_def}\end{equation}

\paragraph{Optical Transfer Function}

The following is a standard calculation.

\begin{fact}
	$\hat{A}_{\sigma}[\omega] = \frac{2}{\pi}(\arccos(\pi \sigma \norm{\omega}) - \pi \sigma \norm{\omega}\sqrt{1 - \pi^2 \sigma^2 \norm{\omega}^2}$.\label{fact:fourier}
\end{fact}

\begin{proof}
	It is enough to show this for $\sigma = 1$. Let $G(\vec{x})\triangleq J_1(\norm{\vec{x}})/\norm{\vec{x}}$. It is a standard fact that the zeroeth-order Hankel transform of the function $r\mapsto J_1(r)/r$ is the indicator function of the interval $[0,1]$. Using our convention for the Fourier transform (see \eqref{eq:fourier_def}), this implies that $\hat{G}[\omega] = 2\pi\cdot\bone{\norm{\omega} \in [0,1/2\pi]}$. Because $A_1 = G^2/\pi$, by the convolution theorem we conclude that $\hat{A}_1$ is $\frac{1}{\pi}$ times the convolution of $\hat{G}$ with itself, which is just $4\pi^2$ times the convolution of the indicator function of the unit disk of radius $1/2\pi$ with itself. By elementary Euclidean geometry one can compute this latter function to be $\omega\mapsto \frac{1}{2\pi^2}\cdot \left(\arccos(\pi \norm{\omega}) - \pi \norm{\omega}\sqrt{1 - \pi^2 \norm{\omega}^2}\right)$, from which the claim follows.
\end{proof}

In optics, the two-dimensional Fourier transform of the point-spread function is called the \emph{optical transfer function}, a term we will occasionally use in the sequel.

Now note that by Fact~\ref{fact:fourier}, $\hat{A}_{\sigma}$ is supported only over the disk of radius $\frac{1}{\pi \sigma}$ centered at the origin in the frequency domain. In the spatial domain, this corresponds to a separation of $\pi \sigma$; this is the definition of the \emph{Abbe limit}.
We will need the following elementary estimate for $\widehat{A}[\omega]$:

\begin{fact}
	For all $\norm{\omega}_2 \le 1$, $\widehat{A}[\omega] \ge (1 - \norm{\omega}_2)^2$.
	\label{fact:ahatestimate}
\end{fact}

\paragraph{Scaling} As the algorithms we give will be scale-invariant, we will assume that $\sigma = 1/\pi$ in the rest of this work and refer to $A_{1/\pi}$ as $A$.

\paragraph{Parameter Estimation Accuracy} The following terminology formalizes what it means for an algorithm to return an accurate estimate for the parameters of a superposition of Airy disks.

\begin{defn}
	$\left(\{\lambda^*_i\}_{i\in[k]}, \{\vmu^*_i\}_{i\in[k]}\right)$is an $(\epsilon_1,\epsilon_2)$-accurate estimate for the parameters of a superposition of $k$ Airy disks $\rho$ with centers $\{\vmu_i\}$ and relative intensities $\{\lambda_i\}$ if there exists a permutation $\tau$ for which \begin{equation}\norm{\vmu_i - \tilde{\vmu}_{\tau(i)}}_2\le \epsilon_1 \; \; \text{and} \; \; |\lambda_i - \tilde{\lambda}_{\tau(i)}|\le \epsilon_2\label{eq:define_accurate}\end{equation} for all $i\in[k]$.
\end{defn}

\paragraph{Generalized Eigenvalue Problems} Given matrices $M,N$, we will denote by $(M,N)$ the \emph{generalized eigenvalue problem} $Mx = \lambda Nx$. In any solution $(\lambda,x)$ to this, $\lambda$ is called a \emph{generalized eigenvalue} and $x$ is called a \emph{generalized eigenvector}.

\paragraph{Bessel Function Estimates} In Section~\ref{sec:lowerbound}, we need the following estimate for $J_{\nu}(z)$:

\begin{thm}[\cite{landau2000bessel}]
	For some absolute constant $\Cl[c]{landaubessel} = 0.7857...$, we have for all $\nu\ge 0$ and $r\in\R$ that $|J_{\nu}(r)|\le \Cr{landaubessel}|r|^{-1/3}$.
	\label{thm:besselbound}
\end{thm}

\paragraph{Matrices, Tensors, and Flattenings} Given a matrix $M\in\co^{a\times b}$, we will denote its $i$-th row vector by $M_i$, its $j$-th column vector by $M^j$, and its $(i,j)$-th entry by $M_{i,j}$.

Given a tensor $\vec{T}\in\co^{m_1\times m_2\times m_3}$ and matrices $M_1\in\co^{m_1\times m'_1}$, $M_2\in\co^{m_2\times m'_2}$, and $M_3\in\co^{m_3\times m'_3}$, define the flattening $\vec{T}(M_1,M_2,M_3)\in\co^{m'_1\times m'_2\times m'_3}$ by \begin{equation}
	\vec{T}(M_1,M_2,M_3)_{i_1,i_2,i_3} = \sum_{(j_1,j_2,j_3)\in[m_1]\times[m_2]\times[m_3]}\vec{T}_{m_1,m_2,m_3}\cdot (M_1)_{j_1,i_1}(M_2)_{j_2,i_2}(M_3)_{j_3,i_3}
\end{equation} for all $(i_1,i_2,i_3)\in[m'_1]\times[m'_2]\times[m'_3]$.

\paragraph{Miscellaneous Notation} Let $\S^{d-1}$ denote the Euclidean unit sphere. Given $r>0$, let $B^d(r)$ denote the Euclidean ball of radius $r$ centered at the origin in $\R^d$.

%% file: algo.tex

\section{Learning Superpositions of Airy Disks}
\label{sec:algo}

In this section we present the technical details of our algorithmic results. In Sections~\ref{subsec:learnwithotf} and \ref{subsec:kde}, we prove the following formal version of Theorem~\ref{thm:main_informal}.

\begin{thm}
	Let $\rho$ be a $\Delta$-separated superposition of $k$ Airy disks with minimum mixing weight $\lambda_{\min}$ and such that $\norm{\vmu_i}\le \radius$ for all $i\in[k]$.

	For any $\epsilon_1,\epsilon_2 > 0$, there is some $\alpha = \poly\left(\log 1/\delta,1/\lambda_{\min},1/\epsilon_1,1/\epsilon_2,\radius,(k\sigma/\Delta)^{k^2}\right)^{-1}$ for which there exists an algorithm with time and sample complexity $\poly(1/\alpha)$ which, given $\varsigma = \poly(\alpha)$-granular sample access to $\rho$, outputs an $(\epsilon_1,\epsilon_2)$-accurate estimate for the parameters of $\rho$ with probability at least $1 - \delta$.
	\label{thm:main}
\end{thm}

Specifically, in Section~\ref{subsec:learnwithotf}, we show how one can use the matrix pencil method to recover the parameters for $\rho$ given oracle access to the \emph{optical transfer function}, i.e. the two-dimensional Fourier transform of $\rho$, up to some small additive error. In Section~\ref{subsec:kde}, we show how to implement this approximate oracle.

In Section~\ref{subsec:above}, we also use the oracle of Section~\ref{subsec:kde} to prove the following formal version of Theorem~\ref{thm:above_informal}.

\begin{thm}
Let $\rho$ be a $\Delta$-separated superposition of $k$ Airy disks with minimum mixing weight $\lambda_{\min}$ and such that $\norm{\vmu_i}\le \radius$ for all $i\in[k]$. Let \begin{equation}
	\factor = \frac{2j_{0,1}}{\pi} = 1.530\ldots,
	\label{eq:factor_def}
\end{equation} where $j_{0,1}$ is the first positive zero of the Bessel function of the first kind $J_0$. For any $\Delta > \factor\cdot \pi\cdot\sigma$, the following holds:

For any $\epsilon_1,\epsilon_2 > 0$, there is some $\alpha = 1/\poly\left(k,\radius,\sigma/\Delta,1/\lambda_{\min},1/\epsilon_1, 1/\epsilon_2, 1/(\Delta - \factor)\right)$ for which there exists an algorithm with time and sample complexity $\poly(1/\alpha)$ which, given $\poly(\alpha)$-granular sample access to $\rho$, outputs an $(\epsilon_1,\epsilon_2)$-accurate estimate for the parameters of $\rho$ with probability at least $4/5$.
\label{thm:above}
\end{thm}



\subsection{Reduction to 2D Superresolution}
\label{subsec:reduce}

In this section we reduce the problem of learning superpositions of Airy disks to the problem of learning a convex combination of Dirac deltas given the ability to make noisy, band-limited Fourier measurements.

Formally, suppose we had access to the following oracle:

\begin{defn}
	An \emph{$m$-query, $\eta$-approximate OTF oracle} $\O$ takes as input a frequency $\omega\in\R^2$ and, given frequencies $\omega_1,...,\omega_m$, outputs numbers $u_1,...,u_m\in\R$ for which $|u_j - \widehat{\rho}[\omega_j]|\le\eta$ for all $j\in[m]$.
\end{defn}

\begin{remark}
	As we will see in Section~\ref{subsec:kde}, $\O$ will be constructed by sampling some number of points from $\rho$ and computing empirical averages. The number $m$ and accuracy $\eta$ of queries that $\O$ can answer dictates the sample complexity of this procedure. As we will see in the proofs of Lemma~\ref{lem:learnairydisks} and Lemma~\ref{lem:tensorresolve_correct} below, the $m$ that we need to take will be small, so the reader can ignore $m$ and pretend it is unbounded for most of this section.
\end{remark}

Given $\omega\in\R^2$, the Fourier transform of $\rho$ evaluated at frequency $\omega$ is given by \begin{equation}
\widehat{\rho}[\omega] = \sum^k_{j=1}\lambda_j\widehat{A}[\omega]e^{-2\pi i\langle \vmu_j,\omega\rangle},
	\label{eq:fourier}
\end{equation} where for $\omega = (r\cos\theta, r\sin\theta)$, we have by Fact~\ref{fact:fourier} that \begin{equation}
	\widehat{A}[\omega] = \frac{2}{\pi}(\arccos(r) - r\sqrt{1-r^2}).\label{eq:Afourier}
\end{equation} In particular, $\widehat{A}[\omega]$ only depends on $r = \norm{\omega}$ (because $A(\cdot)$ is radially symmetric), so henceforth regard $\widehat{A}$ as a function merely of $r$.

Define \begin{equation}F(\omega) = \sum^k_{j=1}\lambda_j e^{-2\pi i\langle \vmu_j,\omega\rangle}.\label{eq:mixture}\end{equation} This is a trigonometric polynomial to which we have noisy pointwise access using $\O$:

\begin{lem}
	Let $0<r<1$. With an $\eta$-approximate OTF oracle $\O$, on input $\omega\in B^2(r)$ we can produce an estimate of $F(\omega)$ to within $\eta/\widehat{A}[r]$ additive error.\label{lem:Fapproximate}
\end{lem}

\begin{proof}
	By dividing by $\widehat{A}[\omega]$ on both sides of \eqref{eq:fourier}, we get that \begin{equation*}
		\frac{\widehat{\rho}[\omega]}{\widehat{A}[\norm{\omega}]} = \sum^k_{j=1}\lambda_j e^{-2\pi i\langle \vmu_j,\omega\rangle},
	\end{equation*} so given that $\O$, on input $\omega$, outputs $u\in\R$ satisfying $|u - \widehat{\rho}[\omega]|\le \eta$, we have that \begin{equation*}\left|\frac{u}{\widehat{A}[\norm{\omega}]} - F(\omega)\right| \le \frac{\eta}{\min_{0\le r'\le r}\widehat{A}[r']} = \frac{\eta}{\widehat{A}[r]},\end{equation*} where the last step uses the fact that $\widehat{A}[\cdot]$ is decreasing on the interval $[0,1]$.
\end{proof}

So concretely, given an $\eta$-approximate OTF oracle, we have reduced the problem of learning superpositions of Airy disks to that of recovering the locations of $\{\vmu_j\}$ given the ability to query $F(\omega)$ at arbitrary frequencies $\omega$ for which $\norm{\omega}_2 < 1$ to witin additive accuracy $\eta/\widehat{A}[\norm{\omega}_2]$.

Lastly, for reasons that will become clear in subsequent sections (see e.g. \eqref{eq:radius}), it will be convenient to assume that $\radius\le 1/3$. This is without loss of generality, as otherwise, we can scale the data down by a factor of $3\radius$ so that they are now i.i.d. samples from the superposition of Airy disks with density $\rho'(\vx)\triangleq \sum^k_{j=1}\lambda_j\cdot A_{1/\radius}\left(\vx - \vmu_j/\radius\right)$. Define the rescaled centers $\vmu'_j \triangleq \vmu_j/\radius$ and note that by assumption, $\|\vmu'_j\|_2 \le 1/3$ for all $j\in[k]$.

The Fourier transform of $\rho'$ is then given by $\hat{\rho}'(\omega) = \hat{A}_{1/\radius}[\omega]\sum^k_{j=1}\lambda_j e^{-2\pi i \langle \vmu'_j, \omega\rangle}$, so by the proof of Lemma~\ref{lem:Fapproximate} we conclude that with an $\eta$-approximate OTF oracle for $\rho$, for any $0<r<1$ on input $\omega\in B^2(r\cdot\radius)$ we can produce an estimate of $\sum^k_{j=1}\lambda_j e^{-2\pi i \langle \vmu'_j, \omega\rangle}$ to within $\eta/\hat{A}[r]$ additive error. Recovering the centers $\{\vmu'_j\}$ to within additive error $\epsilon$ then translates to recovering the centers $\{\vmu_j\}$ to within additive error $3\radius \epsilon$. For this reason, we will henceforth assume that $\radius \le 1/3$.

\subsection{Learning via the Optical Transfer Function}
\label{subsec:learnwithotf}

\input{learnwithotf}

\subsection{Learning Airy Disks Above the Diffraction Limit}
\label{subsec:above}

\input{above}

\subsection{Approximating the Optical Transfer Function}
\label{subsec:kde}

\input{kde2}

%% file: learnwithotf.tex

Our basic approach is as follows. To solve the superresolution problem of Section~\ref{subsec:reduce}, we will project in two random, correlated directions $\omega_1,\omega_2\in\R^2$ and solve the resulting one-dimensional superresolution problems via matrix pencil method (see \textsc{ModifiedMPM}) to recover the projections of $\vmu_1,...,\vmu_k$ in the directions $\omega_1$ and $\omega_2$, as well as the relative intensities $\lambda_1,...,\lambda_k$. From these projections we can then recover the actual centers for $\rho$ by solving a linear system (\textsc{PreConsolidate}). Such an approach already achieves constant success probability, and we can amplify this by repeating and running a simple clustering algorithm (see \textsc{Select}). The full specification of the algorithm is given as \textsc{LearnAiryDisks}.

\subsubsection{Learning in a Random Direction}
\label{subsubsec:randomdirection}

Fix a unit vector $v\in\S^1$. We first show how to leverage Lemma~\ref{lem:Fapproximate} and the matrix pencil method to approximate the projection of $\vmu_1,...,\vmu_k$ along $v$.

By the discussion at the end of Section~\ref{subsec:reduce}, we may assume $\|\vmu_i\|_2 \le 1/2$ for all $i\in[k]$, so $\norm{\vmu_i - \vmu_j}_2 \le 1$ for all $i\neq j$. For $j\in[k]$, let $m_j = \langle\vmu_j, v\rangle$ and $\alpha_j = e^{2\pi i \cdot (m_j /4k)}$. In this section we will assume that $m_j\neq 0$ for all $j\in[k]$

For $\ell\in\Z_{\ge 0}$, let \begin{equation}v_{\ell} = F\left(\frac{\ell v}{4 k}\right) = \sum^k_{j=1}\lambda_j\alpha_j^{\ell}.\label{eq:vell}\end{equation} Note that $v_0 = F(\vec{0}) = \sum_j \lambda_j = 1$. Also note that we do not have access to $\alpha_1,...,\alpha_k$ and would like to recover $m_1,m_2$ given (noisy) access to $\{v_{\ell}\}_{0\le\ell\le 2k-1}$. 

Consider the generalized eigenvalue problem $(VD_{\lambda}V^{\top}, VD_{\lambda}D_{\alpha}V^{\top})$ where \begin{equation}V = \begin{pmatrix}
	1 & 1 & \cdots & 1 \\
	\alpha_1 & \alpha_2 & \cdots & \alpha_k \\
	\vdots & \vdots & \ddots & \vdots \\
	\alpha^{k-1}_1 & \alpha^{k-1}_2 & \cdots & \alpha^{k-1}_k
\end{pmatrix}, \ \ \ \ \ D_{\lambda} = \diag(\lambda), \ \ \ \ \ D_{\alpha} = \diag(\alpha).\end{equation} The following standard facts are key to the matrix pencil method:

\begin{obs}
	The generalized eigenvalues of $(VD_{\lambda}V^{\top}, VD_{\lambda}D_{\alpha}V^{\top})$ are exactly $\alpha_1,...,\alpha_k$.
\end{obs}

\begin{obs}
	\begin{equation}VD_{\lambda}V^{\top} = \begin{pmatrix} v_0 & v_1 & \cdots & v_{k-1} \\
	v_1 & v_2 & \cdots & v_k \\
	\vdots & \vdots & \ddots & \vdots \\
	v_{k-1} & v_k & \cdots & v_{2k-2} \end{pmatrix} \ \ \ VD_{\lambda}D_{\alpha}V^{\top} = \begin{pmatrix} v_1 & v_2 & \cdots & v_{k} \\
	v_2 & v_3 & \cdots & v_{k+1} \\
	\vdots & \vdots & \ddots & \vdots \\
	v_{k} & v_{k+1} & \cdots & v_{2k-1} \end{pmatrix}.\end{equation}
\end{obs}

By Lemma~\ref{lem:Fapproximate}, in reality we only have $\eta'_{\ell}$-approximate access to each $v_{\ell}$, where \begin{equation}\eta'_{\ell}\le \frac{\eta}{\widehat{A}[\ell/4k]},\label{eq:etaell}\end{equation} so we must instead work with the generalized eigenvalue problem $(VD_{\lambda}V^{\top} + E, VD_{\lambda}D_{\alpha}V^{\top} + F)$, where the $(i,j)$-th entry of $E$ (resp. $F$) is the noise $\eta'_{i+j-2}$ (resp. $\eta'_{i+j-1}$) in the observation of $v_{i+j-2}$ (resp. $v_{i+j-1}$).

If $V$ is well-conditioned, one can apply standard perturbation bounds to argue that the solutions to this generalized eigenvalue problem are close to those of the original $(VD_{\lambda}V^{\top}, VD_{\lambda}D_{\alpha}V^{\top})$. Moreover, given approximations $\widehat{\alpha}_1,...,\widehat{\alpha}_k$ to these generalized eigenvalues, we can find approximations $\widehat{\lambda}_1,...,\widehat{\lambda}_k$ to $\lambda_1,...,\lambda_k$ by solving the system of equations $\vec{v} = \widehat{V}\vec{\lambda}$, where $\vec{v} = (v_0,...,v_{k-1})$, $\vec{\lambda} = (\widehat{\lambda}_1,...,\widehat{\lambda}_k)$, and \begin{equation}\widehat{V} = \begin{pmatrix}
	1 & 1 & \cdots & 1 \\
	\widehat{\alpha}_1 & \widehat{\alpha}_2 & \cdots & \widehat{\alpha}_k \\
	\vdots & \vdots & \ddots & \vdots \\
	\widehat{\alpha}^{k-1}_1 & \widehat{\alpha}^{k-1}_2 & \cdots & \widehat{\alpha}^{k-1}_k
\end{pmatrix}.\end{equation} The formal specification of the matrix pencil method algorithm \textsc{ModifiedMPM} that we use is given in Algorithm~\ref{alg:mpm}.

\begin{algorithm}\caption{\textsc{ModifiedMPM}}\label{alg:mpm}
\begin{algorithmic}[1]
	\State \textbf{Input}: $\omega\in\S^1$, $\eta$-approximate OTF oracle $\mathcal{O}$
	\State \textbf{Output}: Estimates $(\widehat{\lambda}_1,...,\widehat{\lambda}_k)$ for the mixing weights and $(\widehat{m}_1,...,\widehat{m}_k)$ for the centers of $\rho$ projected in direction $\omega$
		\State Define $\widehat{v}_0 = 1$.
		\State For $0\le \ell \le 2k - 1$: invoke $\O$ on input $\frac{\ell\omega}{4 k}$ to produce $u_{\ell}\in\R$. Compute $\widehat{v}_{\ell} \triangleq \frac{u_{\ell}}{\widehat{A}[\ell/4 k]}$.
		\State Form the matrices \begin{equation}X\triangleq\begin{pmatrix}
			\widehat{v}_0 & \cdots & \widehat{v}_{k-1} \\
			\vdots & \ddots & \vdots \\
			\widehat{v}_{k-1} & \cdots & \widehat{v}_{2k-2}
		\end{pmatrix} \ \ \ \ \ Y\triangleq \begin{pmatrix}
			\widehat{v}_1 & \cdots & \widehat{v}_k \\
			\vdots & \ddots & \vdots \\
			\widehat{v}_k & \cdots & \widehat{v}_{2k-1}
		\end{pmatrix}\end{equation}
		\State Solve the generalized eigenvalue problem $(X,Y)$ to produce generalized eigenvalues $\widehat{\alpha}_1,\widehat{\alpha}_2$.
		\State For $i = 1,2$, let $\widehat{m}_i$ be the argument of the projection of $\widehat{\alpha}_i$ onto the complex unit disk.
		\State Form the matrix \begin{equation}\widehat{V} = \begin{pmatrix}
			1 & 1 & \cdots & 1 \\
			\widehat{\alpha}_1 & \widehat{\alpha}_2 & \cdots & \widehat{\alpha}_k \\
			\vdots & \vdots & \ddots & \vdots \\
			\widehat{\alpha}^{k-1}_1 & \widehat{\alpha}^{k-1}_2 & \cdots & \widehat{\alpha}^{k-1}_k
		\end{pmatrix}.\end{equation}
		\State Solve for $\widehat{\lambda} = (\widehat{\lambda}_1,...,\widehat{\lambda}_k)$ such that $\widehat{V}\widehat{\lambda} = (\widehat{v}_0,...,\widehat{v}_{k-1})$.
		\State Output $\{\widehat{\lambda}_i\}_{i\in[k]}$ and $\{\widehat{m}_i\}_{i\in[k]}$.
\end{algorithmic}
\end{algorithm}

The following theorem, implicit in the proof of Theorem 2.8 in \cite{moitra2015super}, makes the above reasoning precise. Henceforth, let $\kappa(\Delta')$ and $\sigma_{\min}(\Delta')$ respectively denote the condition number and minimum singular value of $V$ when $\frac{m_i}{4 k}, \frac{m_j}{4 k}$ have minimum separation $\Delta'$ for all $i\neq j$, and define $\lambda_{\min} = \min_i \lambda_i$, $\lambda_{\max} = \max_i \lambda_i$.

\begin{thm}[\cite{moitra2015super}]
	Suppose $\frac{m_1}{4 k},\frac{m_2}{4 k}\in[-1/4,1/4]$ have separation at least $\Delta'$ and we are given $\eta'_{\ell}$-close estimates to $v_{\ell}$ for $0\le\ell\le 2k - 1$.

	Define \begin{equation}\gamma = \frac{2\norm{\eta'}_2}{\lambda_{\min}}\left(4\kappa(\Delta')^2\cdot\frac{\lambda_{\max}}{\lambda_{\min}} + \frac{1}{\sigma_{\min}(\Delta')^2}\right) \; \; \text{and} \; \; \zeta = O\left(\frac{2\gamma\lambda_{\max} + \norm{\eta'}_2}{\sigma_{\min}(\Delta' - 2\gamma)}\right)\end{equation}

	Then if $\norm{E} + \norm{F} < \sigma_{\min}(\Delta')^2\lambda_{\min}$ and $\gamma < \Delta'/4$, \textsc{ModifiedMPM} produces estimates $\{\widehat{\lambda}_i\}$ for the mixing weights and estimates $\{\widehat{m}_i\}$ for the projected centers such that for some permutation $\tau$: \begin{equation}|m_i-\widehat{m}_{\tau(i)}| \le 8\gamma \; \; \text{and} \; \; |\lambda_i - \widehat{\lambda}_i|\le \zeta.\end{equation}\label{thm:stability} for all $i\in[k]$.
\end{thm}

Note that the guarantees of Theorem~\ref{thm:stability} are stated in \cite{moitra2015super} in terms of wraparound distance on the interval $[-1/2,1/2]$, but because $\frac{m_i}{4 k}\in[-1/4,1/4]$ for all $j\in[k]$, $\frac{m_1}{4 k},....,\frac{m_k}{4 k}$ have pairwise separation $\Delta'$ both in absolute and wraparound distance.

In other words, the output of \textsc{ModifiedMPM} converges to the true values for $\{\langle\vmu_1,v\rangle\}_{j\in[k]}$ and $\{\lambda_j\}_{j\in[k]}$ at a rate polynomial in the noise rate, condition number of $V$, and relative intensity of the Airy disks, provided $\sigma_{\min}(\Delta')$ is inverse polynomially large and $\kappa(\Delta')$ is polynomially small in those parameters.

To complete the argument, we must establish these bounds on $\sigma_{\min}$ and $\kappa$. Henceforth, let \begin{equation}\Delta' = \min_{i\neq j}\frac{m_i - m_j}{4 k}.\end{equation}

\begin{lem}
	For any $k \ge 2$, we have that \begin{equation}
		\sigma_{\min}(\Delta')^2 \ge (\Delta'^{k}/k^2)^{k-1} \ \ \ \text{and} \ \ \ \kappa(\Delta')^2 \le k^{2k-1} / \Delta'^{k(k-1)}\label{lem:condnumber_general}
	\end{equation}
\end{lem}

\begin{proof}
	First note that $\sigma_{\max}(\Delta')^2 \le k^2$. Indeed because the entries of $V$ all have absolute value at most 1, we conclude that for any $v\in\S^{k-1}$ and any row index $j\in[k]$, \begin{equation}\langle V_j, v\rangle^2 \le \left(\sum^k_{i=1}|v_j|\right)^2 \le k.\end{equation} On the other hand, we also have that \begin{equation*}\prod^k_{i=1}\sigma_i(V) = |\det(V)| = \prod_{1\le i< j\le k}|\alpha_i - \alpha_j| \le \left|e^{2\pi i \Delta'} - 1\right|^{\binom{k}{2}} = \left(2 - 2\cos(\Delta')\right)^{\binom{k}{2}/2} \ge \Delta'^{k(k-1)/2},
	\end{equation*} where in the first step we used the standard fact that the absolute value of the determinant of a square matrix is equal to the product of its singular values, in the second step we used the standard identity for the determinant of a Vandermonde matrix, in the third step we used the angular separation of the $\alpha_i$'s, and in the final step we used the elementary inequality $\cos(\Delta') \le 1 - \Delta'^2/2$. We may thus naively lower bound $\sigma_{\min}(V)$ by $\frac{\Delta'^{k(k-1)/2}}{k^{k-1}}$, from which the lemma follows.
\end{proof}

This yields the following consequence for \textsc{ModifiedMPM}.

\begin{cor}
	Given $\omega\in\S^1$ and access to an $\eta$-approximate OTF oracle $\O$, if the projected centers $m_j = \langle\vmu_j,v\rangle$ satisfy $|m_i - m_j|\le 4 k\cdot \Delta'$ for all $i\neq j$ for some $0 < \Delta' \le 1/16$, then there exists a constant $\Cl[c]{etagap}>0$ such that provided that \begin{equation}\eta \le \Cr{etagap}\lambda_{\min}^2\Delta'^{k^2}\cdot k^{-2k-1/2},\label{eq:etabound}\end{equation} then \textsc{ModifiedMPM} produces estimates $\{\widehat{\lambda}_i\}$ for the mixing weights and estimates $\{\widehat{m}_i\}$ for the projected centers such that for some permutation $\tau$: \begin{equation}|m_i-\widehat{m}_{\tau(i)}| \le O\left(\frac{k^{2k + 1/2}\cdot \eta}{\lambda^2_{\min}\Delta'^{k(k-1)}}\right) \; \; \text{and} \; \; |\lambda_i - \widehat{\lambda}_{\tau(i)}|\le O\left(\frac{k^{3k - 1/2}\cdot \eta}{\lambda^2_{\min}\Delta'^{3k(k-1)/2}}\right)\end{equation} for all $i\in[k]$.\label{cor:parameters}
\end{cor}

\begin{proof}
	From Lemma~\ref{lem:condnumber_general}, we have that $\sigma_{\min}(\Delta')^2\ge (\Delta'^k/k^2)^{k-1}$. Then because $\kappa(\Delta')^2\le k^2/\sigma_{\min}(\Delta')^2$, 
	we would like to conclude by Theorem~\ref{thm:stability} that $|m_i - \widehat{m}_{\tau(i)}| \le 8\gamma$, where \begin{equation*}
		\gamma = O\left(\frac{\norm{\eta'}_2\cdot k^2}{\lambda_{\min}^2 \cdot \sigma_{\min}(\Delta')^2}\right) = O\left(\frac{\norm{\eta'}_2 \cdot k^2}{\lambda^2_{\min}\cdot (\Delta'^{k}/k^2)^{k-1}}\right) = O\left(\frac{k^{2k + 1/2}\cdot \eta}{\lambda^2_{\min}\Delta'^{k(k-1)}}\right),
	\end{equation*} where in the last step we use that the vector $\eta'$ has length $O(k)$ and satisfies $\norm{\eta'}_{\infty} \le O(\eta)$ by \eqref{eq:etaell}. To do so, we just need to verify that $\norm{E} + \norm{F} < \sigma_{\min}(\Delta')^2\lambda_{\min}$ and $\gamma < \Delta'/4$. The latter clearly follows from the bound \eqref{eq:etabound} for sufficiently small $\Cr{etagap}$. For the former, note that \begin{equation*}
		\norm{E}_2 \le \norm{E}_F \le \sqrt{k}\cdot\sqrt{\eta'^2_1 + \cdots + \eta'^2_{2k-1}} \le \eta\sqrt{k}\cdot \sqrt{\sum^{2k-1}_{\ell=1}\frac{1}{\hat{A}[\ell/4 k]}} \le O(\eta\cdot k),
	\end{equation*} where the last step follows by the fact that $\hat{A}[\ell/4 k] \ge \hat{A}[1/2] \ge \Omega(1)$. The same bound holds for $\norm{F}_2$. Recalling that $\sigma_{\min}(\Delta')^2 \ge (\Delta'^k/k^2)^{k-1}$, it is enough for $\eta \le O(\Delta'^k/k^2)^{k-1} \lambda_{\min}/k$, which certainly holds for $\eta$ satisfying \eqref{eq:etabound}, for $\Cr{etagap}$ sufficiently small.

	Finally, Theorem~\ref{thm:stability} also implies that $|\lambda_i - \widehat{\lambda}_i|\le\zeta$, where \begin{equation}
		\zeta\le O\left(\frac{\gamma + k^{1/2}\eta}{\sigma_{\min}(\Delta' - 2\gamma)}\right) \le O\left(\frac{k^{2k + 1/2}\cdot \eta}{\lambda^2_{\min}\Delta'^{k(k-1)}\cdot \left(\Delta'^{k(k-1)/2}/k^{k-1}\right)}\right) = O\left(\frac{k^{3k - 1/2}\cdot \eta}{\lambda^2_{\min}\Delta^{3k(k-1)/2}}\right)
	\end{equation} as claimed.
\end{proof}

\subsubsection{Combining Directions}

We can run \textsc{ModifiedMPM} to approximately recover $\{\langle \vmu_j,\omega_1\rangle\}_{j\in[k]}$ and $\{\langle \vmu_j,\omega_2\rangle\}_{j\in[k]}$ for two randomly chosen directions $\omega_1,\omega_2\in\S^1$. As these directions are random, with high probability we can combine these estimates to obtain an accurate estimate of $\{\vmu_j\}_{j\in[k]}$. One subtlety is that the estimates $\{\widehat{m}_j\}$ and $\{\widehat{m}'_j\}$ output by \textsc{ModifiedMPM} for the centers projected in directions $\omega_1$ and $\omega_2$ respectively need not be aligned, that is we only know that there exists some permutation $\tau$ for which $\widehat{m}_j = \widehat{m}'_{\tau(j)}$ for $j\in[k]$. 

We first show a ``pairing lemma'' stating that if $\omega_1$ is chosen randomly and $\omega_2$ is chosen to be close to $\omega_1$, then if one sorts the centers $\vmu_1,...,\vmu_k$, first in terms of their projections in the $\omega_1$ direction, and then in terms of their projections in the $\omega_2$ direction, the corresponding elements in these two sorted sequences will correspond to the same centers.

We require the following elementary fact.

\begin{lem}
	For $\vmu\in\R^2$ a unit vector and $\omega\in\R^2$ a random unit vector, $\Pr_{\omega}[|\langle \vmu,\omega\rangle|\le \sin\theta]= 2\theta/\pi$ for all $0\le\theta\le\pi/2$.
	\label{lem:rand_rotation}
\end{lem}

\begin{lem}
	Fix an arbitrary $0<\theta\le\pi/2$ and let $\upsilon = \frac{\Delta\sin\theta}{8}$. Let $\omega_1\in\R^2$ be a random unit vector, and let $\omega_2\in\R^2$ be either of the two unit vectors for which $\norm{\omega_1 - \omega_2}_2 = \upsilon$. For every $i\in[k]$, define $m_i \triangleq \langle \vmu_i,\omega_1\rangle$ and $m'_i \triangleq \langle \vmu_i,\omega_2\rangle$, and let $\hat{m}_i, \hat{m}'_i\in\R$ be any numbers for which $\norm{\hat{m}_i - m_i}_2, \norm{\hat{m}'_i - m'_i}_2 \le 2\upsilon$.

	Then with probability at least $1 - \frac{k(k-1)\theta}{\pi}$, for every $i\neq j$ the following are equivalent: I) $m_i > m_j$, II) $m'_i > m'_j$, III) $\hat{m}_i > \hat{m}_j$, and IV) $\hat{m}'_i > \hat{m}'_j$.
	\label{lem:ordering}
\end{lem}

\begin{proof}
	By Lemma~\ref{lem:rand_rotation} and a union bound we have that with probability $1 - \frac{k(k-1)\theta}{\pi}$, $|m_i - m_j|>\Delta\sin\theta$ for all $i\neq j$. Fix any $i\neq j$ and suppose that $m_i > m_j$. Then by triangle inequality and Cauchy-Schwarz, we have that \begin{equation}m'_i - m'_j = \langle \vmu_1 - \vmu_2, \omega_2\rangle = \langle \vmu_1 - \vmu_2, \omega_1\rangle + \langle \vmu_1 - \vmu_2, \omega_2 - \omega_1\rangle \ge \Delta\sin\theta - 4\upsilon > 0,\end{equation} where the final inequality follows by the definition of $\upsilon$. So I) implies II) and by symmetry we can show II) implies I). We also have that \begin{equation}
		\hat{m}_i - \hat{m}_j \ge (m_i - m_j) - 4\upsilon > 0,
	\end{equation} so I) implies III) and by symmetry we can show II) implies IV).

	It is enough to show that III) implies I). Suppose $\hat{m}_i > \hat{m}_j$. Then \begin{equation}m_i - m_j \ge (\hat{m}_i - \hat{m}_j) - 4\upsilon > -\frac{1}{2}\Delta\sin\theta > -\Delta\sin\theta,\end{equation} so it must be the case that $m_i - m_j > 0$ given that $|m_i - m_j| > \Delta\sin\theta$.
\end{proof}

We now show that we can combine these projected center estimates to approximately recover the two-dimensional centers by solving a linear system. The specification of this algorithm, which we call \textsc{PreConsolidate}, is given in Algorithm~\ref{alg:preconsolidate}.

\begin{algorithm}\caption{\textsc{PreConsolidate}}\label{alg:preconsolidate}
\begin{algorithmic}[1]
	\State \textbf{Input}: Directions $\omega_1,\omega_2\in\S^1$ and estimates $\{\widehat{\lambda}_i\},\{\widehat{m}_i\}$ and $\{\widehat{\lambda}'_i\},\{\widehat{m}'_i\}$ for the parameters of $\rho$ projected in the directions $\omega_1 $ and $\omega_2$ respectively
	\State \textbf{Output}: An estimate of the form $(\{\tilde{\lambda}_i\}, \{\tilde{\vmu}_i\})$ for the parameters of $\rho$.
	\For{$i\in[k]$}
		\State Let $\ell,\ell'\in[k]$ be the indices for which $\hat{m}_{\ell}$ and $\hat{m}'_{\ell'}$ are the $i$-th largest in $\{\hat{m}_j\}_{j\in[k]}$ and $\{\hat{m}'_j\}_{j\in[k]}$ respectively.
		\State Define a formal vector-valued variable $\vec{v}^{(i)}\in\R^2$ and solve the linear system \begin{align*}\langle \omega_1, \vec{v}^{(i)}\rangle & = \hat{m}_{\ell} \\
		\langle \omega_2, \vec{v}^{(i)}\rangle & = \hat{m}'_{\ell'}.\end{align*}
	\EndFor
	\State Output $\left(\{\hat{\lambda}_i\}_{i\in[k]}, \{\vec{v}^{(i)}\}_{i\in[k]}\right)$.
\end{algorithmic}
\end{algorithm}

\begin{lem}
	Let $\xi > 0$. Let the parameters $\theta,\upsilon$ and the random vectors $\omega_1,\omega_2$ be as in Lemma~\ref{lem:ordering}. Suppose $\{\hat{m}_i\}_{i\in[k]}$ and $\{\hat{m}'_i\}_{i\in[k]}$ are collections of numbers for which there exist permutations $\tau_1,\tau_2\in\S_k$ for which \begin{equation}
		\left|\langle \omega_1, \vmu_i \rangle - \hat{m}_{\tau_1(i)}\right|\le \xi \ \ \ \text{and} \ \ \ \left|\langle \omega_2, \vmu_i \rangle - \hat{m}'_{\tau_2(i)}\right| \le \xi
	\end{equation} for all $i\in[k]$.

	Then for any estimates $\{\hat{\lambda}_i\}_{i\in[k]}$ and $\{\hat{\lambda}'_i\}_{i\in[k]}$, with probability at least $1 - \frac{k(k-1)\theta}{\pi}$ we have that the output $\left(\{\tilde{\lambda}_i\}, \{\tilde{\vmu}_i\}\right)$ of \textsc{PreConsolidate}($\omega_1$,$\omega_2$, $\{\hat{\lambda}_i\}$, $\{\hat{m}_i\}$, $\{\hat{\lambda}'_i\}$, $\{\hat{m}'_i\}$) satisfies \begin{equation}\norm{\vmu_i - \tilde{\vmu}_{\tau(i)}}_2\le \frac{\xi}{\upsilon\sqrt{1 - \upsilon^2/4}}\label{eq:accurate}\end{equation} for some permutation $\tau\in\S_k$.\label{lem:preconsolidate}
\end{lem}

\begin{proof}
	Condition on the event of Lemma~\ref{lem:ordering} occurring, which happens with probability at least $1 - \frac{k(k-1)\theta}{\pi}$. This event implies that there is a permutation $\tau\in\S_k$ such that for every $i\in[k]$ in the loop of \textsc{PreConsolidate}, the indices $\ell,\ell'$ in that iteration are such that $\hat{m}_{\ell}$ and $\hat{m}'_{\ell'}$ are $\xi$-close estimates for the projections of $\vmu_{\tau(i)}$ in the directions $\omega_1$ and $\omega_2$ respectively. In other words, $\tau_1(\tau(i)) = \ell$ and $\tau_2(\tau(i)) = \ell'$.

	Let $A\in \R^{2\times 2}$ be the matrix with rows consisting of $\omega_1$ and $\omega_2$. We conclude that \begin{equation}
		\norm{\vmu_{\tau(i)} - \vec{v}^{(i)}}_2 = \norm{A^{-1}\cdot \left(\left(\hat{m}_{\ell},\hat{m}'_{\ell'}\right) - \left(\langle \omega_1,\vmu_{\tau(i)}\rangle, \langle \omega_2, \vmu_{\tau(i)}\rangle\right)\right)}_2 \le \sigma_{\min}(A)\cdot\xi,
	\end{equation} so it remains to bound $\sigma_{\min}(A)$. Without loss of generality we may assume $\omega_1 = (1,0)$ and $\omega_2 = (x,\sqrt{1 - x^2})$ for $x \triangleq 1 - \upsilon^2/2$, in which case $\sigma_{\min}(A) = \upsilon\sqrt{1 - \upsilon^2/4}$, and the claim follows.
\end{proof}

Finally, we show how to boost the success probability via the following naive clustering-based algorithm \textsc{Select} (Algorithm~\ref{alg:select}), whose guarantees we establish below.

\begin{algorithm}\caption{\textsc{Select}}\label{alg:select}
\begin{algorithmic}[1]
	\State \textbf{Input}: Accuracy parameters $\epsilon'_1,\epsilon'_2$, and list $\mathcal{L}$ consisting of $T$ candidate estimates for the parameters of $\rho$, each of the form $\left(\{\tilde{\lambda}^t_i\}_{i\in[k]},\{\tilde{\vmu}^t_i\}_{i\in[k]}\right)$ for $t\in[T]$, such that for at least $1 - \frac{1}{2k}$ fraction of all $t\in[T]$, $\left(\{\tilde{\lambda}^t_i\}_{i\in[k]},\{\tilde{\vmu}^t_i\}_{i\in[k]}\right)$ is an $(\epsilon'_1,\epsilon'_2)$-accurate estimate of the parameters of $\rho$.
	\State \textbf{Output}: A $(3\epsilon'_1,\epsilon'_2)$-accurate estimate of the parameters of $\rho$
			\State Define $\mathcal{S} = T\times[k]$.
			\State Form the graph $G = (V,E)$ whose vertices consist of all $(t,i)$ for which $\tilde{\vmu}^t_i\in\mathcal{S}$ is $2\epsilon'_1$-close to at least $2T/3$ other points in $\mathcal{S}$, with edges between any $(t,i),(t',i')$ for which $\norm{\tilde{\vmu}^t_i - \tilde{\vmu}^{t'}_{i'}} > 6\epsilon'_1$.
			\State $G$ is $k$-partite. Denote the parts by $V^{(1)},...,V^{(k)}\subset V$.
			\For{$j\in[k]$}
				\State Form the set $\{\tilde{\lambda}^t_i\}_{(t,i)\in V^{(j)}}$ and let $\lambda^*_{j}$ be the median of this set, corresponding to some $(t_j,i_j)\in\mathcal{S}$.
				\State Define $\vmu^*_j\triangleq \tilde{\vmu}^{t_j}_{i_j}$.
			\EndFor
			\State Output $\left(\{\lambda^*_j\}_{j\in[k]},\{\vmu^*_j\}_{j\in[k]}\right)$.
\end{algorithmic}
\end{algorithm}

We can now give the full specification of our algorithm \textsc{LearnAiryDisks} (see Algorithm~\ref{alg:learn}).

\begin{algorithm}\caption{\textsc{LearnAiryDisks}}\label{alg:learn}
\begin{algorithmic}[1]
	\State \textbf{Input}: Error parameters $\epsilon_1,\epsilon_2$, confidence parameter $\delta$, access to $\eta$-approximate, $O(\log 1/\delta)$-query OTF oracle
	\State \textbf{Output}: With probability at least $1 - \delta$, an $(\epsilon_1,\epsilon_2)$-accurate estimate $(\{\tilde{\lambda}_i\}, \{\tilde{\vmu}_i\})$ for the parameters of $\rho$.
			\State Initialize a list $\mathcal{L}$ of candidate estimates for the parameters of $\rho$.
			\State Set $\theta \triangleq \frac{\pi}{3k^2(k-1)}$.
			\State Set $\eta$ according to \eqref{eq:eta_final}.
			\For{$T = \Omega(\log(1/\delta))$ iterations}\label{step:loop}
				\State Sample a random unit vector $\omega_1\in\S^1$ and let $\omega_2$ be either of the two unit vectors for which $\norm{\omega_1 - \omega_2}_2 = \Delta\sin\theta/8$ (see Lemma~\ref{lem:ordering}).
				\State Run \textsc{ModifiedMPM}($\omega_1,\mathcal{O}$) and \textsc{ModifiedMPM}($\omega_2,\mathcal{O}$) to obtain estimates $\{\widehat{\lambda}_i\},\{\widehat{m}_i\}$ and $\{\widehat{\lambda}'_i\},\{\widehat{m}'_i\}$ for the parameters of $\rho$ projected in the directions $\omega_1,\omega_2$ respectively.
				\State Let $\{\tilde{\lambda}^t_i\}, \{\tilde{\vmu}^t_i\}$ be the estimates output by \textsc{PreConsolidate}($\omega_1,\omega_2,\{\widehat{\lambda}_i\},\{\widehat{m}_i\},\{\widehat{\lambda}'_i\},\{\widehat{m}'_i\}$). Append these to $\mathcal{L}$.
			\EndFor
			\State Output what is returned by \textsc{Select}($\mathcal{L},\epsilon_1/3,\epsilon_2)$).
\end{algorithmic}
\end{algorithm}

\begin{lem}
Let $\rho$ be a $\Delta$-separated superposition of $k$ Airy disks. For any $\epsilon_1,\epsilon_2,\delta > 0$, let \begin{equation}
		\eta = O\left(\left(\frac{\Delta}{4 k}\right)^{O(k^2)}\cdot\lambda^2_{\min}\right)\cdot \min\left\{\epsilon_1/M,\epsilon_2\right\}.\label{eq:eta_final}
	\end{equation}
 Without loss of generality suppose $\epsilon_1 < 3\Delta/8$. Then the output $(\lambda^*_1,\lambda^*_2,\vmu^*_1,\vmu^*_2)$ of \textsc{LearnAiryDisks}, given $\epsilon_1,\epsilon_2,\delta$ and access to an $\eta$-approximate, $O(\log(1/\delta))$-query OTF oracle $\O$ for $\rho$, satisfies \begin{equation*}\norm{\vmu_i - \vmu^*_{\tau(i)}}_2\le\epsilon_1 \; \; \text{and} \; \; |\lambda_i - \lambda_{\tau(i)}|\le\epsilon_2\end{equation*} for some permutation $\tau$ with probability at least $1 - \delta$. Furthermore, the runtime of \textsc{LearnAiryDisks} is dominated by the time it takes to invoke the OTF oracle $O(\log(1/\delta))$ times.\label{lem:learnairydisks}
\end{lem}

\begin{proof}
	Suppose we are given a valid $\eta$-approximate OTF oracle $\mathcal{O}$.
	By taking $\theta = \frac{\pi}{3k^2(k-1)}$ and invoking Lemmas~\ref{cor:parameters} and \ref{lem:preconsolidate}, we ensure that a single run of \textsc{PreConsolidate} in an iteration of the loop in Step~\ref{step:loop} of \textsc{LearnAiryDisks} will yield, with probability at least $1 - \frac{1}{3k}$, an $(\epsilon'_1,\epsilon'_2)$-accurate estimate, where \begin{equation}
		\epsilon'_1 = \frac{8}{\Delta\sin\theta}\cdot O\left(\frac{k^{2k + 1/2}\cdot \eta}{\lambda^2_{\min}\left(\frac{\Delta\sin\theta}{4 k}\right)^{k(k-1)}}\right) \ \ \ \text{and} \ \ \ \epsilon'_2 = O\left(\frac{k^{3k - 1/2}\cdot\eta}{\lambda^2_{\min}\left(\frac{\Delta\sin\theta}{4 k}\right)^{3k(k-1)/2}}\right).
	\end{equation} In this case we say that such an iteration of the loop in \textsc{LearnAiryDisks} ``succeeds.'' Note that if we take \begin{equation}
		\eta = O\left(\min\left\{\epsilon_1\cdot \frac{\Delta\sin\theta}{8}\cdot \frac{\lambda^2_{\min}\left(\frac{\Delta\sin\theta}{4 k}\right)^{k(k-1)}}{k^{2k + 1/2}}, \epsilon_2 \cdot \frac{\lambda^2_{\min}\left(\frac{\Delta\sin\theta}{4 k}\right)^{3k(k-1)/2}}{k^{3k-1/2}}\right\}\right),
	\end{equation} then we can ensure that $\epsilon'_1 = \epsilon_1/3$ and $\epsilon'_2 = \epsilon_2$. The bound in \eqref{eq:eta_final} then follows from the elementary inequality $\sin\theta \ge \theta/2$ for $0\le \theta \le 1$, together with our choice of $\theta = \frac{\pi}{3k^2(k-1)}$.

	Each iteration of the loop in Step~\ref{step:loop} of \textsc{LearnAiryDisks} individually succeeds with probability at least $1 - \frac{1}{3k}$. So by a Chernoff bound, by taking $T = \Omega(\log(1/\delta))$, we conclude that with probability at least $1 - \delta$, at least $1 - \frac{1}{2k}$ fraction of these iterations will succeed. So of the $k\cdot T$ elements in $\mathcal{S}$, at most $T/2$ correspond to failed iterations.

	Now note that all $(t,i)$ for which $t$ corresponds to a successful iteration will be $2\epsilon'_1$-close to at least $k\cdot T - T/2 > 2T/3$ points. In particular, any such $(t,i)$ will be among the vertices $V$ of $G$ in Algorithm~\ref{alg:select}. Conversely, for any $(t,i)\in V$, $\tilde{\vmu}^t_i$ is by definition $2\epsilon'_1$-close to at least $2T/3$ points and there are at most $T/2 < 2T/3$ points which do not correspond to successful iterations. In particular, at least one of the points that $\tilde{\vmu}^t_i$ is close to will correspond to a successful iteration, so by the triangle inequality $\norm{\tilde{\vmu}^t_i - \vmu_j}\le 3\epsilon'_1$ for some choice of $j\in[k]$.

	Observe that $G$ is $k$-partite because every vertex in $V$ is $3\epsilon'_1$-close to some center of $\rho$, but two vertices which are $3\epsilon'_1$-close to $\vmu_i$ and $\vmu_j$ respectively for $i\neq j$ must be distance at least $\Delta - 6\epsilon'_1 > 2\epsilon'_1$ apart. We conclude that with high probability, \textsc{Select} will output $3\epsilon'_1 = \epsilon_1$-accurate estimates for the centers of $\rho$.

	It remains to show that $\lambda^*_1,\lambda^*_2$ are $\epsilon_2$-accurate estimates for the mixing weights. We know the estimates $\tilde{\lambda}^t_i$ corresponding to successful iterations $t$ and center $\vmu_i$ lie in $\{\tilde{\lambda}^t_i\}_{(t,i)\in V^{(\ell)}}$ for some $\ell$. Then $\{\tilde{\lambda}^t_i\}_{(t,i)\in V^{(\ell)}}$ contains at least $\left(1 - \frac{1}{2k}\right)T > 2T/3$ values that are $\epsilon'_2$-close to $\lambda_1$, and at most $T/2 < 2T/3$ other values. Call these values ``good'' and ``bad'' respectively. Either the median is good, in which case we are done, or the median is bad, in which case because there are strictly more good values than bad values, the median must be upper and lower bounded by good values, in which case we are still done.

	Finally, note that in each iteration of the main loop of \textsc{LearnAiryDisks}, $\O$ is invoked exactly six times. Furthermore, other than these invocations of $\O$, the remaining steps of \textsc{LearnAiryDisks} all require constant time. So the runtime of \textsc{LearnAiryDisks} is indeed dominated by the $O(\log(1/\delta))$ calls to $\O$.
\end{proof}

%% file: above.tex

In this subsection we present the proof of Theorem~\ref{thm:above}. Recall that we are assuming that $\sigma = 1/\pi$ and $\Delta > \factor$, where $\factor$ is defined in \eqref{eq:factor_def}. Let $c \triangleq \frac{1}{2}(\Delta + \factor)$ and define $R \triangleq \frac{\factor}{2c}$ and $r \triangleq 1/2 - R$.

We will use the following Algorithm~\ref{alg:tensor} that we call \textsc{TensorResolve}. While this is only a slight modification of the tensor decomposition algorithm of \cite{huang2015super} for high-dimensional superresolution, our analysis is novel and obtains sharper results in low dimensions by using certain extremal functions \cite{gonccalves2018note,holt1996beurling,carneiro2017hilbert} arising in the study of de Branges spaces (see Theorem~\ref{thm:goncalves}.

\begin{algorithm}\caption{\textsc{TensorResolve}}\label{alg:tensor}
\begin{algorithmic}[1]
	\State \textbf{Input}: Error parameters $\epsilon_1,\epsilon_2$, confidence parameter $\delta$, access to $\eta$-approximate, $\Theta\left(\frac{k^2\log(k/\delta)}{\Min{(\Delta - \factor)}{1}}\right)$-query OTF oracle
	\State \textbf{Output}: With probability at least $1 - \delta$, an $(\epsilon_1,\epsilon_2)$-accurate estimate $(\{\tilde{\lambda}_i\}, \{\tilde{\vmu}_i\})$ for the parameters of $\rho$, provided the separation is sufficiently above the diffraction limit (see Lemma~\ref{lem:tensorresolve_correct})
		\State Define $R = \factor/2c$ and $r \triangleq 1/2 - R$.
		\State Sample $\omega^{(1)},...,\omega^{(m)}$ i.i.d. from the uniform distribution over $B^2(R)$. Also define $\omega^{(m+1)} = (1,0)$, $\omega^{(m+2)} = (0,1)$, and $\omega^{(m+3)} = (0,0)$. Define $m' \triangleq m + 3$
		\State Sample $v$ uniformly from $\S^1$ and define $v^{(1)}= r\cdot v$, $v^{(2)} = 2r\cdot v$, and $v^{(3)} = 0$. \label{step:vfreqs}
		\State Define $\xi_{a,b,i} \triangleq \omega^{(a)} + \omega^{(b)}) + v^{(i)}$ for every $a,b\in[m'], i\in[3]$. Query the OTF oracle at $\{\xi_{a,b,i}\}$ to obtain numbers $\{u_{a,b,i}\}$. Construct the tensor $\tilde{\vec{T}}\in\co^{m'\times m'\times 3}$ given by $\tilde{\vec{T}}_{a,b,i} = u_{a,b,i}/\widehat{A}[\xi_{a,b,i}]$.\label{step:constructtensor}
		\State Let $\hat{V}\in\R^{m'\times k}$ be the output of \textsc{Jennrich}($\tilde{\vec{T}}$) (defined in Algorithm~\ref{alg:jennrich}). Divide each column $\hat{V}^j$ by a factor of $\hat{V}_{m,j}$.\label{step:scaling}
		\State For each $j\in[k]$, $i\in[2]$, let $\hat{\mu}_j\in\R^2$ have $i$-th entry equal to the argument of the projection of $\hat{V}_{m+i,j}$ onto the complex disk. \label{step:standardbasis}
		\State Query the OTF oracle at frequencies $\{\omega^{(a)}\}_{a\in[m']}$ to get numbers $\{u'_a\}_{a\in[m']}$ and form the vector $\hat{b}\in\R^{m'}$ whose $a$-th entry is $u'_a/\widehat{A}[\omega^{(a)}]$ for every $a\in[m']$.
		\State Let $\hat{\lambda}\in\R^k = \argmin_{\lambda}\norm{\hat{V}\lambda - \hat{b}}_2$.\label{step:solvelambda}
		\State Return $(\hat{\lambda}_1,...,\hat{\lambda}_k)$ and $(\hat{\mu}_1,...,\hat{\mu}_k)$.
\end{algorithmic}
\end{algorithm}

Using the notation of \textsc{TensorResolve}, define the tensor $\vec{T}\in\co^{m'\times m'\times 3}$ given by \begin{equation}\vec{T}_{a,b,i} = \sum^k_{j = 1}\lambda_j e^{-2\pi i \langle \mu_j, \omega^{(a)} + \omega^{(b)} + v^{(i)}\rangle}\end{equation} and note that it admits a low-rank decomposition as \begin{equation}
	\vec{T} = \sum^k_{j=1}V^j\otimes V^j\otimes (W^j D_{\lambda}),
	\label{eq:Tdecomp}
\end{equation} where $D_{\lambda}$ is the diagonal matrix whose entries consist of the mixing weights $\{\lambda_j\}$ and, for every $j\in[k]$, $W^j = (e^{-2\pi i \langle \mu_j, v^{(1)}\rangle}, e^{-2\pi i \langle \mu_j, v^{(2)}\rangle}, e^{-2\pi i \langle \mu_j, v^{(3)}\rangle})$ and $V^j = (e^{-2\pi i \langle \mu_j, \omega^{(1)}\rangle}, \cdots ,e^{-2\pi i \langle \mu_j, \omega^{(m')}\rangle})$. Let $V\in\R^{m'\times k}$ denote the matrix whose $j$-th column is $V^j$.

Note that by our choice of $r,R$ and triangle inequality, we have that $\norm{\omega^{(a)} + \omega^{(b)} + v^{(i)}}_2 \le r + 2R = 1 - \frac{c - \factor}{2c} < 1$ for any entry index $a,b,i$. So if $\{u_{a,b,i}\}$ are the numbers obtained from an $\eta$-approximate, $(m+3)$-query OTF oracle as in Algorithm~\ref{alg:tensor}, and $\tilde{\vec{T}}$ is constructed as in Step~\ref{step:constructtensor} of \textsc{TensorResolve}, then by Lemma~\ref{lem:Fapproximate} we have that \begin{equation}
	|\vec{T}_{a,b,i} - \tilde{\vec{T}}_{a,b,i}| \le \frac{\eta}{\widehat{A}[1 - \frac{c - \factor}{2c}]} \le \eta\cdot\left(\frac{c - \factor}{2c}\right)^2,
\end{equation} where the last step follows by Fact~\ref{fact:ahatestimate}.

The following is a consequence of the stability of Jennrich's algorithm.

\begin{restatable}{lem}{tensor}[e.g. \cite{huang2015super}, Lemma 3.5]
	\label{lem:jennrich}
	For any $\epsilon,\delta > 0$, suppose $|\vec{T}_{a,b,i} - \tilde{\vec{T}}_{a,b,i}| \le \eta'$ for $\eta' \triangleq O\left(\frac{(c - \factor)\delta\Delta\lambda_{\min}^2}{k^{5/2} m^{3/2}\kappa(V)^5}\cdot \epsilon\right)$, and let $\hat{V}= \text{\textsc{Jennrich}}(\tilde{\vec{T}})$ (Algorithm~\ref{alg:jennrich}). Then with probability at least $1 - \delta$ over the randomness of $v^{(1)}$, there exists permutation matrix $\Pi$ such that $\norm{\widehat{V} - V\Pi}_F \le \epsilon$ for all $j\in[k]$.
\end{restatable}

The setting of parameters in \cite{huang2015super} is slightly different from ours, so we provide a self-contained proof of Lemma~\ref{lem:jennrich} in Appendix~\ref{app:jennrich}.

We will also need the following basic lemma about the stability of solving for $\hat{\lambda}$ in Step~\ref{step:solvelambda} in \textsc{TensorResolve}.

\begin{lem}
	For any $\epsilon,\epsilon'>0$, if $\lambda\in\R^k$ satisfies $V\lambda = b$ for some $V\in\R^{m'\times k}$ and $b\in\R^{m'}$, and furthermore $\hat{V},\hat{b}$ satisfy $\norm{V - \hat{V}}_2 \le \epsilon$ and $\norm{b - \hat{b}}_2\le \epsilon'$, then $\hat{\lambda}\triangleq \argmin_{\hat{\lambda}}\norm{\hat{V}\hat{\lambda} - \hat{b}}_2$ satisfies $\norm{\lambda - \hat{\lambda}}_2 \le \frac{2\epsilon\norm{\lambda}_2 + 2\epsilon'}{\sigma_{\min}(V) - \epsilon}$.
	\label{lem:solve_lam}
\end{lem}

\begin{proof}
	Note that \begin{equation}\norm{\hat{V}\lambda - \hat{b}}_2 \le \norm{(\hat{V} - V)\lambda}_2 +\norm{\hat{b} - b}_2 \le \epsilon \norm{\lambda}_2 + \epsilon'.\end{equation} By triangle inequality and definition of $\hat{\lambda}$, $\norm{\hat{V}(\hat{\lambda} - \lambda)}_2 \le 2\epsilon\norm{\lambda}_2 + 2\epsilon'$, so $\norm{\hat{\lambda}- \lambda}_2 \le \frac{2\epsilon\norm{\lambda}_2 + 2\epsilon'}{\sigma_{\min}(\hat{V})}$. The lemma follows because $\sigma_{\min}(V') \ge \sigma_{\min}(V) - \epsilon$.
\end{proof}

It remains to show the following condition number bound.

\begin{lem}
	For any $\delta>0$, if $m = \Theta\left(\frac{k^2\log(/\delta)}{\Min{(\Delta - \factor)}{1}}\right)$, then $\kappa(V) \le O\left(\Max{k}{\frac{k}{\sqrt{\Delta - \factor}}}\right)$ and $\sigma_{\min}(V) \ge \Omega\left(k^2\log(/\delta)\right)$ with probability at least $1 - \delta$.
	\label{lem:condition_number_bound}
\end{lem}

\begin{proof}
	Let $V^*\in\R^{m\times k}$ denote the submatrix given by the first $m$ rows of $V$. We will need the following basic lemma from \cite{huang2015super} relating the condition number of $V^*$ to that of $V$:

	\begin{lem}[\cite{huang2015super}, Lemma 3.8]
		$\kappa(V) \le \sqrt{2k}\cdot \kappa(V^*)$.
		\label{lem:triv_kappa}
	\end{lem}

	The primary technical component of this section is to upper bound $\kappa(V^*)$. First, note that given any $\lambda\in\co^{k-1}$, we have that \begin{equation}
		\lambda^{\dagger}{V^*}^{\dagger}V^*\lambda = \sum^m_{i=1}|\langle \lambda, V^*_i\rangle|^2 = \sum^m_{i=1}\left|\sum^k_{j=1}\lambda_j e^{-2\pi i \langle \mu_j, \omega^{(i)}\rangle}\right|^2.
		\end{equation}
	As each $\omega^{(i)}$ is an independent draw from the uniform distribution over $\S^1$, we have that 
	\begin{equation}
		\E_{\omega^{(1)},...,\omega^{(m)}}[\lambda^{\dagger}{V^*}^{\dagger}V^*\lambda] = m\int_{B^2(R)} \left|\sum^k_{j=1}\lambda_j e^{-2\pi i \langle \mu_j, \omega\rangle}\right|^2 \d\psi(\omega),
	\end{equation} where $\d\psi(\omega)$ is the uniform measure over $B^2(R)$. Furthermore, for any $\omega\in B^2(R)$ and $i\in[m]$, we have that \begin{equation}
		0\le |\langle\lambda, V^*_i\rangle|^2\le \norm{\lambda}^2_1 \le k\cdot\norm{\lambda}^2_2.
		\label{eq:lambda_cs}
	\end{equation}
	So by matrix Hoeffding applied to the random variables $m\cdot {V^*}^{\dagger}_1V^*_1,.\ldots, m\cdot {V^*}^{\dagger}_mV^*_m$, each of which is upper bounded in spectral norm by $m\cdot k$ based on \eqref{eq:lambda_cs}, we conclude that \begin{equation}
		\Pr\left[\norm{{V^*}^{\dagger}V^* - \E_{\omega^{(1)},\ldots,\omega^{(m)}}[{V^*}^{\dagger}V^*]}_2 > \sqrt{m}kt \right] \le k\cdot e^{-\Omega(t^2)} \ \forall \ t > 0.\label{eq:chernoff}
	\end{equation}

	Lemma~\ref{lem:use_minorant} below allows us to bound the quadratic form given by the expectation term evaluated at any $\lambda$. Taking $t = O(\sqrt{\log k/\delta})$ and $m = \Theta\left(\frac{k^2\log(k/\delta)}{\Min{(\Delta - \factor)}{1}}\right)$ in \eqref{eq:chernoff} and applying Lemma~\ref{lem:use_minorant}, we conclude that with probability at least $1 - \delta$, \begin{equation}
		\Omega(m)\cdot\{\Min{(\Delta - \factor)}{1}\}\cdot\norm{\lambda}^2_2 \le \lambda^{\dagger}{V^*}^{\dagger}V^*\lambda \le O(m)\cdot\left(k + \{\Min{(\Delta - \factor)}{1}\}\right)\cdot\norm{\lambda}^2_2,
	\end{equation}
	from which it follows that with this probability, $\kappa(V^*) \le O\left(\frac{k}{\Min{(\Delta - \factor)}{1}}\right)^{1/2}$, from which the lemma follows by Lemma~\ref{lem:triv_kappa}.
\end{proof}

It remains to show Lemma~\ref{lem:use_minorant} below, the key technical ingredient of this section. We will require the following special case of a result of \cite{gonccalves2018note}, which essentially follows from results of \cite{carneiro2017hilbert,holt1996beurling}. This can be thought of as the high-dimensional generalization of the well-known Beurling-Selberg minorant (see, e.g., \cite{vaaler1985some} for a discussion of the one-dimensional case).

\begin{thm}[\cite{gonccalves2018note}, Theorem 1]
	For any $d\in\N$ and $\frac{j_{d/2-1,1}}{\pi} < r < \frac{j_{d/2,1}}{\pi}$, there exists a function $M\in L^1(\R^d)$ whose Fourier transform is supported in $B^d(r)$, and which satisfies $M(x) \le \bone{x\in B^d(1)}$ for all $x\in\R^d$ and $\widehat{M}[0] = \frac{(2/r)^d}{|\S^{d-1}|}\cdot\frac{C(d,r)}{1 + C(d,r)/d}$, where $|\S^{d-1}|$ denotes the surface area of $\S^{d-1}$ and $C(r,d) \triangleq -\frac{\pi r J_{d/2-1}(\pi r)}{J_{d/2}(\pi r)} > 0$.
	\label{thm:goncalves}
\end{thm}

\begin{lem}
	\begin{equation}
		\Omega(\Min{(\Delta - \factor)}{1})\cdot \norm{\lambda}^2_2 \le \int_{B^2(R)} \left|\sum^k_{j=1}\lambda_j e^{-2\pi i \langle \mu_j, \omega\rangle}\right|^2 \d\psi(\omega) \le k\norm{\lambda}^2_2
		\label{eq:hermitian}
	\end{equation} where $\d\psi(\omega)$ denotes the uniform probability measure over $B^2(R)$ for $R = \frac{\factor}{2\Delta}$.\footnote{In fact, one can improve the upper bound in \eqref{eq:hermitian} by using a suitable majorant for the indicator of the ball, but because we are only after polynomial time and sample complexity, this is not needed.}
	\label{lem:use_minorant}
\end{lem}

\begin{proof}
	The upper bound follows by \eqref{eq:lambda_cs}. We now show the lower bound. By Theorem~\ref{thm:goncalves} applied to $d = 2$, for any $\factor/2 < r < \frac{j_{1,1}}{\pi}$ there is a function $M$ which minorizes the indicator function of $B^2(1)$ and has Fourier transform supported in $B^2(r)$. Take $r = \{\Min{\Delta R}{\frac{\factor/2 + j_{1,1}/\pi}{2}}\}$ which satisfies $\factor/2 < r < \frac{j_{1,1}}{\pi}$. This implies that the function $M'(\omega) \triangleq \frac{1}{\pi R^2}\cdot \frac{1}{R}\cdot M(\omega/R)$ minorizes the density $\psi(\omega)$, has Fourier transform supported in $B^2(r)\subseteq B^2(\Delta)$, and satisfies \begin{equation}
		\widehat{M'}[0] = \frac{1}{\pi R^2}\frac{(2/r)^2}{|\S^1|}\cdot\frac{C(2,r)}{1 + C(2,r)/2} = \frac{4C(2,r)}{\pi^2 r^3 R^2\cdot (2 + C(2,r))} \ge \frac{r - \factor/2}{R^2} \ge 4r - 2\factor,
		\label{eq:Mzero}
	\end{equation} where in the last step we used that $R< 1/2$.
	 We can lower bound \eqref{eq:hermitian} by \begin{align}
		\int \left|\sum^k_{j=1}\lambda_j e^{-2\pi i \langle \mu_j, \omega\rangle}\right|^2 \cdot M'(\omega) \d\omega &= \sum^k_{j,j'=1}\lambda_j\lambda^{\dagger}_{j'} \int e^{-2\pi i \langle \mu_j - \mu_{j'}, \omega\rangle} \cdot M'(\omega) \d\omega \\
		&= \sum^k_{j,j'=1}\lambda_j\lambda^{\dagger}_{j'}\widehat{M'}[\mu_j - \mu_{j'}] \ge (4r - 2\factor)\norm{\lambda}^2_2,
	\end{align} where the last step follows by \eqref{eq:Mzero} and the fact that $\widehat{M'}[\mu_j - \mu_{j'}] = 0$ for all $j\neq j'$. The lemma follows from noting that $4r - 2\factor > \Min{\{2\factor(\Delta/c- 1)\}}{\left\{\frac{2j_{1,1}}{\pi} - \factor\right\}} \ge O(\Min{\Delta - \factor}{1})$.
\end{proof}

Putting everything together, we have the following guarantee:

\begin{lem}
	Let $\rho$ be a $\Delta$-separated superposition of $k$ Airy disks. For any $\epsilon_1,\epsilon_2,\delta > 0$, let \begin{equation}
			m = \Theta\left(\frac{k^2\log(k/\delta)}{\Min{(\Delta - \factor)}{1}}\right) \qquad \text{and} \qquad \eta = O\left(\frac{4\Delta^3\delta\lambda_{\min}^2}{(\Delta - \factor)k^{5/2} m^{3/2}\kappa(V)^5}\cdot \epsilon_1\right).\label{eq:eta_final2}
		\end{equation}
	Without loss of generality suppose $\epsilon_1 < 1/6$. Then the output $(\lambda^*_1,\lambda^*_2,\vmu^*_1,\vmu^*_2)$ of \textsc{TensorResolve}, given $\epsilon_1,\epsilon_2,\delta$ and access to an $\eta$-approximate, $m$-query OTF oracle $\O$ for $\rho$, satisfies \begin{equation*}\norm{\vmu_i - \vmu^*_{\tau(i)}}_2\le\epsilon_1 \; \; \text{and} \; \; |\lambda_i - \lambda_{\tau(i)}|\le\epsilon_2\end{equation*} for some permutation $\tau$ with probability at least $1 - \delta$. Furthermore, the runtime of \textsc{LearnAiryDisks} is polynomial in $k$, the number of OTF oracle queries, and the time it takes to make those queries.\label{lem:tensorresolve_correct}
\end{lem}

\begin{proof}
	By Lemma~\ref{lem:jennrich}, if we take $m = \Theta\left(\frac{k^2\log(k/\delta)}{\Min{(\Delta - \factor)}{1}}\right)$ and $\eta' = O\left(\frac{(c - \factor)\delta\Delta\lambda_{\min}^2}{k^{5/2} m^{3/2}\kappa(V)^5}\cdot \epsilon_1\right)$, then the output $\hat{V}$ of \textsc{Jennrich}{($\tilde{\vec{T}}$)} satisfies $\norm{\hat{V} - V\Pi}_F \le \epsilon_1$ for some permutation matrix $\Pi$. Assume without loss of generality that $\Pi = \Id$. Then we get that for all $j\in[k]$ and $\ell\in[m']$, \begin{equation}|\hat{V}_{\ell,j} - V_{\ell,j}| = \left|e^{-2\pi i \langle \hat{\mu}_j - \mu_j, \omega^{(\ell)}\rangle} - 1\right| \le \epsilon_1,\end{equation} and because of the elementary inequality $|e^{-2\pi i x} - 1| \ge 2|x|$ for any $|x|\le 2/3$ and the fact that \begin{equation}\langle \hat{\mu}_j - \mu_j, \omega^{(\ell)}\rangle \le \norm{\hat{\mu}_j - \mu_j}_2 \norm{\omega^{(\ell)}}_2 \le 2\radius \le 2/3,\label{eq:radius}
	\end{equation} we conclude that $|\langle \hat{\mu}_j - \mu_j, \omega^{(\ell)}\rangle| \le \epsilon_1/2$ for all $j\in[k]$, $\ell\in[m']$. In particular, this holds for all $\ell = m+1$ and $\ell = m+2$, so $\norm{\hat{\mu}_j - \mu_j}_{\infty} \le \epsilon_1$. By dividing $\epsilon_1$ by $\sqrt{2}$ and absorbing constants, we get that the estimates $\{\hat{\mu}_j\}$ for the centers are $\epsilon_1$-close to the true centers.

	To show that the mixing weights are $\epsilon_2$-close to the true mixing weights, we can apply Lemma~\ref{lem:solve_lam} to conclude that \begin{equation}
		\norm{\lambda - \hat{\lambda}}_2 \le O\left(\frac{\epsilon_1 + \eta'}{k^2\log (1/\delta) - \epsilon_1}\right) = O\left(\frac{\epsilon_1}{k^2\log(k/\delta)}\right),
	\end{equation} so, possibly by modifying $\epsilon_1$ to be $\frac{\epsilon_2}{k^2\log (1/\delta)}$, we get that the estimates $\{\hat{\lambda}_j\}$ for the mixing weights are $\epsilon_2$-close to the true mixing weights.
\end{proof}

Note that we can also amplify the success probability of \textsc{TensorResolve} by running \textsc{Select} from Section~\ref{subsec:learnwithotf}, but we do not belabor this point here.

%% file: kde2.tex

In this section, we show that the following algorithm \textsc{DFT} is a valid implementation of an approximate OTF oracle. We begin by showing that when the samples have granularity $\varsigma=0$, \textsc{DFT} can achieve arbitrarily small error with polynomially many samples.

\begin{algorithm}\caption{\textsc{DFT}}\label{alg:kde}
\begin{algorithmic}[1]
	\State \textbf{Input}: Error tolerance $\eta > 0$, sample access to $\rho$, confidence parameter $\beta>0$, frequencies $\omega_1,...,\omega_m$
	\State \textbf{Output}: With probability at least $1 - \beta$, numbers $u_1,...,u_m$ such that for each $j\in[m]$, $|u_j - \widehat{\rho}[\omega_j]| \le \eta$
		\State $N\gets O(\log(m/\beta)/\eta^2)$.
		\State Draw samples $\vx_1,...,\vx_N$ from $\rho$.
		\State For each $j\in[m]$, compute the average $u_j \gets \frac{1}{N}\sum^N_{i=1}\cos(2\pi\cdot\langle \omega_j, \vx_i\rangle)$.
		\State Output $u_1,...,u_m$.
\end{algorithmic}
\end{algorithm}

\begin{lem}
	For any $0 < \beta < 1$, $\eta > 0$, and frequencies $\omega_1,...,\omega_m\in\R^2$, \textsc{DFT}($\{\omega_i\}_{i\in[m]}$) draws $N =O(\log(m/\beta)/\eta^2)$ samples and in time $T = O(N\cdot m)$ outputs numbers $u_1,...,u_m$ for which $|u_j - \widehat{\rho}[\omega_j]| \le \eta$.
	\label{lem:kde}
\end{lem}

\begin{proof}
	By a union bound, it suffices to show that for any single $j\in[m]$, $|u_j - \widehat{\rho}[\omega_j]| \le \eta$ with probability at least $1 - \beta/m$. Note that \begin{equation}\E[u_j] = \E_{\vec{x}\sim\rho}[\cos(2\pi\cdot \langle \omega_j, \vec{x}\rangle)] = \E\left[\Re\, \widehat{\rho}[\omega_j]\right] = \widehat{\rho}[\omega_j],\end{equation} where the last step follows by the fact that $\widehat{\rho}$ is real-valued (by circular symmetry of $A$). Furthermore, the summands in $\sum^N_{i=1}\cos(2\pi\cdot\langle \omega_j, \vx_i\rangle)$ are $[-1,1]$-valued, so by Chernoff, \begin{equation}\Pr\left[|u_j - \E[u_j]| > \eta\right] \le \exp(-\Omega(N\eta^2)),\end{equation} from which the lemma follows by our choice of $N$.
\end{proof}

We now show that for general granularity $\varsigma>0$, the output of \textsc{DFT} still achieves error $\eta + O(\varsigma)$.

\begin{cor}
	For any $0 < \beta < 1$, $\eta,\varsigma > 0$, and frequencies $\omega_1,...,\omega_m\in\R^2$, if \textsc{DFT}($\{\omega_i\}_{i\in[m]}$) draws $N =O(\log(m/\beta)/\eta^2)$ samples of granularity $\varsigma$, then in time $T = O(N\cdot m)$ it outputs numbers $u_1,...,u_m$ for which $|u_j - \widehat{\rho}[\omega_j]| \le \eta + O(\varsigma\cdot\norm{\omega_j}_2)$.
	\label{cor:granularity_kde}
\end{cor}

\begin{proof}
	Note that $\cos(\cdot)$ is $\alpha$-Lipschitz for some $\alpha < 3/4$. This implies that for any $\omega\in\R^2$, the function $\vx\mapsto\cos(2\pi\langle \vx,\omega\rangle)$ is at most $O(\norm{\omega}_2)$-Lipschitz with respect to $\ell_2$.

	Take any collection of 0-granular samples $\vx'_1,...,\vx'_N$ for which the averages $u'_1,...,u'_m$ computed by \textsc{DFT} would be $\eta$-accurate. If \textsc{DFT} were instead passed $\varsigma$-granular samples $\vx_1,...,\vx_N$ for which $\norm{\vx'_i - \vx_i}_2 \le \varsigma$ for each $i\in[N]$, then by triangle inequality, the averages $u_1,...,u_m$ computed by \textsc{DFT} with these samples would satisfy $|u_j - u'_j| \le \eta + O(\varsigma\cdot \norm{\omega_j}_2)$ for each $j\in[m]$, as claimed.
\end{proof}

Finally, with Lemma~\ref{lem:learnairydisks} and Lemma~\ref{lem:kde}, we can complete the proof of Theorem~\ref{thm:main}.

\begin{proof}[Proof of Theorem~\ref{thm:main}]
	By Lemma~\ref{lem:learnairydisks}, it suffices to produce an $\eta$-approximate, $m$-query OTF oracle for $\eta$ defined in \eqref{eq:eta_final} and $m = O(\log 1/\delta)$. By Corollary~\ref{cor:granularity_kde}, this can be done using \begin{equation}\log(m/\delta)/\eta^2 = \widetilde{O}\left(\log(1/\delta)\cdot \poly(1/\lambda_{\min},1/\epsilon_1,1/\epsilon_2,(4k/\Delta)^{k^2})\right)\end{equation} samples of granularity $\eta/2$ with probability at least $1 - \delta$. Theorem~\ref{thm:main} then follows by a union bound over the failure probabilities of \textsc{LearnAiryDisks} and \textsc{DFT}, and replacing $2\delta$ with $\delta$ and absorbing constant factors. Finally, note that the dependence on $\radius$ follows by the discussion at the end of Section~\ref{subsec:reduce}.
\end{proof}

\begin{proof}[Proof of Theorem~\ref{thm:above}]
	By Lemma~\ref{lem:tensorresolve_correct}, it suffices to produce an $\eta$-approximate, $m$-query OTF oracle for $\eta$ defined in \eqref{eq:eta_final2} and $m = \Theta\left(\frac{k^2}{\Min{(\Delta - \factor)}{1}}\right)$. By Corollary~\ref{cor:granularity_kde}, this can be done with probability $9/10$ using \begin{equation}\log(10 m)/\eta^2 = \widetilde{O}\left(\poly(k,1/\Delta,1/\lambda_{\min},1/\epsilon_1,1/\epsilon_2,k,\Min{(\Delta-\factor)}{1})\right)\end{equation} samples of granularity $\eta/2$ with probability at least $1 - \delta$. Theorem~\ref{thm:main} then follows by a union bound over the failure probabilities of \textsc{TensorResolveCorrect} and \textsc{DFT}. As in the proof of Theorem~\ref{thm:main}, the dependence on $\radius$ follows by the discussion at the end of Section~\ref{subsec:reduce}.
\end{proof}

%% file: stoc_lowerbound.tex

\section{Information Theoretic Lower Bound}
\label{sec:lowerbound}

In this section we will exhibit two superpositions of Airy disks, both with minimum separation below the diffraction limit, which are close in statistical distance. Let $\rho$ and $\rho'$ respectively have mixing weights $\{\lambda_i\}$ and $\{\lambda'_i\}$, and centers $\{\vmu_i\}$ and $\{\vmu'_i\}$, where for each $i$, $\vmu_i \triangleq (a_i,b_i)$ and $\vmu'_i \triangleq (a'_i,b'_i)$ for some $a_i,a'_i,b_i,b'_i\in\R$. Concretely, \begin{equation}
	\rho(\vx) = \sum^{\lceil k/2\rceil}_{i=1}\lambda_i\cdot A\left(\frac{\norm{\vx - \vmu_i}}{\sigma}\right) \; \; \; \text{and} \; \; \; \rho'(\vx) = \sum^{\lfloor k/2\rfloor}_{i=1}\lambda'_i\cdot A\left(\frac{\norm{\vx - \vmu'_i}}{\sigma}\right)
\end{equation} for some $0 < \sigma < 1$. Note that under this setting of parameters, the Abbe limit corresponds to separation $\pi\sigma$.

\begin{thm}
	Let $\lbound\triangleq \sqrt{4/3}$. There exists a choice of $\{\vmu_i\}$, $\{\vmu'_i\}$, $\{\lambda_i\}$, $\{\lambda'_i\}$ such that the minimum separation among centers of $\rho$ and among centers of $\rho'$ is $\Delta = (1-\epsilon)\lbound\pi\sigma$, and $\tvd(\rho,\rho') \le \exp(-\Omega(\epsilon\sqrt{k}))$.\label{thm:mainlowerbound}
\end{thm}

We first bound $\norm{\rho - \rho'}_{L^2}$. By Plancherel's,
\begin{align}\norm{\rho - \rho'}^2_{L^2} &= \norm{\widehat{\rho} - \widehat{\rho'}}^2_{L^2}\nonumber \\
&= \sigma^2\int_{\R^2}\widehat{A}[\sigma\omega]^2\abs*{\sum_i \lambda_i e^{-2\pi i\langle \vmu_i,\omega\rangle} - \sum_i \lambda'_i e^{-2\pi i\langle\vmu'_i,\omega\rangle}}^2 d\omega \nonumber\\
&\le \sigma^2\int_{B_{1/\pi\sigma}(0)}(1 - \norm{\pi\sigma\omega})^2\cdot\abs*{\sum_i \lambda_i e^{-2\pi i\langle \vmu_i,\omega\rangle} - \sum_i \lambda'_i e^{-2\pi i\langle\vmu'_i,\omega\rangle}}^2 d\omega  \label{eq:l2align}
\end{align} where the inequality follows by the elementary bound \begin{equation*}
	\frac{2}{\pi}(\arccos(\pi r) - \pi r\sqrt{1-\pi^2 r^2}) \le 1 - \pi r.
\end{equation*}

Recall now the construction in Lemma~\ref{lem:l2bound} (see the beginning of Section~\ref{sec:lowerboundprev}). As the entries of the vector $u$ constructed in Lemma~\ref{lem:l2bound} satisfy $\sgn(u_{j_1,j_2}) = (-1)^{j_1,j_2}$, let $I_0$ (resp. $I_1$) denote the elements $\vec{j} = (j_1,j_2)\in\calJ\times\calJ$ for which $j_1 + j_2$ is even (resp. odd), and for every $\vec{j}\in I_0$ (resp. $\vec{j} \in I_1$), define $\lambda_{\vec{j}}$ and $\vmu_{\vec{j}}$ (resp. $\lambda'_{\vec{j}}$ and $\vmu'_{\vec{j}}$) by $u_{\vec{j}}$ and $\left(\frac{j_1}{m},\frac{\sqrt{3}j_2}{m}\right)$. This construction is illustrated for $k = 25$ in Figure~\ref{fig:lattice}.
By design, $\{\lambda_{\vec{j}}\}_{\vec{j}\in I_0}$ and $\{\lambda'_{\vec{j}}\}_{\vec{j}\in I_1}$ consist solely of nonnegative scalars and respectively sum to 1, so $\rho,\rho'$ are valid superpositions of Airy disks. Furthermore, by design, \begin{equation*}
	\sum u_{j_1,j_2} e^{-2\pi i\cdot (j_1x_1 + \sqrt{3}j_2x_2)/m} = \sum_{\vec{j} \in I_0} \lambda_{\vec{j}} e^{-2\pi i\cdot \iprod{\vmu_{\vec{j}},x}} - \sum_{\vec{j} \in I_1} \lambda'_{\vec{j}} e^{-2\pi i\cdot \iprod{\vmu'_{\vec{j}},x}},
\end{equation*} so we may now bound \eqref{eq:l2align} to get the desired $L_2$ bound \begin{equation}
	\norm{\rho - \rho'}^2_{L^2} \le \exp(-\Omega(\epsilon\sqrt{k}))\sigma^2\int_{B_{1/\pi\sigma}(0)}(1 - \norm{\pi \sigma\omega})^2 = \exp(-\Omega(\epsilon\sqrt{k}))/6\pi = \exp(-\Omega(\epsilon\sqrt{k})).
\end{equation}
We are now ready to show that $\tvd(\rho,\rho')$ is small. The following is a generic $L^1$ bound for functions whose univariate restrictions have bounded $L^2$ mass, whose derivatives inside some region $\Omega$ are bounded, and which decay sufficiently quickly outside of $\Omega$.

\begin{lem}
	Suppose for an $f\in L^1(\R^2)$, there exists some $T\ge 0$ such that for $\Omega = [-T,T]^2$ the following are satisfied: \begin{enumerate}
		\item For all $y\in[-T,T]$, $\max_{x\in[-T,T]}|f'(x,y)|\le C$,
		\item $\int_{\Omega^c}|f| \le \eta$,
		\item $f(-T,y) \le \delta$ for all $y\in[-T,T]$.
	\end{enumerate} Then we have that \begin{equation*}\norm{f}_{L^1}\le (2T)^{5/3}\cdot(3C\norm{f}^2_{L^2} + 2T\delta^3)^{1/3} + \eta.\end{equation*}\label{lem:genericl1bound}
\end{lem}

\begin{proof}
	By the triangle inequality and condition 3, it is enough to verify that \begin{equation*}\int_{\Omega}|f|\le (2T)^{5/3}\cdot(3C\norm{f}^2_{L^2} + 2T\delta^3)^{1/3}.\end{equation*} Note that for a fixed $y\in[-T,T]$, we have by the fundamental theorem of calculus and conditions 2 and 3 that for any $x\in[-T,T]$, \begin{equation}\frac{1}{3}|f(x,y)^3| \le \frac{1}{3}|f(-T,y)^3| + \left(\int^x_{-T}f(t,y)^2 dt\right)\cdot\max_{t\in[-T,x]}|f'(t,y)| \le \frac{1}{3}\delta^3 + C\int^T_{-T}f(t,y)^2 dt\label{eq:calculus}.\end{equation} Define $g(y)\triangleq \int^T_{-T}f(t,y)^2 dt$ and note that $\int^T_{-T}g(y)dy\le\norm{f}^2_{L^2}$. Then \begin{align*}
		\int_{\Omega}|f| &= \int^T_{-T}\int^T_{-T}|f(x,y)|dx\,dy \\
		&\le \int^T_{-T}\int^T_{-T}(3C\cdot g(y) + \delta^3)^{1/3} dx\,dy\\
		&= 2T\int^T_{-T}(3C\cdot g(y) + \delta^3)^{1/3} dy \\
		&\le (2T)^2\left(\int^T_{-T}\frac{1}{2T}(3C\cdot g(y) + \delta^3)dy\right)^{1/3} \\
		&= (2T)^{5/3}\left(\int^T_{-T}(3C\cdot g(y) + \delta^3)dy\right)^{1/3} \\
		&\le (2T)^{5/3}\cdot(3C\norm{f}^2_{L^2} + 2T\delta^3)^{1/3},
	\end{align*} where the penultimate inequality follows from the measure-theoretic generalization of Jensen's inequality.
\end{proof}

We will show that for an appropriate choice of $T$, the function $f:\R^2\to\R$ given by \begin{equation} f(\vx) \triangleq \rho(\vx) - \rho'(\vx) =  \sum_{\vec{j} \in J_0}\lambda_{\vec{j}}\cdot A\left(\frac{\norm{\vx - \vmu_{\vec{j}}}}{\sigma}\right) - \sum_{\vec{j}}\lambda'_{\vec{j}}\cdot A\left(\frac{\norm{\vx - \vmu'_{\vec{j}}}}{\sigma}\right) \label{eq:fdef}\end{equation} satisfies the conditions of Lemma~\ref{lem:genericl1bound}.

For $\vec{a} = (a_1,a_2)\in\R^2,y\in\R$, define $A_{\vec{a},y}:\R\to\R$ by \begin{equation*}
	A_{\vec{a},y}(x) = A\left(\frac{\sqrt{(x - a_1)^2 + (y - a_2)^2}}{\sigma}\right).
\end{equation*} As we will see below, by triangle inequality it will suffice to verify certain properties of $A_{\vec{a},y}$.

\begin{lem}
	For any $y\in\R$, $\max_{x\in\R}|f'(x,y)| = O(1/\sigma)$.\label{lem:derivbound}
\end{lem}

\begin{proof}
	By triangle inequality, it suffices to show that for any $a,y\in\R$, $\max_{x\in\R}\left|\frac{\partial A_{\vec{a},y}(x)}{\partial x}\right| = O(1/\sigma)$. Of course we may as well assume $a_1 = 0$, in which case by a change of variable in $x$, we have that \begin{align*}
		\MoveEqLeft \max_{x\in\R}\left|\frac{\partial A_{\vec{a},y}(x)}{\partial x}\right| \\
		&= \frac{1}{\pi\sigma}\max_x\left|\frac{\partial}{\partial x}\frac{J_1(\sqrt{x^2 + ((y - a_2)/\sigma)^2})^2}{x^2 + ((y-a_2)/\sigma)^2}\right| \\
		&= \frac{1}{\pi\sigma}\max_x\left|\frac{2x J_1(\sqrt{x^2 + ((y-a_2)/\sigma)^2})\cdot J_2(\sqrt{x^2 + ((y-a_2)/\sigma)^2})}{(x^2 + ((y-a_2)/\sigma)^2)^{3/2}}\right| \\
		&\le \frac{1}{\pi\sigma}\max_x\left|\frac{2\sqrt{x^2 + ((y-a_2)/\sigma)^2} J_1(\sqrt{x^2 + ((y-a_2)/\sigma)^2})\cdot J_2(\sqrt{x^2 + ((y-a_2)/\sigma)^2})}{(x^2 + ((y-a_2)/\sigma)^2)^{3/2}}\right| \\
		&= \frac{1}{\pi\sigma}\max_z\left|\frac{2J_1(z)J_2(z)}{z^2}\right| \\
		&= \frac{1}{\pi\sigma}\max_z\left|\frac{\partial A(z)}{\partial z}\right| \\
		&\le O(1/\sigma)
\	\end{align*} as claimed.
\end{proof}

\begin{lem}
	For $T > \Delta(\sqrt{k}-1)/4$, we have that $f(-T,y)\le \Omega\left((T/\sigma)^{-8/3}\right)$ for all $y\in[-T,T]$.\label{lem:boundarybound}
\end{lem}

\begin{proof}
	By linearity, it suffices to show that for any $y\in[-T,T]$ and any $j_1,j_2\in\calJ$, the claimed bound holds for $A_{\nu_{j_1,j_2},y}(-T)$. By Theorem~\ref{thm:besselbound}, we know that \begin{equation*}
		A_{\nu_{j_1,j_2},y}(-T) \le \frac{1}{\pi}\Cr{landaubessel}^2\cdot |r|^{-8/3},
	\end{equation*} where \begin{equation*}
		r \triangleq \frac{1}{\sigma}\left(\left(-T - \frac{j_1}{m}\right)^2 + \left(y - \frac{\sqrt{3}j_2}{m}\right)^2\right)^{1/2}\ge \frac{-T - j_1/m}{\sigma} > -2T/\sigma,
	\end{equation*} where in the last step we used the fact that $j_1/m \le \Delta(\sqrt{k}-1)/4 < T$.
\end{proof}

\begin{lem}
	For $T > \Delta(\sqrt{k}-1)/2$, we have that $\int_{\Omega^c}|f| \le O(T^{-2/3}\sigma^{8/3})$, where $\Omega = [-T,T]^2$.\label{lem:tails}
\end{lem}

\begin{proof}
	By linearity and the fact that $\norm{\vx - (j_1/m,\sqrt{3}j_2/m)}_2 \ge T - j_1/m \ge T/2$ for every $\vx\not\in\Omega$, it suffices to show that for any $j_1,j_2\in\calJ$, the claimed bound holds for $\int_{B_0(T/2)^c}|A(x,y)| dx \, dy$. Expressing this as a polar integral, we have \begin{align*}
		\int_{\Omega^c}\abs*{A_{\nu_{j_1,j_2},y}(x)} dx \, dy &= \int^{2\pi}_0 \int^{\infty}_{T/2} r\cdot |A(r/\sigma)| dr \, d\theta \\
		&\le 2\cdot\int^{\infty}_{T/2} r\cdot\left(\Cr{landaubessel}^2\cdot (r/\sigma)^{-8/3}\right) \\
		&\le O(T^{-2/3}\sigma^{8/3}),
	\end{align*} as desired.
\end{proof}

\begin{proof}[Proof of Theorem~\ref{thm:mainlowerbound}]
	Take $T = \Theta(\norm{f}^{-1/5}_{L^2}) \ge \exp(-\Omega(\epsilon\sqrt{k}))$ which we may assume without loss of generality, by scaling $\sigma,\Delta$ appropriately, is greater than $\Delta\sqrt{k}$. By Lemma~\ref{lem:genericl1bound} and Lemmas~\ref{lem:derivbound}, \ref{lem:boundarybound}, and \ref{lem:tails}, we have that for $f$ defined by \eqref{eq:fdef}, \begin{align*}
		\int_{\R^2}|f| &\le (2T)^{5/3}\cdot\left(O(1/\sigma)\cdot\norm{f}^2_{L^2} + O(T\cdot (T/\sigma)^{-8})\right)^{1/3} + O\left(T^{-2/3}\sigma^{8/3}\right) \\
		&\le O\left(\norm{f}^{2/15}_{L^2}\sigma^{-1/3}\right) \\
		&\le O\left(\exp\left({-\Omega(\epsilon\sqrt{k})}\right)\sigma^{-1/3}\right),
	\end{align*} so as soon as $k\ge C\log(1/\sigma))$ for sufficiently large $C>0$, we have that $\tvd(\rho,\rho') \le \exp(-\Omega(\epsilon\sqrt{k}))$.
\end{proof}



%% file: science.tex

\section{Related Work In the Sciences}
\label{sec:science}

In this section, we survey previous approaches to understanding diffraction limits in the optics literature, as well as recent practical works on the need and methodologies to rigorously assess claims of achieving super-resolution.

\subsection{Previous Approaches in Optics}
\label{sec:comparison}

In this section we will survey the many previous attempts to rigorously understand diffraction limits in the optics literature. There, the focus has been squarely on the semiclassical detection model (SDM). After describing this line of work, we explain the ways in which it falls short.

The SDM was originally proposed by \cite{mandel1959fluctuations} and has been the \emph{de facto} generative model in essentially all subsequent works on the statistical foundations of resolution. We note that there are some minor differences in the definition of our model and that of the SDM, which we will discuss formally in Appendix~\ref{subsec:compare_models}.

Arguably the first significant work to study the SDM was that of Helstrom~\cite{helstrom1964detection}, who considered it from the perspective of parameter estimation and hypothesis testing, initiating the study of the following two problems which remarkably have almost exclusively occupied this line of work. For normalized point spread function $A(\cdot)$ and separation parameter $d$, define \begin{equation}
	\rho_0(\vec{x}) = A(\vec{x}),\ \ \rho_1(\vec{x}) = \frac{1}{2}\cdot A(\vec{x} - \vmu) + \frac{1}{2}\cdot A(\vec{x} + \vmu), \ \ \vmu = \begin{cases}
		d/2 & D = 1 \\
		(0,d/2) & D = 2.
	\end{cases}\label{eq:old_stat_setup}
\end{equation}

\begin{problem}[Parameter Estimation]
	Given samples from $\rho_1$, estimate $d$.\label{problem:paramestimate}
\end{problem}

\begin{problem}[Hypothesis Testing]
	Suppose we know the parameter $d$, and we know that either $\rho = \rho_0$ or $\rho = \rho_1$. Given samples from $\rho$, decide whether $\rho = \rho_0$ or $\rho = \rho_1$.\label{problem:hypotest}
\end{problem}

For Problem~\ref{problem:paramestimate}, Helstrom \cite{helstrom1964detection,helstrom1969detection,helstrom1970resolvability} studied the maximum likelihood estimator and computed Cramer-Rao lower bounds for a host of point-spread functions including the Airy PSF, both for the SDM and for progressively more physically sophisticated (though less practically relevant) models. The conceptual insights and problem formulation of \cite{helstrom1964detection} were refined, or often rediscovered, numerous times \cite{tsai1979performance,bettens1999model,van2002high,shahram2004imaging,shahram2006statistical,ram2006beyond,farrell1966information,chao2016fisher}, and the primary thrust of this line of work has been centered on Cramer-Rao-style calculations for assorted point-spread functions and, to a lesser extent, analysis of the optimization landscape of the log-likelihood from the perspective of singularity theory \cite{van2001resolution,van2001resolutionreconsider,bettens1999model,den1996model}.

For Problem~\ref{problem:hypotest}, Helstrom~\cite{helstrom1964detection} computed the reliability of the likelihood ratio test for various PSFs, under a CLT appoximation to the log-likelihood ratio. Similar calculations for the log-likelihood ratio for other PSFs followed in \cite{harris1964resolving,acuna1997statistical,shahram2004imaging,shahram2006statistical,farrell1966information}.

We emphasize that, with the exception of \cite{shahram2004imaging,shahram2006statistical}, all works giving rigorous guarantees have made the assumption implicit in \eqref{eq:old_stat_setup} that the two point sources defining $\rho_1$ are located at \emph{known} points $\vmu$ and $-\vmu$ centered about the origin. \cite{shahram2004imaging,shahram2006statistical} study Problems~\ref{problem:paramestimate} and \ref{problem:hypotest} when the locations of the point sources are unknown and study the (locally optimal) generalized likelihood ratio test.

With regards to applications, Problems~\ref{problem:paramestimate} and \ref{problem:hypotest} have gained popularity in optical astronomy \cite{falconi1967limits,zmuidzinas2003cramer,feigelson2012modern,lucy1992resolution,lucy1992statistical} as well as fluorescence microscopy \cite{mortensen2010optimized,small2014fluorophore,deschout2014precisely,von2017three}. Cramer-Rao bounds as a ``modern'' proxy for assessing the limits of imaging systems have gained such popularity with practicioners that a number of review articles and surveys on the topic have appeared in the recent single-molecule microscopy literature \cite{small2014fluorophore,deschout2014precisely,chao2016fisher}, most of which focus on the related parameter estimation problem of \emph{localization}, that is, estimating the location of a \emph{single} test object given its noisy image.

One other interesting line of work has focused on the generalizations of Problems~\ref{problem:paramestimate} and \ref{problem:hypotest} to the quantum setting. Elaborating on this literature would take us too far afield, so we mention only the comprehensive recent survey \cite{tsang2019resolving} and the references therein.

\subsection{Comparison with Our Approach}

Most crucially, all works on the SDM focus exclusively on \emph{two}-point resolution. In the context of hypothesis testing, as we note above, these works even assume the two points lie on the $x$-axis at the same \emph{known} distance $d/2$ from the origin, with the exception of \cite{shahram2004imaging,shahram2006statistical}. That such a strong assumption is made and such focus is placed on $k = 2$ is evidently not just for aesthetics. From the standpoint of hypothesis testing, as noted in \cite{shahram2004imaging,shahram2006statistical}, any deviation from this idealized model would induce a composite hypothesis testing problem, for which the (generalized) likelihood ratio test has no global optimality guarantees.
In the context of parameter estimation, because of the focus on $k = 2$, the conclusion in the literature has repeatedly been that the classical resolution criteria (Abbe, Rayleigh, etc.) are not meaningful in a statistical sense, and that the only true limitation comes from the number of samples. We view this as one of the primary reasons that a result like Theorem~\ref{thm:lowerbound_informal} has gone overlooked for so long.

Another drawback of the literature is that because of the focus on Cramer-Rao bounds, which only provide guarantees for the maximum likelihood estimate in the infinite-sample limit, none of these works actually give non-asymptotic algorithmic guarantees. Additionally, Cramer-Rao bounds only apply to unbiased estimators, and to the best of our knowledge, the only paper that addresses biased estimators is \cite{tsang2018conservative}, which only derives Bayesian Cramer-Rao bounds for the already well-studied setting of a mixture of two Gaussians.
From a technical standpoint, another disadvantage of existing works is that they work either with the Gaussian point-spread function or invoke Taylor approximations of the Airy point-spread function. And because the log-likelihood here is analytically cumbersome, it is common to invoke a central limit theorem-style approximation.

One last shortcoming arises from the definition of the SDM itself (see Definition~\ref{defn:goodman}): it models photon detection as a Poisson process when in reality this need not be the case. As Goodman (Chapter 9.2 of \cite{goodman2015statistical}) notes, ``in most problems of real interest, however, the light wave incident on the photosurface has stochastic attributes \dots For this reason, it is necessary to regard the Poisson distribution as a \emph{conditional} probability distribution \dots the statistics are in general \emph{not} Poisson when the classical intensity has random fluctuations of its own.'' The increased generality of not assuming Poissonanity allows our model to smoothly handle such stochastic fluctuations.

\subsection{Super-Resolution and the Practical Need to Understand Diffraction Limits}

In the past half century, a host of techniques of increasing sophistication have been developed to shift or fundamentally surpass the diffraction limit. As these techniques change the underlying physical setup of the imaging system, they are not relevant to the theoretical setting we consider, though we believe that placing the classical setting of Fraunhofer diffraction on a rigorous statistical footing can pave the way towards better understanding notions of resolution in these modern techniques.Here we very briefly describe some these techniques, deferring to the comprehensive overviews on the matter found in \cite{heintzmann2009subdiffraction,hell2007far,hell2009microscopy,huang2010breaking,ji2008advances,lauterbach2012finding,lippincott2009putting,maznev2017upholding,rice2007beyond,weisenburger2015light}.
The earliest attempts at going below the diffraction limit involved modifying the aperture, e.g. via \emph{apodization} as pioneered by Toraldo di Francia \cite{di1952super}. Among even more elementary approaches, an annular aperture can be used to distinguish a pair of points sources slightly better than a circular one, a fact that \cite{maznev2017upholding} notes was known even to Rayleigh.
Other approaches for circumventing the diffraction limit include near-field optics \cite{ash1972super,pohl1984optical,synge1928xxxviii}, TIRF \cite{axelrod1981cell,temple1981total}, confocal microscopy \cite{minsky1961microscopy}, two-lens techniques \cite{hell1992properties,hell1994confocal}, structured illumination \cite{gustafsson1999extended}, UV/X-ray/electron microscopy \cite{brand1997single,kirz1995soft,ruska1934fortschritte}.


Betzig, Hell, and Moerner were awarded the 2014 Nobel Prize in Chemistry for their pioneering work on super-resolution microscopy, which now includes technologies such as STED \cite{hell1994breaking,klar1999subdiffraction}, RESOLFT \cite{hell1995ground,hell2004strategy,bretschneider2007breaking}, PALM \cite{betzig2006imaging}, STORM \cite{rust2006sub}, and FPALM \cite{hess2006ultra}. These fundamentally break the diffraction limit by leveraging the ability to switch fluorescent markers between a bright and a dark state via photophysical effects like stimulated emission and ground-state depletion.
In light of such advancements, rigorously characterizing the resolving power of imaging systems remains a challenge of practical as much as theoretical interest. \cite{demmerle2015assessing} revisited what resolution means given these new technologies technologies and proposed approaches for comparing resolution between different super-resolution methods. \cite{horstmeyer2016standardizing} pushed back on some claims of super-resolution in nonfluorescent microscopy, advocating for the Siemens star as an imaging benchmark and for the adoption of certain standards when documenting such claims. Sheppard \cite{sheppard2017resolution} was similarly motivated to clarify such claims and calculates the images of various test object geometries and suggests ``these results can be used as a reference \dots to determine if super-resolution has indeed been attained.''

%% file: model.tex

\section{Physical Basis for Our Model}
\label{sec:physical}

In this paper we focus on the idealized setting of Fraunhofer diffraction of incoherent illumination by a circular aperture, originally studied in the pioneering work of Airy \cite{airy1835diffraction}. In this section, we first give a brief overview of this setting in Appendix~\ref{subsec:review}, deferring the details to any of a number of excellent expository texts on the subject \cite{kenyon2008light,hecht2015optics,goodman2005introduction,goodman2015statistical,jenkins1937fundamentals,fowles1989introduction}. Then in Appendix~\ref{subsec:model}, we demonstrate how our probabilistic model arises naturally from the preceeding setup. Finally, in Appendix~\ref{subsec:menagerie}, we catalogue the various resolution criteria that have appeared in the literature and instantiate them in our framework.

\subsection{A Review of Fraunhofer Diffraction}
\label{subsec:review}

Consider a scenario in which plane waves of monochromatic, incoherent light emanate from a far-away point source in the image plane, pass through a circular aperture, and form a diffraction pattern on a far-away observation plane. This is the standard setting of \emph{Fraunhofer diffraction}. As depicted in Figure~\ref{fig:fraunhofer}, the far-field assumption on the observation plane is captured in practice by placing a lens behind the aperture and placing the observation plane at the focal plane of the lens.

Under the Huygens-Fresnel-Kirchhoff theory, the aperture induces a diffraction pattern, a so-called \emph{Airy disk}, on the observation plane because the secondary spherical wavelets emanating from different points of the aperture are off by phase factors. Concretely, suppose the plane waves are parallel to the optical axis, and take a point $P$ on the observation plane at angular distance $\theta$ from the optical axis, and a point $\vec{u}$ on the circular aperture $A$, say of radius $r$. Letting $\vec{v}$ be the unit vector from the center of the aperture to $P$, we see that the propagation path of the wavelet from the center of the aperture to $P$ and that of the wavelet from $\vec{u}$ to $P$ differ in length by $\langle \vec{u},\vec{v}\rangle$, corresponding to a phase delay of $\frac{2\pi}{\lambda}\langle\vec{u},\vec{v}\rangle$ where $\lambda$ is the wavelength of light. So by integrating over the contributions to the amplitude of the electric field at $P$ by the points $\vec{u}$ in $A$, we conclude that the amplitude at $P$ is \begin{equation}
	E = E_0\int_{A} e^{2\pi i\cdot \langle \vec{u},\vec{v}\rangle/\lambda} d\vec{u},\label{eq:EP}
\end{equation} where $E_0$ is, up to phase factors, a constant capturing the contribution to the field per unit area of the aperture. In other words, the amplitude at $P$ is proportional to the 2D Fourier transform of the pupil function $F(\vec{u}) = \bone{\vec{u}\in A}$ at frequency $\vec{v}/\lambda$. This can be computed explicitly as \begin{equation}
	E = 2\pi r^2 E_0 \cdot \frac{J_1(\kappa r\sin\theta)}{\kappa r\sin\theta},
\end{equation} where $\kappa \triangleq \frac{2\pi}{\lambda}$ is the wavenumber of the light. In particular, the intensity $I(\theta)$ of the diffraction pattern at $P$ is the squared modulus of $E$. We conclude that \begin{equation}
	I(\theta) = I(0)\cdot \left(\frac{2J_1(\kappa r\sin\theta)}{\kappa r\sin\theta}\right)^2,\label{eq:intensity}
\end{equation} where $J_1(\cdot)$ is the Bessel function of the first kind. We will typically regard $I(\cdot)$ as a function $\R^2\to\R_{\ge 0}$ which takes in a point $(x,y)\in\R^2$ and outputs $I(\theta)$, where $\theta$ is the angular distance between $(x,y)$ and the optical axis. The function $I(x,y)$ is the so-called \emph{Airy point spread function}.

\begin{remark}
	In general, if the plane waves of the point source travel at an angle $\psi$ to the optical axis, they will be focused not at the focal point but at some other point on the observation plane at an angular distance of $\psi$ with respect to the optical axis. In this case the resulting Airy point spread function will be shifted to be centered at that point.
\label{remark:shift}
\end{remark}

\subsection{Photon Statistics and Our Model}
\label{subsec:model}

First suppose there is a single point source of light. In a sense which can be made rigorous via Feynman's path integral formalism (see e.g. Section 4.11 of \cite{hecht2015optics}), the intensity $I(x,y)$ of the diffraction pattern at a point $(x,y)$ on the observation plane is proportional to the (infinitesimal) probability of detecting a photon at $P$. That is, the point spread function $I(x,y)$ can be identified with a probability density \begin{equation}\rho(x,y) \triangleq \frac{1}{Z}\cdot I(x,y), \ \ \ \text{where} \ \ \ Z\triangleq \int_{\R^2}I(x,y) \; \d x \d y\end{equation} over the two-dimensional observation plane. Concretely, for any measurable subset $S$ of the observation plane, if one were to count photons arriving over time and compute the fraction that land inside the region $S$, this fraction would tend towards $\int_S \rho(x,y) \; \d x \d y$.

In the presence of $k$ \emph{incoherent} point sources of light, the absence of interference means that the contributions from each point source to the intensities of the resulting diffraction pattern simply add. In other words, if $I_1(\cdot),...,I_k(\cdot)$ are the corresponding point spread functions, which by Remark~\ref{remark:shift} are merely shifted versions of \eqref{eq:intensity}, the resulting probability density $\rho$ over the observation plane is simply proportional to $\sum^k_{i=1}I_i(\cdot)$.

For every $i\in[k]$, let $Z_i\triangleq \int_{\R^2}I_i(x,y) \; \d x \d y$ be the normalizing constant for the $i$-th density $\rho_i(\cdot) \triangleq \frac{1}{Z_i}I_i(\cdot)$. Let $\lambda_i \triangleq \frac{Z_i}{\sum^m_{j=1}Z_j}$. Then we see that \begin{equation}
	\rho(x,y) = \sum^m_{i=1}\lambda_i \cdot \rho_i(x,y).
\end{equation} In the jargon of statistics, this is an example of a \emph{mixture model}, i.e. a convex combination of structured distributions, and one can think of sampling from $\rho$ by first sampling an index $i\in[m]$ with probability $\lambda_i$ and then sampling a point $(x,y)$ in the observation plane according to the probability density associated to the $i$-th point source. This brings us to the generative model that we study in this work, the definition of which we restate here for the reader's convenience.

\airy*

We now describe briefly how the parameters in Definition~\ref{defn:model} translate to the setting of Fraunhofer diffraction by a circular aperture that we have outlined thus far. One should think of the spread parameter $\sigma$ as $(\kappa r)^{-1}$. As $\sigma$ in practice depends on known quantities pertaining to the underlying optical system, we assume henceforth that it is known \emph{a priori}. The norm of the argument in $A_{\sigma}(\norm{\vec{x} - \vmu_i}_2)$ corresponds to the quantity $\sin\theta$, where $\theta$ is the angle of displacement between the line from the center of the aperture to the center $\vmu_i$ of the $i$-th Airy disk, and the line between the center of the aperture and the point $\vx$ on the observation plane. Lastly, by Remark~\ref{remark:shift}, angular separation of $\psi$ between two point sources translates to angular separation of $\psi$ between the centers of their Airy disks on the observation plane. The parameters $\Delta$ and $\radius$ can thus be interpreted respectively as the minimum angular separation among the point sources, and the maximum angular distance of any of the point sources to the optical axis.

\subsection{Comparison to Semiclassical Detection Model}
\label{subsec:compare_models}

In this section we clarify the distinctions between the model we study and the semiclassical detection model. We begin by formally defining the latter.

\begin{defn}[Semiclassical Detection Model]
	For $D = 1,2$, let $S_1,...,S_m\subset\R^D$ be disjoint subsets corresponding to different regions of a photon detector, and suppose the detector receives some number $N'$ of photons, where $N'\sim\Poi(N)$. We observe photon counts $N_1,...,N_m$ corresponding to the number of photons that interact with each region of the detector, where for each $i\in[m]$, \begin{equation}N_i \triangleq N'_i + \gamma_i, \ \ \ N'_i\sim\Poi(\lambda_i\cdot N), \gamma_i \sim\N(0,\sigma^2),\end{equation} where $N'_1,...,N'_m,\gamma_1,...,\gamma_m$ are independent, $\gamma_i$ represents white detector noise\footnote{While these white noise terms $\{\gamma_i\}$ were not present in \cite{mandel1959fluctuations,helstrom1964detection}, they are considered in some later treatments of this model, so we include them here for completeness.}, and \begin{equation}
		\lambda_i \triangleq \int_{S_i} \rho(\vec{x}) \ \d x,
	\end{equation} where $\rho(\cdot)$, as in our model, is the idealized, normalized intensity profile of the optical signal.\label{defn:goodman}
\end{defn}

To see how this relates to our model, first consider the idealized case where $\sigma = 0$ and that the different regions $S_i$ of the detector form a partition of the entire ambient space. To get quantitative guarantees, existing works assume that each of these regions $S_i$ is, e.g., a segment or box of fixed length $\varsigma$. In this case, the semiclassical detection model is a special case of our model. Indeed, if one samples $\Poi(N)$ points from $\rho$ and moves each of them by distance $O(\varsigma)$ to the center of the region $S_i$ of the photon detector to which they respectively belong, this collection of $O(\varsigma)$-granular samples from $\rho$ is identical in information and distribution to a sample of photon counts $\{N_i\}$ from the semiclassical detection model.

Our model can also capture the case where the regions $S_i$ only partition a \emph{subset} of the ambient space $\R^D$. In this case, we only get access to samples from the density $\rho_{\text{trunc}}(\vec{x}) \propto \bone{\vec{x}\in \cup S_i}\cdot \rho(\vec{x})$, but this has known Fourier transform, given up to a universal multiplicative factor by the convolution of $\hat{\rho}$ with the indicator function of $\cup S_i$. So our techniques still apply in a straightforward fashion. In addition, by standard estimates on the tails of $J_1$, for $\cup S_i$ of radius polynomially large in the relevant parameters, with high probability none of the samples used by our algorithms will fall outside of $\cup S_i$ to begin with. For these reasons, we will not belabor this point in this work and will assume $\cup S_i = \R^D$ throughout.

Lastly, while our model does not incorporate white detector noise $\sigma$, we note that our algorithms can nevertheless handle the semiclassical detection model with $\sigma>0$: from a set of photon counts $N_1,...,N_m$, we can still estimate the Fourier transform of $\rho$ to accuracy depending polynomially on $N$ and inverse polynomially on $\sigma$ and the sizes of the detector regions, so our techniques based on the matrix pencil method still apply.

\subsection{A Menagerie of Diffraction Limits}
\label{subsec:menagerie}

In this section we give a precise characterization of the various limits that have appeared in the literature as candidates for the threshold at which resolution becomes impossible in diffraction-limited optical systems.

\paragraph{Abbe Limit}

The Abbe limit first arose in Abbe's studies \cite{abbe1873beitrage} of the following setup in microscopy: light illuminates an idealized object, namely an diffraction grating consisting of infinitely many closely spaced slits corresponding to the fine features of the object being imaged, and passes through the slits, behind which is an aperture stop placed in the back focal plane of the lens. Abbe observed that the angle at which the light gets diffracted by the slits increases as the grating gets finer, and he calculated the point at which the angle is too wide to enter the aperture. This threshold is now called the \emph{Abbe limit}, and in the modern language of Fourier optics, the Abbe limit corresponds to the point at which the Fourier transform of the corresponding point spread function (see Fact~\ref{fact:fourier}) vanishes. In the remainder of this section, we will refer to the Abbe limit as $\tau$.

\begin{remark}[Scaling and Numerical Aperture]
The argument $z$ in $A_{\sigma}(z)$ corresponds to the more familiar-looking quantity \begin{equation}
	z = \frac{2\pi}{\lambda}\cdot a\sin\theta \label{eq:zexp},
\end{equation} where $\lambda$ is the average wavelength of illumination, $a$ is the radius of the aperture, and $\theta$ is the angle of observation.

As noted above, $\widehat{A}_{\sigma}[\omega]$ is only supported on $\omega$ for which $\norm{\omega}\le \frac{1}{\pi}$. Equating this threshold $\frac{1}{\pi}$ with $1/z$, where $z$ is given by \eqref{eq:zexp}, and rearranging, we conclude that $\sin\theta = \frac{\lambda}{2a}$. We may write $\sin\theta$ as $q/R$ for $q$ the distance between the observation point and the optical axis and $R$ the distance between the observation point and the center of the aperture. It then follows that $q = \frac{\lambda R}{2a} \approx \frac{\lambda}{2\NA}$, where $\NA$ is the numerical aperture. This recovers the usual formulation of the \emph{Abbe limit}.
\end{remark}

In the literature on super-resolution microscopy, the Abbe limit is the definition of diffraction limit that is usually given. Indeed, Lauterbach notes in his survey \cite{lauterbach2012finding} that ``Abbe is perhaps the one who is most often cited for the notion that the resolution in microscopes would always be limited to half the wavelength of blue light.''

\paragraph{Rayleigh Criterion}

The Rayleigh criterion is the point at which the point spread function first vanishes. For $\sigma = 1$, this is precisely the smallest positive value of $r$ for which $J_1(r) = 0$, which can be numerically computed to be $r\approx 3.83 \approx 1.22\cdot \pi$. So for general $\sigma$, we conclude that the Rayleigh criterion is $\approx 1.22\tau$.

This is typically touted in standard references as the most common definition of resolution limit. Indeed, Weisenburger and Sandoghdar remark in their survey \cite{weisenburger2015light} that ``Although Abbe’s resolution criterion is more rigorous, a more commonly known formulation...is the Rayleigh criterion.'' Kenyon \cite{kenyon2008light} calls it the ``standard definition of the limit of the resolving power of a lens system.'' In his classic text, Hecht \cite{hecht2015optics} refers to it as the ``ideal theoretical angular resolution'' Rayleigh himself \cite{rayleigh1879xxxi} emphasized however that ``This rule is convenient on account of its simplicity and it is sufficiently accurate in view of the necessary uncertainty as to what exactly is meant by resolution.'' We refer to Appendix~\ref{app:quotes} for further quotations regarding the Rayleigh criterion.

\paragraph{Sparrow Criterion}

The Sparrow criterion, put forth in \cite{sparrow1916spectroscopic}, is the smallest $\Delta$ for which a superposition of two $\Delta$-separated Airy disks becomes unimodal. Numerically, this threshold is $\approx 0.94\tau$.

The Sparrow limit is often cited as the most mathematically rigorous resolution criteria (in den Dekker and van de Bos' survey \cite{den1997resolution}, they even call it ``the natural resolution limit that is due to diffraction...even a hypothetical perfect measurement instrument would not be able to detect a central dip in the composite intensity distribution, simply because there is no such dip anymore.''). It is less relevant in practical settings as it requires perfect knowledge of the functional form of the point spread function. Again, we refer to Appendix~\ref{app:quotes} for further quotations regarding the Sparrow criterion.

\paragraph{Houston Criterion}

The Houston criterion is twice the radius at which the value of the density is half of its value at zero, i.e. the ``full width at half maximum'' (FWHM). This threshold is $\approx 1.03\tau$.

This measure is one of the most popular in practice where one does not have fine-grained knowledge of the point spread function, in particular because it can apply even when the point spread function in question does not fall exactly to zero, either due to noise or aberrations in the lens. In \cite{demmerle2015assessing} where the authors explore alternative means of assessing resolution in light of new super-resolution microscopy technologies, they remark in their conclusion that ``the best approach to compare between techniques is still to perform the simple and robust fitting of a Gaussian to a sub-resolution object and then to extract the FWHM.''

\paragraph{Miscellaneous Additional Criteria}

The \emph{Buxton limit} is nearly the same as Houtson, except it is the FWHM for the \emph{amplitude} rather than the intensity, which yields a threshold of $\approx 1.46\tau$ \cite{buxton1937xli}. The \emph{Schuster criterion} is defined to be twice the Rayleigh limit \cite{schuster1904introduction}, that is, two Airy disks are separated only when their central bands are disjoint, which yields a threshold of $\approx 2.44\tau$. The \emph{Dawes limit}, which is $\approx 1.02\tau$, is a threshold proposed by Dawes \cite{dawes1867catalogue}; its definition is purely empirical, as it was derived by direct observation by Dawes.

%% file: quotes.tex

\section{Debate Over the Diffraction Limit: A Historical Overview}
\label{app:quotes}

In this section, we catalogue quotations from the literature relevant to the challenge of identifying the right resolution criterion, as well as to the need to take noise into account when formulating such definitions.



\subsection{Identifying a Criterion}

\noindent Since its introduction, the Rayleigh criterion has repeatedly been both touted as a practically helpful proxy by which to roughly assess the resolving power of diffraction-limited imaging systems, and characterized as somewhat arbitrary. 

Rayleigh himself in his original 1879 work \cite{rayleigh1879xxxi}:

\begin{quoting}
	\noindent ``This rule is convenient on account of its simplicity and it is sufficiently accurate in view of the necessary uncertainty as to what exactly is meant by resolution.''
\end{quoting}

Williams \cite[p.~79]{williams1950applications} in 1950:

\begin{quoting}
	\noindent ``Although with the development of registering microphotomers such as the Moll, dips much smaller than [the one exhibited by a superposition of two Airy disks at the Rayleigh limit] can be accurately measured, it is convenient for the purpose of comparison with gratings and echelons to keep to this standard.''
\end{quoting}

Born and Wolf \cite[p.~418]{born2013principles} in 1960:

\begin{quoting}
	\noindent ``The conventional theory of resolving power...is appropriate to direct visual observations. With other methods of detection (e.g. photometric) the presence of two objects of much smaller angular separation than indicated by Rayleigh’s criterion may often be revealed.''
\end{quoting}

Feynman \cite[Section 30-4]{feynman2011feynman} in his Lectures on Physics from 1964:

\begin{quoting}
	\noindent ``...it seems a little pedantic to put such precision into the resolving power formula. This is because Rayleigh’s criterion is a rough idea in the first place. It tells you where it begins to get very hard to tell whether the image was made by one or by two stars. Actually, if sufficiently careful measurements of the exact intensity distribution over the diffracted image spot can be made, the fact that two sources make the spot can be proved even if $\theta$ is less than $\lambda/L$.''
\end{quoting}

Hecht in his standard text \cite[p.431,492]{hecht2015optics} from 1987:

\begin{quoting}
	\noindent ``We can certainly do a bit better than this, but Rayleigh's criterion, however arbitrary, has the virtue of being particularly uncomplicated.''

	\noindent ``Lord Rayleigh's criterion for resolving two equal-irradiance overlapping slit images is well-accepted, even if somewhat arbitrarily in the present application.''
\end{quoting}

\noindent In fact, as early as 1904, Schuster \cite[p.~158]{schuster1904introduction} made the same point and on the same page advocated for an alternative criterion, corresponding to twice the separation posited by Rayleigh:

\begin{quoting}
	\noindent ``There is something arbitrary in (the Rayleigh criterion) as the dip in intensity necessary to indicate resolution is a physiological phenomenon, and there are other forms of spectroscopic investigation besides that of eye observation... It would therefore have been better not to have called a double line ``resolved'' until the two images stand so far apart, that no portion of the centeral band of one overlaps the central band of the other, as this is a condition which applies equally to all methods of observation. This would diminish to one half the at present recognized definition of resolving power.''
\end{quoting}

\noindent Ever since, the question of identifying the ``right'' notion of a resolution criterion has been periodically revisited in the literature.

Ramsay et al. \cite[p.~26]{ramsay1941criteria} in 1941, on this problem's theoretical and practical importance:

\begin{quoting}
	\noindent ``Before the theory itself can be developed in full, and applied to the assignment of numerical values, it is necessary to consider the persistently vexing problem of criteria for a limit of resolution.''
\end{quoting}

Three decades after Ramsay's work, Thompson \cite[p.~171]{thompson1969iv}:

\begin{quoting}
	\noindent ``The specification of the quality of an optical image is still a major problem in the field of image evaluation and assessment. This statement is true even when considering purely incoherent image formation.''
\end{quoting}

\noindent The Sparrow criterion is often regarded as the most mathematically rigorous resolution criterion. 

Sparrow \cite[p.~80]{sparrow1916spectroscopic} in 1916 on its mathematical and physiological justification:

\begin{quoting}
	\noindent ``It is obvious that the undulation condition should set an upper limit to the resolving power. The surprising fact is that this limit is \emph{apparently actually attained}, and that the doublet still appears resolved, the effect of contrast so intensifying the edges that the eye supplies a minimum where none exists. The effect is observable both in positives and in negatives, as well as by direct vision...My own observations on this point have been checked by a number of my friends and colleagues.''
\end{quoting}

In the survey of den Dekker and van den Bos \cite[p.~548]{den1997resolution} eighty years later:

\begin{quoting}
	\noindent ``Since Rayleigh's days, technical progress has provided us with more and more refined sensors. Therefore, when visual inspection is replaced by intensity measurement, the natural resolution limit that is due to diffraction would be [the Sparrow limit]...even a hypothetical perfect measurement instrument would not be able to detect a central dip in the composite intensity distribution, simply because there is no such dip anymore.''
\end{quoting}

\noindent In light of advancements in super-resolution microscopy, rigorously characterizing the resolving power of imaging systems remains as pressing a challenge as ever.

\begin{adjustwidth}{1.5em}{0pt}
In 2017, Demmerle et al. \cite{demmerle2015assessing} revisited what resolution means in light of these new technologies technologies and propose approaches for comparing resolution between different super-resolution methods. As they note in their introduction \cite[p.~3]{demmerle2015assessing}:
\end{adjustwidth}

\begin{quoting}
	\noindent ``The recent introduction of a range of commercial super-resolution instruments means that resolution has once again become a battleground between different microscope technologies and rival companies.''
\end{quoting}

\begin{adjustwidth}{1.5em}{0pt}
Notably, in the conclusion, they remark that a classical Houston criterion-style approach is still the best for comparing different methods \cite[p.~9]{demmerle2015assessing}.
\end{adjustwidth}

\begin{quoting}
	\noindent ``Given the above points, the best approach to compare between techniques is still to perform the simple and robust fitting of a Gaussian to a sub-resolution object and then to extract the FWHM.''
\end{quoting}

\subsection{The Importance of Noise}

\noindent An idea that has been repeated one way or another in the literature is that if one has perfect access to the \emph{exact} intensity profile of the diffraction image of two point sources, then one could brute-force search over the space of possible parameters to find a hypothesis that fits the point spread function arbitrarily well, thereby learning the positions of the point sources regardless of their separation. As such, for any notion of diffraction limit to have practical meaning, it must take into account factors like aberrations and measurement noise that preclude getting perfect access to the intensity profile.

This perspective was distilled emphatically by di Francia \cite[p.~497]{di1955resolving} in 1955:

\begin{quoting}
	\noindent ``Moreover it is only too obvious that from the mathematical standpoint, the image of two points, however close to one another, is different from that of one point. It is not at all absurd to assume that technical progress may provide us with more and more refined kinds of receptors, detecting the difference between the image of a single point and the image of two points located closer and closer to another. This means that at present there is only a \emph{practical} limit (if any) and not a \emph{theoretical} limit for two-point resolving power.''
\end{quoting}

\begin{adjustwidth}{1.5em}{0pt}
Contemporaneously, in discussions at the 1955 Meeting of the German Society of Applied Optics culminating in \cite[p.~459]{ronchi1961resolving}, Ronchi made the following distinction:
\end{adjustwidth}

\begin{quoting}
	\noindent ``Nowadays it seems imperative to differentiate three kinds of images, i.e., (1) the ethereal image, (2) the calculated image, and (3) the detected image.

	The nature of the ethereal image should be physical, but in reality it is only a hypothesis. It is said that the radiant flux emitted by the object...is concentrated and distributed in the so-called image by means of a number of processes. But actually this is only a hypothesis...attempts have been made to give a mathematical representation of the phenomenon, both geometrically and algebraically...The images which have been calculated in this way...should therefore be called calculated images.

	If we now consider the field of experience, we find the detected images. They are the figures either perceived by the eye when looking through the instrument, or obtained by means of a photosensitive emulsion, or through a photoelectric device.
\end{quoting}

den Dekker and van den Bos \cite[p.~547]{den1997resolution} in their 1997 survey:

\begin{quoting}
	\noindent ``Since Ronchi's paper, further research on resolution--- concerning detected images instead of calculated ones--- has shown that in the end, resolution is limited by systematic and random errors resulting in an inadequacy of the description fo the observations by the mathematical model chosen. This important conclusion was independently drawn by many researchers who were approaching the concept of resolution from different points of view.''
\end{quoting}

den Dekker and van den Bos summarize the state of affairs as follows \cite[p.~547]{den1997resolution}:

\begin{quoting}
	\noindent ``If calculated images were to exist, the known two-component model could be fitted numerically to the observations with respect to the component locations and amplitudes. Then the solutions for these locations and amplitudes would be exact, a perfect fit would result, and in spite of diffraction there would be no limit to resolution no matter how closely located the two point sources; this would mean that no limit to resolution for calculated images would exist. However, imaging systems constructed without any aberration or irregularity are an ideal that is never reached in practice....Therefore one should consider the resolution of detected images instead of calculated images.''
\end{quoting}

Goodman \cite[p.~326-7]{goodman2015statistical} in 2000:

\begin{quoting}
	\noindent ``...the question of when two closely spaced point sources are barely resolved is a complex one and lends itself to a variety of rather subjective answers...An alternative definition is the so-called Sparrow criterion...In fact, the ability to resolve two point sources depends fundamentally on the signal-to-noise ratio associated with the detected image intensity pattern, and for this reason criteria that do not take account of noise are subjective.''
\end{quoting}

Maznev and Wright \cite[p.~3]{maznev2017upholding} in 2016 on the earlier quote by Born and Wolf:

\begin{quoting}
	\noindent ``Indeed, if any number of photons is available for the measurement, there is no fundamental limit to how well one can resolve two point sources, since it is possible to make use of curve fitting to arbitrary precision (however, there are obvious practical limitations related to the finite measurement time and other factors such as imperfections in the optical system, atmospheric turbulence, etc).''
\end{quoting}

Demmerle et al. in the work mentioned in the previous section \cite[P.~9]{demmerle2015assessing}:

\begin{quoting}
	\noindent ``If one, a priori, knows that there are two point sources, then measuring their separation, and hence calculating the system's resolution is purely limited by Signal-to-Noise Ratio.''
\end{quoting}

\noindent A related point that has been made repeatedly in the literature is that the original setting in which Abbe introduced his diffraction limit should not be conflated with the setting of resolving two point sources of light.

\begin{adjustwidth}{1.5em}{0pt}
In the work of di Francia cited above \cite[p.~498]{di1955resolving}, he notes that the classic impossibility result for resolving a lattice of alternatively dark and bright points with separation below the Abbe limit says nothing about the impossibility of resolving a pair of points sources:
\end{adjustwidth}

\begin{quoting}
	``[The impossibility result at the Abbe limit] has often been given a wrong interpretation and it has too hastily been extended to the case of two points. The [Abbe limit] applies only \emph{when we want the available information uniformly distributed} over the whole image. Mathematics cannot set any lower limit for the distance of two resolvable points.''
\end{quoting}

\begin{adjustwidth}{1.5em}{0pt}
Indeed, he argues informally, by way of the Nyqist sampling theorem, that when there is a prior on the number of components in a superposition of Airy disks being upper bounded by a known constant, then in theory, there is no diffraction limit. Rather, he posits, it is the entropy of the prior that dictates the limits of resolution \cite[p.~498]{di1955resolving}:
\end{adjustwidth}

\begin{quoting}
	\noindent ``The fundamental question of how many independent data are contained in an image formed by a given optical instrument. This seems to be the modern substitute for the theory of resolving power.''
\end{quoting}

\begin{adjustwidth}{1.5em}{0pt}
Sheppard \cite[p.~597]{sheppard2017resolution} in 2017, sixty years after di Francia's work, clarifies again that the abovementioned impossibility result should not be misinterpreted as saying anything about the impossibility of resolving two point sources:
\end{adjustwidth}

\begin{quote}
	\noindent ``The Abbe resolution limit is a sharp limit to the imaging of a periodic object such as a grating. Super-resolution refers to overcoming this resolution limit. The Rayleigh resolution criterion refers to imaging of a two-point object. It is based on an arbitrary criterion, and does not define a sharp transition between structures being resolved or not resolved.''
\end{quote}

%% file: perturbation.tex

\section{Proof of Lemma~\ref{lem:jennrich}}
\label{app:jennrich}

\textsc{TensorResolve} (Algorithm~\ref{alg:tensor}) uses the standard subroutine given in Algorithm~\ref{alg:jennrich}. We remark that this algorithm appears to be deterministic unlike usual treatments of Jennrich's algorithm simply because we have absorbed the usual randomness of the choice of flattening into the construction of the tensor $\vec{T}$ on which \textsc{TensorResolve} calls \textsc{Jennrich}.

\begin{algorithm}\caption{\textsc{Jennrich}($\tilde{\vec{T}}$)}\label{alg:jennrich}
\begin{algorithmic}[1]
	\State \textbf{Input}: Tensor $\tilde{\vec{T}}\in\co^{m\times m\times 3}$ which is close to a rank-$k$ tensor $\vec{T}$ of the form \eqref{eq:Tdecomp}
	\State \textbf{Output}: $\hat{V}\in\co^{m\times k}$ close to $V$ up to column permutation (see Lemma~\ref{lem:jennrich})
		\State Compute the $k$-SVD $\hat{P}\hat{\Lambda}\hat{P}^{\dagger}$ of the flattening $\tilde{\vec{T}}(\Id,\Id,e_1)$.
		\State Define the whitened tensor $\hat{\vec{E}} = \tilde{\vec{T}}(\hat{P},\hat{P},\Id)$ and its flattenings $\hat{E}_i \triangleq \hat{\vec{E}}(\Id,\Id,e_i)$ for $i\in[2]$.
		\item Define $\hat{M} \triangleq \hat{E}_1\hat{E}^{-1}_2$.
		\item Form the matrix $\hat{U}$ whose columns are equal to the eigenvectors, scaled to have norm $\sqrt{m}$, for the $k$ eigenvalues of $\hat{M}$ that are largest in absolute value.
		\item Output $\hat{V} \triangleq \hat{P}\hat{U}$.
\end{algorithmic}
\end{algorithm}

We restate Lemma~\ref{lem:jennrich} here for the reader's convenience:

\tensor*

This proof closely follows that of \cite{huang2010breaking}, though we must make some modifications because the scaling of the frequencies $v^{(i)}$ for $i\in[3]$ defined in Step~\ref{step:vfreqs} of \textsc{TensorResolve} is different.

\begin{proof}
	We first define the noiseless versions of the objects $\hat{P}, \hat{\Lambda}, \hat{E}, \hat{E}_1, \hat{E}_2, \hat{M}, \hat{U}$ introduced in \textsc{Jennrich}. Note that for $i\in[2]$, \begin{equation}\vec{T}(\Id,\Id,e_i) = V D_i V^{\dagger}\label{eq:flat}\end{equation} for $D_i$ the diagonal matrix whose diagonal entries are given by $\{\lambda_j e^{-2\pi i \langle \mu_j, v^{(i)}\rangle}\}_{j\in[k]}$. Denote the $k$-SVD of $\vec{T}(\Id,\Id,e_1)$ by $P\Lambda P^{\dagger}$. Define the whitened tensor $\vec{E} \triangleq \vec{T}(P,P,\Id)$ and its flattenings $E_i = \vec{E}(\Id,\Id,e_i)$ for $i\in[2]$. Finally, define $U\triangleq P^{\dagger}V$ so that \begin{equation}
		\vec{E} = \sum^k_{j=1}\lambda_j U^j\otimes U^j\otimes W^j
	\end{equation} and $E_i = UD_iU^{\dagger}$ for $i\in[2]$. Note that $U$ also satisfies $M\triangleq E_1E^{-1}_2 = UDU^{\dagger}$ for diagonal matrix $D \triangleq D_1D^{-1}_2$, and for every $j\in[k]$, $\norm{U^j}_2 = \norm{V^j}_2 = \sqrt{m}$, so $U$ is indeed the noiseless analogue of $\hat{U}$.

	For any $j\in[k]$, we have that \begin{equation}
		D_{j,j} = e^{-2\pi i \langle \mu_j, v^{(1)} - v^{(2)}\rangle}.
		\label{eq:Djj}
	\end{equation} Define $\Delta_D\triangleq \min_{j\neq j'}|D_{j,j} - D_{j',j'}|$.

	For every $j,j'\in[k]$, by triangle inequality and the fact that $V^j = PU^j$ and $\hat{V}^ = \hat{P}\hat{U}^j$, we have \begin{equation}
		\norm{\hat{V}^j - V^{j'}}_2 \le \norm{\hat{P} - P}_2\norm{\hat{U}^j}_2 + \norm{P}_2\norm{\hat{U}^j - U^{j'}}_2 \le \sqrt{m}\norm{\hat{P} - P}_2 + \norm{\hat{U}^j - U^{j'}}_2.
		\label{eq:main_jennrich}
	\end{equation}
	We proceed to upper bound $\norm{\hat{P} - P}_2$ and $\norm{\hat{U}^j - U^{j'}}_2$.

	\begin{lem}
		$\norm{\hat{P} - P}_2 \le \frac{\eta'\sqrt{m}}{\lambda_{\min}\sigma_{\min}(V)^2}$.
		\label{lem:wedin}
	\end{lem}

	\begin{proof}
		By Wedin's theorem, \begin{equation}\norm{\hat{P} - P}_2 \le \frac{\norm{\tilde{\vec{T}}(\Id,\Id,e_1) - \vec{T}(\Id,\Id,e_1)}_2}{\sigma_{\min}(\vec{T}(\Id,\Id,e_1))}.\end{equation} By \eqref{eq:flat}, $\sigma_{\min}(\vec{T}(\Id,\Id,e_1)) \ge \lambda_{\min}\sigma_{\min}(V)^2$. Additionally, $\norm{\tilde{\vec{T}}(\Id,\Id,e_1) - \vec{T}(\Id,\Id,e_1)}_F \le \eta'\sqrt{m}$, from which the claim follows.
	\end{proof}

	\begin{lem}
		If $\norm{M - \hat{M}}_2 \le \frac{\Delta_D}{2\sqrt{k}\kappa(U)}$, then the eigenvalues of $\hat{M}$ are distinct, and there exists a permutation $\tau$ for which \begin{equation}
			\norm{\hat{U}^j - U^{\tau(j)}}_2 \le \frac{3m\norm{M - \hat{M}}_2}{\Delta_D\sigma_{\min}(U)} \ \ \ \forall \ j\in[k].
		\end{equation}
		\label{lem:Udiff}
	\end{lem}

	\begin{proof}
		Consider the matrix $U^{-1}\hat{M}U = D - U^{-1}(M - \hat{M})U$. Because $\norm{U^{-1}(M - \hat{M})U}_2 \le \Delta_D/2\sqrt{k}$ by assumption, we conclude by Gershgorin's that the eigenvalues of $U^{-1}\hat{M}U$, and thus of $\hat{M}$, are distinct and each lies within $\Delta_D/2$ of a unique eigenvalue of $M$. Let $\tau$ be the permutation matching eigenvalues $\{\hat{\beta}_j\}$ of $\hat{M}$ to eigenvalues $\{\beta_j\}$ of $M$ which are closest, and without loss of generality let $\tau$ be the identity permutation.

		For fixed $j\in[k]$, let $\{c_{j'}\}$ be coefficients for which $\hat{U}^j = \sum c_{j'}U^j$ and $\sum_{j'} c_{j'}^2 = 1$. Note that we have \begin{equation}
			\hat{\lambda}_j\sum_{j'} c_{j'}U^{j'} = \hat{\lambda}_j\hat{U}^j = \hat{M}\hat{U}^j = \sum_{j'}\lambda_{j'}c_{j'}U^{j'} + (M - \hat{M})\hat{U}^j,
		\end{equation} so $\{c_{j'}\}$ is the solution to the linear system \begin{equation}
			\sum_{j'}c_{j'}\cdot (\hat{\lambda}_j - \lambda_{j'})U^{j'} = (M - \hat{M})\hat{U}^j.
		\end{equation} Recalling that $\norm{U^j}_2 = \sqrt{m}$ and that $\sum c_{j'}^2 = 1$, we get that \begin{align}
			\norm{\hat{U}^j - U^j}^2_2 &= \sum_{j'\neq j}c^2_{j'}\norm{U^{j'}}^2_2 + (c_j - 1)^2\norm{U^j}^2_2 \le 2m\sum_{j'\neq j}c^2_{j'} \\
			&\le \frac{8m\norm{U^{-1}(M - \hat{M})\hat{U}^j}^2_2}{\Delta^2_D} \le \frac{8m^2\norm{M - \hat{M}}^2_2}{\Delta_D^2\sigma_{\min}(U)^2}.
		\end{align}
	\end{proof}

	Finally, we must estimate $\norm{M - \hat{M}}^2_2$ in the bound in Lemma~\ref{lem:Udiff}:

	\begin{lem}
		If $\eta' \le \frac{\lambda^2_{\min}\sigma_{\min}(V)^2}{6\sqrt{m}\kappa(V)^2}$, then $\norm{M - \hat{M}}_2 \le \frac{9\eta'\sqrt{m}\kappa(V)^2}{\lambda^2_{\min}\sigma_{\min}(V)^2}$.
		\label{lem:Mdiff}
	\end{lem}

	\begin{proof}
		Define $Z_i \triangleq \hat{E}_i - E_i$ for $i\in[2]$ so by taking Schur complements \begin{equation}
			M  - \hat{M} = E_1E^{-1}_2 - (E_1 + Z_1)(E_2+Z_2)^{-1} = M Z_2(\Id + E^{-1}_2Z_2)^{-1}E^{-1}_2 + Z_1E^{-1}_2 \triangleq MH+G
			\label{eq:defGH}
		\end{equation} Note that \begin{equation}
			\sigma_{\max}(H) \le \frac{\norm{Z_2}_2}{\sigma_{\min}(E_2) - \norm{Z_2}_2} \le \frac{\norm{Z_2}_2}{\lambda_{\min}\sigma_{\min}(U)^2  - \norm{Z_2}_2} , \ \ \
			\sigma_{\max}(G) \le \frac{\sigma_{\max}(Z_1)}{\sigma_{\min}(E_2)} \le \frac{\norm{Z_1}_2}{\lambda_{\min}\sigma_{\min}(U)^2} 
		\end{equation} and furthermore for either $i\in[2]$, because $Z_i = \hat{P}^{\dagger}\tilde{\vec{T}}(\Id,\Id,e_i)\hat{P} - P^{\dagger}\vec{T}(\Id,\Id,e_i)P$, \begin{align}
			\norm{Z_i}_2 &\le \norm{P}_2 \norm{\vec{T}(\Id,\Id,e_i)}_2 \norm{P - \hat{P}}_2 + \norm{\hat{P}}_2\norm{\vec{T}(\Id,\Id,e_i)}_2\norm{\hat{P} - P}_2 + \norm{\hat{P}}^2_2 \norm{\tilde{\vec{T}}(\Id,\Id,e_i)}_2 \\
			&\le 2\frac{\eta'\sqrt{m}}{\lambda_{\min}\sigma_{\min}(V)^2}\cdot \lambda_{\max}\sigma_{\max}(V)^2 + \lambda_{\max}\sigma_{\max}(V)^2 \le \frac{3\eta'\sqrt{m}\kappa(V)^2}{\lambda_{\min}}\label{eq:Zbound}
		\end{align}
		Because $\sigma_{\min}(U)^2 = \sigma_{\min}(V)^2$, by the bound on $\eta'$ in the hypothesis, $\sigma_{\max}(H) \le \frac{2\norm{Z_2}_2}{\lambda_{\min}\sigma_{\min}(V)^2}$. Finally, noting that $\norm{M}_2 \le \sigma_{\max}(D) = 1$, we conclude the proof from \eqref{eq:defGH} and \eqref{eq:Zbound}.
	\end{proof}

	It remains to bound $\Delta_D$.

	\begin{lem}
		For any $\delta>0$, with probability at least $1 - \delta$, $\Delta_D \ge O\left(\frac{(c - \factor)\delta'\Delta}{k^2}\right)$.
		\label{lem:sepD}
	\end{lem}

	\begin{proof}
		Using the elementary inequality $|e^{-2\pi i x} - 1| \le 2\pi|x|$ for any $x\in\R$, we conclude that $|D_{j,j} - D_{j',j'}| \le |e^{-2\pi i \langle \mu_j - \mu_{j'},v^{(1)} - v^{(2)}\rangle} - 1| \le 2\pi|\langle \mu_j - \mu_{j'}, v^{(1)} - v^{(2)}\rangle|$. By standard anti-concentration, for any $j\neq j'$ and $\delta'>0$ we have that $|\langle \mu_j - \mu_{j'}, v^{(1)} - v^{(2)}\rangle| \le O(\delta'\norm{\mu_j - \mu_{j'}}_2 \cdot \norm{v^{(1)} - v^{(2)}}_2)$ with probability at most $\delta'$. The proof follows by taking $\delta' = \delta/k^2$, union bounding, and recalling the definition of $v^{(1)},v^{(2)}$ in \textsc{TensorResolve}.
	\end{proof}

	Combining \eqref{eq:main_jennrich} and Lemmas~\ref{lem:wedin}, \ref{lem:Udiff}, \ref{lem:Mdiff}, \ref{lem:sepD}, there exists a permutation $\tau$ for which \begin{equation}
		\norm{\hat{V}^j - V^{\tau(j)}}_2 \le \frac{\eta' m}{\lambda_{\min}\sigma_{\min}(V)^2} + \frac{27\eta' m^{3/2}\kappa(V)^2}{\Delta_D \lambda^2_{\min}\sigma_{\min}(V)^3} \le O\left(\frac{k^2\eta' m^{3/2}\kappa(V)^5}{(c - \factor)\delta\Delta\lambda^2_{\min}}\right) \ \ \forall j\in[k].
	\end{equation} We conclude that for the permutation matrix $\Pi$ corresponding to $\tau$, $\norm{\hat{V} - V\Pi}_F \le \sqrt{k}\max_{j\in[k]}\norm{\hat{V}^j - V^{\tau}(j)}_2 \le O\left(\frac{k^{5/2}\eta' m^{3/2}\kappa(V)^5}{(c - \factor)\delta\Delta\lambda^2_{\min}}\right)$ as claimed.
\end{proof}

%% file: fig_gen.tex

\section{Generating Figure~\ref{fig:tv}}
\label{sec:fig_details}

Here we elaborate on how Figure~\ref{fig:tv} was generated. While Theorem~\ref{thm:mainlowerbound} yields an explicit construction which rigorously demonstrates the phase transition at the diffraction limit, empirically we found that this phase transition was even more pronounced when we slightly modified the construction. Specifically, we empirically evaluated the following instance: for even $k$, separation $\Delta>0$, and $1\le i\le k$, let $\vmu_i = (a_i,0)$ and let $\vmu'_i = (b_i,0)$ for $a_i \triangleq \frac{\Delta}{2}\cdot\left(2i-\frac{k+3}{2}\right)$ and $b_i \triangleq \frac{\Delta}{2}\cdot\left(2i-\frac{k+1}{2}\right)$, and take $\{\lambda_i\}$ and $\{\lambda'_i\}$ to be the unique solution to the affine system \begin{equation}
	\sum^{\lceil k/2\rceil}_{i=1} \lambda_i = 1 \, \text{and} \, \sum^{\lfloor k/2\rfloor}_{i=1} \lambda'_i = 1
	\qquad \qquad 
	\sum^{\lceil k/2\rceil}_{i=1} \lambda_i a_i^{\ell} = \sum^{\lfloor k/2\rfloor}_{i=1} \lambda'_i b_i^{\ell} \ \ \forall \ \ 0\le \ell < k-1.
\end{equation} These are the weights for which the superposition of point masses at $\{\vmu_i\}$ with weights $\{\lambda_i\}$ matches the superposition of point masses at $\{\vmu'_i\}$ with weights $\{\lambda'_i\}$ on all moments of degree at most $k - 2$. While moment-matching does not directly translate to any kind of statistical lower bound, it is often the starting point for many such lower bounds in the distribution learning literature \cite{moitra2010settling,diakonikolas2017statistical,hardt2015tight,kearns1998efficient}. The ``carefully chosen pair of superpositions'' referenced in the caption of Figure~\ref{fig:tv} refers to this moment-matching construction. Henceforth refer to these two superpositions, both of which are $\Delta$-separated superpositions of $k/2$ Airy disks, as $\calD_0(\Delta,k)$ and $\calD_1(\Delta,k)$ respectively. We will omit the parenthetical $\Delta,k$ when the context is clear.

Unfortunately, there is no closed form for the expression for $d_{\text{TV}}(\calD_0,\calD_1)$. Instead, we estimated this via numerical integration. Direct evaluation of the integral $\int_{\R^2}\left|\calD_0(\vx) - \calD_1(\vx)\right| \, dx$ poses issues because of the heavy tails of the Airy point spread function. To tame these tails, we used a carefully chosen proposal measure $\mu$ in order to rewrite $d_{\text{TV}}(\calD_0,\calD_1)$ as $\int_{\R^2}\left|\frac{\calD_0(\vx)}{\mu(\vx)} - \frac{\calD_1(\vx)}{\mu(\vx)}\right| \, d\mu$. Because of the heavy tails, we needed to use a similarly heavy-tailed proposal distribution, so we took $\mu$ to be the convolution of the superposition of point masses at $\{\vmu_i\}\cup\{\vmu'_i\}$ having weights $\{\lambda_i\}\cup\{\lambda'_i\}$ with the following kernel $P(\cdot)$. To sample from the density over $\R^2$ correpsonding to $P$, with probability 1/2 sample a radius $r$ uniformly from $[0,1]$ and output a random vector in $\R^2$ of norm $r$, and with the remaining probability 1/2, sample from the Pareto distribution with parameter $2/3$ over $[1,\infty]$ and output a random vector of norm $r$. The motivation for $P$ and in particular for the parameter $2/3$ is that it is a rough approximation to the tail behavior of the radial density $\frac{J_1(r)^2}{r}$ defining the Airy point spread function, which by Theorem~\ref{thm:besselbound} decays roughly as $r^{-5/3}$.

To generate the curves in Figure~\ref{fig:tv}, for each $k\in[2,4,6,12,20,30,42,56,72,90]$ and each $\Delta\in[-2,-1.92,-1.84,...,1.84,1.92,2]$, we simply estimated the corresponding $d_{\text{TV}}(\calD_0,\calD_1)$ by sampling $10$ million points $\vx$ from $\mu$ and computing the empirical mean of the quantity $\left|\frac{\calD_0(\vx)}{\mu(\vx)} - \frac{\calD_1(\vx)}{\mu(\vx)}\right|$.

We have made the code for Figure~\ref{fig:tv} available at \url{https://github.com/secanth/airy/}.